%% file: main.tex
\documentclass[11pt]{article}
\title{Conditionally Tight Algorithms for Maximum $k$-Coverage and Partial $k$-Dominating Set via Arity-Reducing Hypercuts}
\input{imports}
\usepackage[shortcuts]{extdash}
\usepackage{xcolor}
\date{}
\author{Nick Fischer \and Marvin K\"unnemann \and Mirza Redzic}

\newenvironment{subproof}[1][\proofname]{%
  \begin{proof}[#1]%
}{%
  \end{proof}%
}
\begin{document}

\maketitle
\begin{abstract}
    \noindent
    We revisit the classic Maximum $k$-Coverage problem: Determine the largest number $t$ of elements that can be covered by choosing $k$ sets from a given family $\mathcal{F} = \{S_1,\dots, S_n\}$ of a size-$u$ universe. A notable special case is Partial $k$-Dominating Set, in which the task is to select $k$ nodes in a graph to maximize the number $t$ of dominated vertices. Extensive research has established strong hardness results for various aspects of Maximum $k$-Coverage, such as tight inapproximability results, $W[2]$-hardness, and a conditionally tight worst-case running time of $n^{k\pm o(1)}$ (for sufficiently large constant values of $k$). In this paper we ask: (1) Can this time bound be improved for small $t$, at least for Partial $k$-Dominating Set, ideally to time~$t^{k\pm O(1)}$? (2)~More ambitiously, can we even determine the best-possible running time of Maximum $k$-Coverage with respect to the perhaps most natural parameters: the universe size $u$, the maximum set size~$s$, and the maximum frequency $f$?

	We successfully resolve both questions. (1) We give an algorithm that solves Partial $k$\=/Dominating Set in time~\smash{$O(nt + t^{\frac{2\omega}{3} k+O(1)})$} if $\omega \ge 2.25$ and time~\smash{$O(nt+ t^{\frac{3}{2} k+O(1)})$} if $\omega \le 2.25$, where $\omega \le 2.372$ is the matrix multiplication exponent. From this we derive a time bound that is \emph{conditionally optimal}---regardless of $\omega$---based on the well-established $k$-clique and 3\=/uniform hyperclique hypotheses from fine-grained complexity. We also obtain matching upper and lower bounds for \emph{sparse} graphs. To address (2) we design an algorithm for Maximum $k$-Coverage running in time
	\[
        \min \left\{ (f\cdot \min\{\sqrt[3]{u}, \sqrt{s}\})^k + \min\{n,f\cdot \min\{\sqrt{u}, s\}\}^{k\omega/3}, n^k\right\}
        \cdot g(k)n^{\pm O(1)},
    \]
    and, surprisingly, further show that this complicated time bound is also \emph{conditionally optimal}.
    
    Our results are mainly based on a new algorithmic technique which we call \emph{arity-reducing hypercuts}. We are optimistic that this technique finds further applications in the future, perhaps for other problems with (currently) incomparable conditional lower bounds from $k$-clique detection in graphs and hypergraphs.

\end{abstract}

\thispagestyle{empty}
\clearpage
\setcounter{page}{1}

\input{chapters/introduction}

\section{Preliminaries}\label{sec:prelims}
For a positive integer $n$, let $[n]$ denote the set $\{1,\dots, n\}$. 
If $S$ is an $n$-element set and $0\leq k\leq n$ is an integer, then by $\binom{S}{k}$ we denote the set of all $k$-element subsets of $S$.

Let $\omega<2.371552$ \cite{WilliamsXXZ24} denote the optimal exponent of multiplying two $n\times n$ matrices and $\MM(a,b,c)$ the time required to multiply two rectangular matrices of dimensions $a\times b$ and $b\times c$.
Note that if $\omega=2$, $\MM(a,b,c) \leq (ab + ac + bc)^{1+o(1)}$.

For a graph $G$ and a vertex $v\in V(G)$, the \emph{neighborhood} of $v$, denoted $N(v)$ is the set of vertices adjacent to $v$. 
The \emph{closed neighborhood} of $v$, denoted $N[v]$ is defined as $N[v]:=N(v)\cup\{v\}$.
For the subset $S\subseteq V(G)$, we denote $N(S) := \bigcup_{v\in S} N(v)$ (respectively $N[S]:=\bigcup_{v\in S} N[v]$).
The \emph{degree} of $v$ denotes the size of its neighborhood ($\deg(v) = |N(v)|$).
For a (hyper)graph $G = (V,E)$ and a set $S\subseteq V$, we denote by $G[S]$ the subgraph of $G$ induced on $S$.

Given a graph $G$ with $n$ vertices, the \emph{$k$-Clique Detection} problem is to decide if $G$ contains a clique of size $k$. 
If $k$ is divisible by $3$, we can solve this problem by constructing a graph $T$, such that each vertex in $T$ corresponds to a clique in $G$ of size $k/3$ and adding an edge between the vertex corresponding to a clique $C_1$ and the vertex corresponding to a clique $C_2$ if and only if every vertex in $C_1$ is adjacent to every vertex in $C_2$. Now detecting a $k$-clique in $G$ is equivalent to detecting a triangle in $T$. A simple matrix multiplication algorithm detects triangles in graphs with $N$ vertices in time $N^{\omega}$. This yields an algorithm that solves $k$-Clique Detection in $n^{\omega k /3}$. \cite{nevsetvril1985complexity,EisenbrandG04}
Remarkably, no significant improvement over this simple algorithm has been made for decades. This led to the following hardness assumption (see e.g. \cite{AbboudBVW18}).
\begin{hypothesis}[$k$-Clique Hypothesis]
    For no $\varepsilon>0$ and $k\geq 3$ is there an algorithm solving $k$-Clique Detection in time $O(n^{k\omega/3 - \varepsilon})$.
\end{hypothesis}
The \emph{$h$-Uniform $k$-Hyperclique Detection} problem is given an $h$-uniform hypergraph $G$ with $n$ vertices to decide if $G$ contains a hyperclique of size $k$.
It turns out for $h\geq 3$, the similar matrix multiplication techniques fail to detect hypercliques of size $k$ in $h$-uniform hypergraphs.
In fact, no algorithm running in $\bigO(n^{k-\varepsilon})$ is known to be able to solve this problem, and it has been shown that any such algorithm would imply significant improvement for other problems that are conjectured to be hard, like Max-$h$-SAT and, for $h\ge 4$, Max-Weight $k$-Clique (\cite{LincolnVWW18})
This led to the following hardness assumption.
\begin{hypothesis}[$h$-Uniform $k$-Hyperclique Hypothesis]
    For no $\varepsilon>0, h\geq 3, k\geq h+1$ is there an algorithm solving $h$-Uniform $k$-Hyperclique Detection in time $\bigO(n^{k - \varepsilon})$.
\end{hypothesis}

Next, we introduce \emph{$k$-Orthogonal Vectors} and related problems. For vectors~\makebox{$v_1, \dots, v_k \in \{0, 1\}^d$} we write $v_1 \odot \dots \odot v_k = \sum_{y=1}^d x_1[y] \cdot \dots \cdot x_k[y]$ (i.e., a generalization of the inner product to $k$ vectors). In the \emph{$k$-Orthogonal Vectors} ($k$-OV) problem the goal is to decide whether for given size-$n$ sets $V_1, \dots, V_k \subseteq \{0, 1\}^d$, there are vectors $v_1 \in V_1, \dots, v_k \in V_k$ with $v_1 \odot \dots \odot v_k = 0$. The $k$-OV hypothesis postulates that there is no $k$-OV algorithm running in time $\bigO(n^{k-\varepsilon})$ (for any $\varepsilon > 0$) in the regime where $d = \omega(\log n)$.

More generally, consider vector sets $V_1,\dots, V_k\subseteq \{0,1\}^d$ where each coordinate $y\in [d]$ is associated to $h$ (pairwise distinct) \emph{active indices} $i_1,\dots, i_h\in [k]$.
We denote by $a(i_1,\dots, i_h)$ the set of coordinates $y$ such that the active indices associated to $y$ are precisely $i_1,\dots, i_h$.
Note that $a(i_1,\dots, i_h) = a(\pi(i_1),\dots, \pi(i_h))$ for any permutation $\pi$ (i.e. $a$ is a symmetric function).
For any $2\leq r\leq h$ and pairwise distinct $i_1,\dots, i_r\in [k]$, let $v_{i_1}\in V_{i_1},\dots, v_{i_r}\in V_{i_r}$ be vectors and write $v_{i_1}\odot \dots \odot v_{i_r}$ to denote the number of coordinates $y\in a(i_1,\dots, i_r, i_{r+1},\dots, i_h)$ (for any valid choice of $i_{r+1},\dots, i_h$) such that $v_{i_1}[y]=\dots = v_{i_r}[y] = 1$. Furthermore, for $r\geq h$, let $v_1\cdot\dots \cdot v_r = \sum_{i_1<\dots<i_h \in [r]}v_{i_1}\odot \dots \odot v_{i_h}$.
In this notation, for fixed constants $k,h$, the \emph{$(k,h)$-Orthogonal Vectors} ($(k,h)$-OV) problem is to decide if there are vectors $v_1\in V_1,\dots, v_k\in V_k$ satisfying $v_1\cdot \dots \cdot v_k = 0$.
The \emph{$(k,h)$-Maximal Inner Product} ($(k,h)$-maxIP) and \emph{$(k,h)$-Minimal Inner Product} ($(k,h)$-minIP) are the natural optimization versions of $(k,h)$-OV. Namely, $(k,h)$-maxIP ($(k,h)$-minIP) is to find the vectors $v_1\in V_1,\dots, v_k\in V_k$ such that the value $v_1\cdot \dots \cdot v_k$ is maximal (minimal). 

\section{Algorithms for Max $k$-Cover and Partial $k$-Dominating Set}
\label{sec:algorithms}
For the rest of the paper, we will consider the graph-theoretic formulation of Max-$k$-Cover: 
Let $G = (X\cup Y,E)$ be a bipartite graph with $|X| = n$, $|Y| = u$ such that any vertex in $X$ has degree at most $\Delta_s$, and any vertex in $Y$ has degree at most $\Delta_f$.
Then Max-$k$-Cover problem is to maximize, over all $x_1,\dots, x_k\in X$, the value $|N(x_1)\cup\dots \cup N(x_k)|$.

We proceed to construct the algorithms for Max-$k$-Set Cover and Partial $k$-Dominating Set. 
In particular, we prove the following two theorems.
\begin{theorem}[Max-$k$-Set Cover Algorithm]\label{theorem:max-k-cover-alg}
    Given a bipartite graph $G = (X\cup Y,E)$ with $|X| = n$, $|Y| = u$, such that the maximum degree of any vertex $x\in X$ is $\Delta_s$ and the maximum degree of any vertex $y\in Y$ is $\Delta_f$, we can find a collection of $k$ vertices $x_1,\dots, x_k\in X$ that maximize the value $|N(x_1)\cup \dots \cup N(x_k)|$ in time
    \begin{equation*}
        \bigO\Big(\big((\min\{n,\Delta_f\cdot \min\{u^{1/3}, \sqrt{\Delta_s}\}\})^k + (\min\{n,\Delta_f\cdot \min\{\sqrt{u}, \Delta_s\}\})^{k\omega/3}\big)\cdot (\Delta_s\Delta_f)^{c}\Big)
    \end{equation*}
    where $c$ is a constant independent of $k$.
\end{theorem}

The above theorem provides the upper bound for our main result of Theorem~\ref{thm:main2}. We will begin by giving an algorithm for the special case Partial $k$-Dominating Set, i.e., establish the upper bound of Theorem~\ref{thm:main1}.

\begin{theorem}[Partial $k$-Dominating Set Algorithm]\label{theorem:partial-dom-algorithm}
    Given a graph $G$ with $n$ vertices and maximum degree $\Delta$, we can compute the maximum value $|N[v_1]\cup \dots \cup N[v_k]|$ over all $v_1,\dots, v_k\in V(G)$ in time \[ \bigO\big((\min\{n,\Delta^{3/2}\}^k + \min\{n,\Delta^{2}\}^{k\omega/3})\cdot \Delta^{c}\big),\] where $c$ is a constant independent on $k$.
\end{theorem}

Before giving the overview of our algorithm, we first prove a lemma that allows us to bound the size of $X$ in terms of the parameters $\Delta_s$ and $\Delta_f$.
\begin{lemma}\label{lemma:bounding-number-of-sets}
    Let $G = (X\cup Y,E)$ be as above and let $H\subseteq X$ be a set consisting of the first $\min\{k\Delta_f\cdot \Delta_s, n\}$ many vertices in $X$ sorted in the decreasing order of degrees. 
    Then 
    \[
    \max_{x_1,\dots, x_k\in X} |N(x_1) \cup \dots \cup N(x_k)| = \max_{x_1,\dots, x_k\in H} |N(x_1) \cup \dots \cup N(x_k)|.
    \]
\end{lemma}
\begin{proof}
    If $k\Delta_f\cdot \Delta_s\geq  n$, the claim trivially holds. Hence assume that $k\Delta_f\cdot \Delta_s < n$ and let $x_1,\dots, x_k$ be vertices contained in $X$ with $x_1\not\in H$. It is sufficient to argue that we can replace $x_i$ by some vertex $x'_1\in H$ such that $|N(x_1) \cup N(x_2) \cup \dots \cup N(x_k)| \leq |N(x'_1) \cup N(x_2) \cup \dots \cup N(x_k)|$.
    Let $H'\subseteq H$ consist of all vertices in  $H$ that share a common neighbor with at least one $x_j$ for $j\geq 2$. 
    Observe that for each $x_j$ there are at most $\Delta_s\cdot\Delta_f$ many vertices $x\in X$ that share a common neighbor with~$x_j$, and thus $H'$ consists of at most $(k-1)\Delta_s\cdot\Delta_f$ many vertices.
    In particular, the set $H\setminus H'$ is non-empty.
    We claim that by setting $x_1'$ to be any vertex from $H\setminus H'$, we get the desired inequality. 
    Indeed, since $x_1\in X\setminus H$ and $x'_1\in H$, by construction of $H$ it holds that $\deg(x_1)\leq \deg(x'_1)$, and since $x'_i\not\in H'$ it shares no common neighbors with any $x_j$ (for $j\geq 2$) and we obtain the following chain of inequalities:
    \[
    \Big|\bigcup_{1\leq i \leq k}N(x_i)\Big| \leq \deg(x_1) + \Big|\bigcup_{2\leq i \leq k}N(x_i)\Big| \leq \deg(x'_1) + \Big|\bigcup_{2\leq i \leq k}N(x_i)\Big| = \Big|N(x'_1)\cup \bigcup_{2\leq i \leq k}N(x_i)\Big| \qedhere
    \]
\end{proof}

We follow the approach outlined in Section~\ref{sec:technical-overview}: We consider a hypergraph representation $\mathcal{H}$ of a given instance, and explore whether there exists an optimal solution $S$ consisting of $k$ vertices, such that the subhypergraph induced on $S$ admits a ''balanced'' cut.
We then proceed to make a win-win argument: If such an optimal solution exists, we argue that we can reduce this instance to an instance of Max-Weight-Triangle problem with small weights, which we can then solve efficiently. 
Otherwise, we argue that the obstructions to such a cut have a nice structure, so that we can enumerate them efficiently, and after guessing only constantly many such obstructions, we obtain a self reduction to a smaller instance that has an optimal solution with a ''balanced'' cut in the hypergraph representation. 
In the following paragraphs, we proceed to formally introduce the notation and terminology required to construct such an algorithm.

\paragraph{Arity-Reducing Hypercuts and Bundles} 
Let $H$ be a hypergraph.
An \emph{arity-reducing hypercut} of $H$ is a partition of vertices of $H$ into $d$ sets $S_1,\dots, S_d$ such that there is no edge crossing all $d$ parts (i.e. for any $d$-tuple of vertices $s_1\in S_1,\dots, s_d\in S_d$ it holds that $\{s_1,\dots, s_d\}\not\in E(H)$).
For the rest of the paper, we will work only with the $3$-uniform hypergraphs, hence when we talk about arity-reducing hypercuts, we will always assume that $d=3$.
If $H$ has $k$ vertices, we say that an arity-reducing hypercut $S_1, S_2, S_3$ is \emph{balanced} if $|S_1| = \lceil k/3\rceil$, $|S_2| = \lceil (k-1)/3\rceil$, $S_3 = \lfloor k/3\rfloor$.
For a given bipartite graph $G=(X\cup Y,E)$, let the \emph{hypergraph representation of $G$} denoted $\mathcal{H}(G)$ be the $3$-uniform hypergraph constructed as follows. Let $V(\mathcal{H}(G)) = X$, and for each triple of vertices $x_1, x_2, x_3\in X$, let $\{x_1,x_2,x_3\}\in E(\mathcal{H}(G))$ if and only if there is a vertex $y\in Y$ such that $y\in N(x_1)\cap N(x_2)\cap N(x_3)$.
We say that a set $S\subseteq X$ \emph{admits} a balanced arity-reducing hypercut, if there exists a balanced arity-reducing hypercut in the subhypergraph of $\mathcal{H}(G)$ induced by $S$. 
Let $c$ be a non-negative integer and define \emph{$c$-bundle} recursively as follows. A $0$-bundle is a set consisting of a single vertex in $X$. Given a $c$-bundle $B$, let $u,v\in X\setminus B$ be such that for some vertex $b\in B$ it holds that $\{u,v,b\}\in E(\mathcal{H}(G))$. Then $B\cup \{u,v\}$ is a $(c+1)$-bundle.
Clearly each $c$-bundle contains exactly $1+2c$ many vertices.
A $c$-bundle $B$ is called \emph{maximal} if for no pair $u,v$ in $X$, $B\cup \{x,y\}$ forms a $(c+1)$-bundle.
We now proceed to show that large bundles can be thought of as obstructions for balanced arity-reducing hypercuts and in particular that for any bipartite graph $G = (X\cup Y, E)$,
it suffices to remove at most two bundles from $X$, such that the remaining part admits a balanced arity-reducing hypercut in $\mathcal{H}(G)$.
\begin{lemma}\label{lemma:bundles-and-cuts}
    Let $G=(X\cup Y,E)$ be a bipartite graph. 
    There exist sets $D_1, D_2$ such that the following conditions are satisfied:
    \begin{itemize}
        \item Each $D_i$ is either empty or a $c$-bundle for some $c\ge 0$.
        \item The set $X\setminus(D_1\cup D_2)$ admits a balanced arity-reducing hypercut.
    \end{itemize}
\end{lemma}
\begin{proof}
    Let $\mathcal B_\alpha$ be the set containing all maximal $\alpha$-bundles for any $\alpha\geq 0$ and define $\mathcal B := \bigcup_{\alpha \geq 0} \mathcal B_\alpha$. 
    Note that any distinct pair $B_1,B_2\in \mathcal B$ is disjoint, since if there is a vertex $v\in B_1\cap B_2$, then $B_1\cup B_2$ is also a bundle, hence by maximality $B_1 = B_2$.
    Moreover, $\mathcal B$ forms a partition of $X$.
    Let~$\ell$ denote the value $|\mathcal B|$ and let $B_1,\dots, B_\ell$ be the bundles from $\mathcal B$ ordered by size in the increasing order.
    We build a partition of bundles $B_1,\dots, B_{\ell-2}$ into three sets greedily as follows. 
    Initially set $S_1 = S_2 = S_3 = \emptyset$.
    Iterate over $B_1,\dots, B_{\ell-2}$ and in each iteration $B_i$ put the set $B_i$ in the set~$S_j$ that is the smallest so far.
    Note that after the iteration $B_i$ it holds that $||S_p|-|S_q||\leq |B_i|$ for each $p,q\in[3]$.
    In particular, after the last iteration, we have $||S_p|-|S_{q}||\leq |B_{\ell-2}|$.
    Without loss of generality assume that after last iteration it holds $|S_1|\geq |S_2|\geq |S_3|$. 
    We claim we can find a subset $B_{\ell-1}'$ of $B_{\ell-1}$ such that the following conditions hold:
    \begin{itemize}
        \item $|S_1|-|S_2\cup B_{\ell-1}'| \in \{0,1\}$
        \item $B_{\ell-1}'':= B_{\ell-1} \setminus B_{\ell-1}'$ is either a $c$-bundle (for some $c\geq 0$), or empty.
    \end{itemize}
    Indeed, since $||S_1|-|S_{2}||\leq |B_{\ell-2}|\leq B_{\ell-1}$, we can construct the desired bundle $B_{\ell-1}''$ recursively as follows. 
    If $|S_1|-|S_2\cup B_{\ell-1}'| \in \{0,1\}$, do nothing. 
    Otherwise, find a pair of vertices $\{x,y\}\subseteq B_{\ell-1}'$ (unless it is the first iteration, then take a single vertex $x$) such that $B_{\ell-1}''\cup \{x,y\}$ is a bundle. 
    Set $B_{\ell-1}'' = B_{\ell-1}''\cup \{x,y\}$ and $B_{\ell-1}'= B_{\ell-1}'\setminus \{x,y\}$ and recurse.
    This procedure clearly terminates with the partition of $B_{\ell-1}$ satisfying both desired conditions.
    Now repeat the same construction with $S_1$, $S_3$ and $B_\ell$ to get the sets $B_{\ell}', B_{\ell}''$ such that $|S_1|-|S_3\cup B_{\ell}'| \in \{0,1\}$ and $B_{\ell}''$ is a bundle.
    It is now straightforward to verify that $S_1, S_2\cup B_{\ell-1}', S_3\cup B_{\ell}'$ (up to reordering) forms a balanced arity-reducing hypercut of the subhypergraph of $\mathcal{H}(G)$ induced on $X\setminus (B_{\ell-1}''\cup B_{\ell}'')$, where $B_{\ell-1}'',B_{\ell}''$ are either empty, or form bundles.
\end{proof}
For any optimal solution $S$, by applying the lemma above on the hypergraph $\mathcal{H}(G)[S]$, we can conclude that we only need to guess two bundles, such that the remaining part of the solution admits a balanced arity-reducing hypercut.
It only remains to argue that we can exploit the structure of these bundles to be able to efficiently enumerate them.
\begin{lemma}\label{lemma:listing-bundles}
    Let $G=(X\cup Y,E)$ be a bipartite graph with $|X| = n$ and $\max_{x\in X} \deg(x) = \Delta_s$, $\max_{y\in Y}\deg(y) = \Delta_f$. For any fixed constant $c$, we can list all $c$-bundles in $X$ in time $\bigO(n\Delta_s^c\Delta_f^{2c})$.
\end{lemma}
\begin{proof}
    We prove this by induction on $c$. 
    For $c=0$, this bound is trivial. 
    Assume now that the asserted bound holds for some $c\geq 0$.
    Any $(c+1)$-bundle is by definition obtained by extending some $c$-bundle, by introducing two new vertices to it.
    By induction hypothesis, there are $\bigO(n\Delta_s^c\Delta_f^{2c})$ $c$-bundles to choose from to extend.
    For any $c$-bundle $S$, there are at most $(1+2c)\Delta_s = \bigO(\Delta_s)$ vertices in $Y$ that are adjacent to at least one vertex $v\in S$. 
    We can now simply iterate over all of those vertices and in time $\bigO(\binom{\Delta_f}{2}) = \bigO(\Delta_f^2)$ choose any pair of vertices that can be added to $S$ to form a $(c+1)$-bundle.
    Observe that this procedure lists all $(c+1)$-bundles in time $\bigO(n\Delta_s^{c+1}\Delta_f^{2+2c})$ as desired.
\end{proof}
These tools allow us to now construct a surprisingly simple algorithm for Partial $k$-Dominating Set and in particular to prove \autoref{theorem:partial-dom-algorithm}.
\subsection{Algorithm for Partial $k$-Dominating Set}
\begin{algorithm}
\begin{algorithmic}[1]
    \Procedure{Partial-DS}{$X,Y,k$}
    \State $t'\gets 0$
    \For{bundle $S_1$ with $0 \leq |S_1| \leq k$}
        \For{bundle $S_2$ with $0 \leq |S_1| + |S_2| \leq k$ and $S_1\cap S_2 = \emptyset$}
            \State $k' \gets k - |S_1| - |S_2|$
            \State $V_1 \gets \binom{X}{\lceil k'/3\rceil}$,  $V_2 \gets \binom{X}{\lceil (k'-1)/3\rceil}$, $V_3 \gets \binom{X}{\lfloor k'/3\rfloor}$
            \State $t'\gets \max \{t', |N(S_1)| + |N(S_2)| +\text{max-weight-triangle}(V_1, V_2, V_3, Y-N(S_1)-N(S_2))\}$
        \EndFor
    \EndFor 
    \Return $t'$
    \EndProcedure
\end{algorithmic}
\caption{}
\label{alg:partial-dom-alg}
\end{algorithm}
We now proceed to give a simple and efficient algorithm for the Partial $k$-Dominating Set problem. In fact, by copying the vertex set of a given graph twice (setting $X = Y = V(G)$) and adding edges between $x\in X$, $y\in Y$ if and only if $y$ is dominated by $x$ in $G$, we reduce Partial $k$-Dominating Set to a special case of the Max-$k$-Cover problem, where $|X|= |Y| = n$ and $\Delta_f = \Delta_s = \Delta$.
We will focus on solving this slightly more general problem and any algorithm for this problem running in time $T(n,\Delta)$ clearly implies the existence of an algorithm solving Partial $k$-Dominating Set in $\bigO(T(n,\Delta))$.

Let $G = (X\cup Y,E)$ be a bipartite graph, and for fixed positive integers $k_1,k_2,k_3$ let $V_1 \subseteq \binom{X}{k_1}, V_2\subseteq \binom{X}{k_2}, V_3\subseteq \binom{X}{k_3}$ (i.e. each vertex in $V_i$ corresponds to a subset of $X$ of size $k_i$). 
To make a distinction between vertices in $G$ and those in $V_i$, we will call vertices in $V_i$ \emph{nodes} and denote them by using the overline notation (i.e. $\overline{v}\in V_i$). Furthermore, for simplicity, if a vertex $x\in X$ is contained in the set corresponding to the node $\overline{v}\in V_i$, we will denote this by $x\in \overline{v}$. This allows us to use the set-theoretic notions (union, intersection, etc.) directly on the nodes.
Let max-weight-triangle$(V_1, V_2, V_3, Y)$ be an algorithm that constructs a double-weighted complete tripartite graph~$T$ with parts $V_1, V_2, V_3$, where the weight of each node $\overline v$ is equal to the number of vertices in $Y$ that are adjacent to some vertex in the set corresponding to $\overline v$. That is $w(\overline v) := |\bigcup_{x\in \overline v}N(x)|$. Moreover, the weight of each edge $\{\overline u,\overline v\}$ is equal to the negative number of vertices in $Y$ that are adjacent to both a vertex in the set corresponding to $\overline u$ \emph{and} a vertex in the set corresponding to~$\overline v$. That is $w(\overline u,\overline v) = -|\bigcup_{x\in \overline u, x'\in \overline v}N(x)\cap N(x')|$. After constructing this graph, the algorithm then finds a triangle in $T$ with the maximum weight in this graph. 
By applying the inclusion-exclusion principle, we can show that the weight of any triangle $\overline u,\overline v,\overline z$ in $T$ is bounded by the number of vertices in $Y$ that are adjacent to at least one vertex in $\overline u\cup \overline v\cup \overline z$.
\begin{lemma}\label{lemma:triangle-weight-bound-inclusion-exclusion}
    Given a bipartite graph $G = (X\cup Y,E)$, let $T$ be a weighted complete tripartite graph constructed as above. Then for any triangle $\overline u ,\overline v,\overline z$ in $T$ it holds that
    \begin{enumerate}[label=(\roman*)]
        \item \label{lemma:triangle-weight-bound-inclusion-exclusion:item1} $w(\overline u,\overline v,\overline z) \leq \big|\bigcup_{x\in \overline u\cup \overline v\cup \overline z} N(x)\big|.$
        \item \label{lemma:triangle-weight-bound-inclusion-exclusion:item2} If $\overline u,\overline v,\overline z$ is an arity-reducing hypercut of $\mathcal H(G)[\overline u\cup \overline v\cup \overline z]$, then $w(\overline u,\overline v,\overline z) = \big|\bigcup_{x\in \overline u\cup \overline v\cup \overline z} N(x)\big|$ 
    \end{enumerate}
\end{lemma}
\begin{proof}
    By the principle of inclusion-exclusion, we have 
    \begin{align*}
        \big|\bigcup_{x\in \overline u\cup \overline v\cup \overline z} N(x)\big| & = |\bigcup_{x\in \overline u} N(x)\big| + |\bigcup_{x\in \overline v} N(x)\big| + |\bigcup_{x\in \overline z} N(x)\big| \\ & -\big|\bigcup_{x\in \overline u, x'\in \overline v} N(x)\cap N(x')\big|-\big|\bigcup_{x\in \overline u, x'\in \overline z} N(x)\cap N(x')\big|-\big|\bigcup_{x\in \overline v, x'\in \overline z} N(x)\cap N(x')\big| \\
        & +\big|\bigcup_{x\in \overline u, x'\in \overline v, x''\in \overline z} N(x)\cap N(x')\cap N(x'')\big|.
    \end{align*}
    Recall the definition of the weight function $w$ on $T$, plugging it in the above equation, we have 
    \begin{equation*}
    \begin{split}
        \big|\bigcup_{x\in \overline u\cup \overline v\cup \overline z} N(x)\big| & = w(\overline u) + w(\overline v) + w(\overline z)\\
        & + w(\overline u,\overline v) + w(\overline u,\overline z) + w(\overline v,\overline z)\\
        & +\big|\bigcup_{x\in \overline u, x'\in \overline v, x''\in \overline z} N(x)\cap N(x')\cap N(x'')\big|.
    \end{split}
    \end{equation*}
    By definition the weight of the triangle in $T$ is just $w(\overline u,\overline v,\overline z):= w(\overline u) + w(\overline v) + w(\overline z) + w(\overline u,\overline v) + w(\overline u,\overline z) + w(\overline v,\overline z)$, and since $\big|\bigcup_{x\in \overline u, x'\in \overline v, x''\in \overline z} N(x)\cap N(x')\cap N(x'')\big|\geq 0$, we get the desired inequality for \ref{lemma:triangle-weight-bound-inclusion-exclusion:item1}.
    Moreover, if $\overline u,\overline v,\overline z$ is an arity-reducing hypercut, then by definition for each $x\in \overline u, x'\in \overline v, x''\in \overline z$ it holds that $N(x)\cap N(x')\cap N(x'') = \emptyset$, and the equality in \ref{lemma:triangle-weight-bound-inclusion-exclusion:item2} follows.
\end{proof}
So far, we have used max-weight-triangle algorithm as a black box. However, it is known that finding a maximal weight triangle in a double-weighted graph with $n$ vertices can be done in the running time of computing $(\min,+)$-product of two $n\times n$ matrices (see e.g. \cite{VassilevskaW06}).
Moreover, Zwick proved in \cite{Zwick98} that if all the vertex and edge weights are integer in range $[-M,\dots, M]$, then this running time is at most $\bigO(Mn^{\omega})$.
We now adapt this argument to obtain a desired algorithm for our setting.
\begin{lemma}\label{lemma:max-weight-triangle-algorithm}
    Let $G = (X\cup Y,E)$ be a bipartite graph with each vertex in $X$ having degree at most $\Delta_s$. For a fixed positive integer $k$, let $k_1 = \lceil k/3\rceil,k_2 = \lceil (k-1)/3\rceil, k_3 = \lfloor k/3 \rfloor$, and $V_1 = \binom{X}{k_1},V_2 = \binom{X}{k_2},V_3 = \binom{X}{k_3}$. There is an algorithm max-weight-triangle that:
    \begin{itemize}
        \item Constructs a double-weighted complete tripartite graph $T = (V_1, V_2, V_3, E')$ with the weight of each node $\overline v$ assigned as $w(\overline v) := |\bigcup_{x\in \overline v}\{y\in Y \mid x,y\in E\}|$ and the weight of each edge $\{\overline u,\overline v\}$ defined as $w(\overline u,\overline v) = -|\bigcup_{x\in \overline u, x'\in \overline v}\{y\in Y \mid x,y\in E \land x',y\in E\}|$.
        \item Finds a triangle with maximum weight in $T$.
        \item Runs in time bounded by $\bigO(|X|^{\omega \lceil k/3\rceil)}\Delta_s)$.
    \end{itemize}
\end{lemma}
\begin{proof}
     We can first compute the weights of any node and pair of nodes in time $\bigO(|X|^{k_1+k_2}\Delta_s) \leq \bigO(|X|^{2\lceil k/3\rceil}\Delta_s)$, and since $\omega\geq 2$, the construction of the graph $T$ takes at most $\bigO(|X|^{\omega\lceil k/3\rceil)}\Delta_s)$ as desired.
    Now note that any node and edge weight is an integer in $\{-\Delta_s \cdot k/3,\dots, \Delta_s\cdot k/3\}$, which can be written as $\{-\bigO(\Delta_s),\dots, \bigO(\Delta_s)\}$ (assuming that $k$ is a fixed constant).
    Thus we may apply \cite{VassilevskaW06,Zwick98} to detect a maximum weight triangle in time $\bigO(|X|^{k_1\omega}\Delta_s) = \bigO(|X|^{\lceil k/3\rceil\omega}\Delta_s)$.\footnote{We could even improve this running time slightly by employing the fastest rectangular matrix multiplication algorithm, but for our use-case, this running time suffices.}
\end{proof}
We are now ready to construct an algorithm for Partial $k$-Dominating Set.
\begin{proof}[Proof (of \autoref{theorem:partial-dom-algorithm})]
    Given a graph $G = (V,E)$ with $n$ many vertices and maximal degree $\Delta$, construct the bichromatic instance $G':=(X\cup Y,E')$ by copying the vertex set twice and adding the edges naturally, as discussed above.
    Sort the vertices in $X$ by degree and remove the first $\max\{0, |X|-\Delta^2\}$ many, so that the size of $X$ remains bounded by $\min\{n,\Delta^2\}$. By \autoref{lemma:bounding-number-of-sets} this yields an equivalent instance.
    Finally, run Algorithm \ref{alg:partial-dom-alg} and report the output of this algorithm as the output of the original instance.
    \begin{claim}
        The algorithm described above yields the correct solution for the Partial $k$-Dominating Set problem.
    \end{claim}
    \begin{proof}
        Let $\textsc{opt} := \max_{x_1,\dots, x_k\in X}|N(x_1)\cup\dots \cup N(x_k)|$ and let $t$ be the value returned by the algorithm above.
        We prove that $\textsc{opt} = t$.
        By Lemma \ref{lemma:triangle-weight-bound-inclusion-exclusion}, it follows that $t\leq \textsc{opt}$.
        Fix vertices $x_1,\dots, x_k\in X$ such that $|N(x_1)\cup \dots \cup N(x_k)| = \textsc{opt}$. 
        By Lemma \ref{lemma:bundles-and-cuts}, there exist bundles $S_1, S_2$ such that the 
        set $\{x_1,\dots, x_k\}\setminus (S_1\cup S_2)$ admits a balanced arity-reducing hypercut. 
        In particular, this means that there is a partition of set $\{x_1,\dots, x_k\}\setminus (S_1\cup S_2)$ into three sets $A_1\in \binom{X}{\lceil k'/3\rceil}$, $A_2\in \binom{X}{\lceil (k'-1)/3\rceil}$, $A_3\in \binom{X}{\lfloor k'/3\rfloor}$, where $k'$ is the size of the set $\{x_1,\dots, x_k\}\setminus (S_1\cup S_2)$, such that no edge of $\mathcal{H}(G')$ crosses all three sets $A_1,A_2,A_3$.
        By construction, the algorithm will guess bundles $S_1, S_2$ at some iteration and since the remaining vertices admit a balanced arity-reducing hypercut, by Item \ref{lemma:triangle-weight-bound-inclusion-exclusion:item2} of Lemma \ref{lemma:triangle-weight-bound-inclusion-exclusion}, we obtain the inequality $t\geq \textsc{opt}$, as desired.
    \end{proof}
    \begin{claim}
        The algorithm above runs in time $\bigO\big((\Delta^{3/2k} + \min\{n, \Delta^{2}\}^{\omega/3k})\cdot \Delta^{c}\big)$, where $c$ is a constant independent on $k$.
    \end{claim}
    \begin{proof}
        Construction of graph $G'$ takes only linear time. 
        After removing the light vertices, we are left with only $\min\{n,\Delta^2\}$ many vertices. 
        By Lemma \ref{lemma:listing-bundles}, for any fixed $0\leq c\leq (k-1)/2$, there are at most $\bigO(|X|\Delta^{3c})$ many $c$-bundles in $G'$ and each $c$-bundle has size $1+2c$. 
        \marvin{the following assumption is irrelevant for time analysis: }\mirza{This is just to give the intuition where the terms are coming from. But I can also remove this?} \marvin{well, you repeat the correctness argument. But it suffices to analyze the running time of Algorithm 1 as given.}Assume that we are given a promise that an optimal solution admits a balanced arity-reducing hypercut after removing two bundles $S_1, S_2$, such that $S_1$ is a $c_1$-bundle and $S_2$ is a $c_2$-bundle. Then we have a total of $\bigO(\min\{n,\Delta^2\}^2\Delta^{3c_1+3c_2})$ iterations to guess $S_1, S_2$, and the remaining number of vertices to guess with the max-weight-triangle algorithm is $k':= k-(2c_1 + 2c_2 + 2)$.
        Thus, by iterating over all possible values of $c_1,c_2$, this gives us the following running time (for simplicity, we will drop the constant factors below). 
        \begin{align*}
        T_k(n,\Delta) \leq &\; \min\{n,\Delta^2\}^{\omega\lceil k/3\rceil}\Delta & \text{(both $S_1,S_2$ are empty)}\\ 
        &+ \sum_{0\leq c_1\leq (k-1)/2} \Big( \min\{n,\Delta^2\}^{\omega\lceil(k-2c_1-1)/3\rceil+1}\Delta^{3c_1+1} & \text{($S_1$ is a $c_1$-bundle, $S_2$ is empty)}
        \\&+ \sum_{0\leq c_2\leq (k-2c_1-2)/2} \min\{n,\Delta^2\}^{\omega\lceil(k-2c_1-2c_2-2)/3\rceil+2}\Delta^{3c_1+3c_2+1} \Big) & \text{($S_1, S_2$ are $c_1, c_2$-bundles, resp.)}
        \end{align*}
        Write $n = \Delta^{\gamma}$ and notice that we can assume without loss of generality that $1\leq \gamma\leq 2$. 
        In particular, $\gamma \geq 1$ is a trivial lower bound, since no vertex can have degree larger than the number of vertices in graph, and if $\gamma > 2$, we can apply \autoref{lemma:bounding-number-of-sets} to remove all but $\Delta^2$ many vertices and hence $\gamma \leq 2$. 
        We can now plug in $\min\{n,\Delta^2\}\leq \Delta^\gamma$ in the time complexity analysis above and compute
        \begin{align*}
        T_k(n,\Delta) \leq & \; \Delta^{\gamma \omega\lceil k/3\rceil+1} & \\ 
        &+ \sum_{0\leq c_1\leq (k-1)/2} \Big(\Delta^{\gamma \omega\lceil (k-2c_1-1)/3\rceil+\gamma}\Delta^{3c_1+1} &
        \\& + \sum_{0\leq c_2\leq (k-2c_1-2)/2} \Delta^{\gamma \omega\lceil(k-2c_1-2c_2-2)/3\rceil+2\gamma}\Delta^{3c_1+3c_2+1} \Big) &
        \end{align*}
        We now prove that if $\omega\leq \frac{9}{2\gamma}$, then this running time is at most $\bigO(\Delta^{3/2k + 5})$. In order to do that, we bound each of the summands by this value, and since both sums range only over $f(k) = \bigO(1)$ many values $c_1, c_2$, we get the desired.
        We start with the simplest summand first.
        \begin{align*}
            \Delta^{\gamma \omega\lceil k/3\rceil+1} &\leq \Delta^{\gamma \omega (k/3 + 2/3) +1} & (\lceil k/3\rceil \leq \tfrac{k+2}{3}) \\
            & \leq \Delta^{3/2(k + 2)+1} & (\omega \leq \tfrac{9}{2\gamma}) \\
            & = \Delta^{3/2k +4}
        \end{align*}
        Moving on to the second part of our expression.
        \begin{align*}
            \Delta^{\gamma \omega\lceil(k-2c_1-1)/3\rceil+\gamma}\Delta^{3c_1+1} & \leq \Delta^{\gamma \omega(k-2c_1+1)/3+\gamma}\Delta^{3c_1+1} & (\text{using }\lceil\tfrac{x}{3}\rceil \leq \tfrac{x+2}{3}) \\
            & \leq \Delta^{3/2(k-2c_1 + 1)+\gamma}\Delta^{3c_1+1} & (\omega\leq \tfrac{9}{2\gamma}) \\ 
            & \leq \Delta^{3/2(k+1)+3} & (\gamma \leq 2)\\
            & < \Delta^{3/2k+5}
        \end{align*}
        Using the exact same approach, we can bound the last part of the expression as well.
        \begin{align*}
            \Delta^{\gamma\omega\lceil(k-2c_1-2c_2-2)/3\rceil+2\gamma}\Delta^{3c_1+3c_2+1} \leq \Delta^{3/2k+5}
        \end{align*}
        This implies that if $\omega\leq \tfrac{9}{2\gamma}$, then $T_k(n,\Delta) \leq f(k)\cdot \Delta^{3/2k+5} = \bigO(\Delta^{3/2k+5})$.

        We remark that for the case when $\Delta$ is small in comparison to the number of vertices in the input graph (i.e. $\gamma \geq 2$), we have $\frac{9}{2\gamma} \leq 2.25$.
        It is still open whether $\omega\leq 2.25$, and at the time of submitting this paper, the best known upper bound is $\omega < 2.371552$ \cite{WilliamsXXZ24}. 
        In fact, with this value of $\omega$, for all values of $\gamma$ larger than roughly $1.898$, the computation above yields no bound on the running time of our algorithm.
        Thus, it still remains to consider the time complexity of our algorithm in terms of state-of-the-art value of $\omega$.
        To this end, we prove that if $\omega>\frac{9}{2\gamma}$, then each summand is bounded by $\Delta^{\gamma \omega k/3 + 5}$. Clearly, the first summand $\Delta^{\gamma \omega \lceil k/3\rceil + 1}$ satisfies this bound, so we can move on to the remaining terms.
        \begin{align*}
            \Delta^{\gamma\omega(k-2c_1+1)/3+\gamma}\Delta^{3c_1+1} & \leq \Delta^{\gamma\omega/3(k+1)}\Delta^{3c_1 -(2\gamma\omega/3) c_1 + \gamma+1} & \text{(rearranging the terms)} \\
            &< \Delta^{2\omega/3(k+1)}\Delta^{\gamma + 1} & (\omega>\tfrac{9}{2\gamma}) \\
            &< \Delta^{2\omega k/3 + 5} & (\omega<3, \gamma \leq 2)
        \end{align*} 
        Finally, we apply the similar approach for the last summand.
        \begin{align*}
            \Delta^{\gamma \omega\lceil(k-2c_1-2c_2-2)/3\rceil+2\gamma}\Delta^{3c_1+3c_2+1} & \leq  \Delta^{\gamma\omega(k-2c_1-2c_2)/3+2\gamma}\Delta^{3c_1+3c_2+1} \\
            & = \Delta^{\gamma\omega k/3+2\gamma}\Delta^{3c_1+3c_2+1 - \gamma\omega(2c_1+2c_2)/3} & \text{(rearranging)}\\
            & < \Delta^{\gamma \omega k/3+5} & (\omega>\tfrac{9}{2\gamma}, \gamma \leq 2) \\
        \end{align*}
        We can thus conclude that the running time of our algorithm is bounded by either $f(k)\Delta^{\gamma\omega k/3+5}$, or by $f(k)\Delta^{3/2k+5}$, depending on the value of $\omega$.
        Hence, regardless of value of $\omega$, we can bound \[T_k(n,\Delta)\leq \bigO\big((\Delta^{\gamma\omega k/3} + \Delta^{3/2k})\Delta^5\big) = 
        \bigO\big((\min\{n,\Delta^2\}^{\omega k/3} + \Delta^{3/2k})\Delta^5\big),\] 
        as desired.
    \end{proof}
    By combining the previous two claims, whenever $\gamma>3/2$, by running the algorithm above we can achieve the running time $\bigO\big((\Delta^{3k/2} + \min\{n,\Delta^{2}\}^{k\omega/3})\cdot \Delta^{5}\big)$. On the other hand, the case $\gamma\leq 3/2$ can be handled by the baseline $\bigO(n^{k+o(1)})$ algorithm, and thus, we can solve the Partial $k$-Dominating Set in time 
    \[
    T_{k}(n,\Delta)\leq \bigO\big((\min\{n,\Delta^{3/2}\}^k + \min\{n,\Delta^{2}\}^{k\omega/3})\cdot \Delta^{5}\big),
    \]
    as desired.
\end{proof}
\subsection{Extending to Max $k$-Cover}\label{sec:max-k-cover}
In the last subsection, we have shown that using Algorithm \ref{alg:partial-dom-alg}, combined with some preprocessing, we can efficiently solve the special case of Max-$k$-Cover, where $\Delta_f = \Delta_s = \Delta$ and $|X|=|Y|$. In fact, we proceed to show that running the same algorithm on any instance of Max-$k$-Cover yields a correct solution.
\begin{lemma}\label{lemma:max-k-cover-algorithm-large-universe}
    Given a bipartite graph $G = (X\cup Y,E)$ with $|X| = n$ and every $x\in X$ satisfying $\deg(x)\leq \Delta_s$ and respectively every $y\in Y$ satisfying $\deg(y)\leq \Delta_f$, we can find the $k$ vertices $x_1,\dots, x_k\in X$ maximizing the value $|N(x_1)\cup\dots\cup N(x_k)|$ in time $\bigO\Big(((\Delta_f\cdot \sqrt{\Delta_s})^k + (\min\{n,\Delta_s\Delta_f\})^{k\omega/3})(\Delta_f\Delta_s)^5\Big)$.
\end{lemma}
\begin{proof}
    By \autoref{lemma:bounding-number-of-sets} we can construct an equivalent instance satisfying $|X| \leq \min\{n, \Delta_s\Delta_f\}$.
    Now run Algorithm \ref{alg:partial-dom-alg} and report the output of this algorithm as the output of the original instance as above.
    The proof of correctness is the same as in the proof of \autoref{theorem:partial-dom-algorithm}.
    It only remains to verify the running time.
    We start by noting that by definition $n\geq \Delta_f$ (otherwise there is a $y\in Y$ such that $|N(y)|>n = |X|$, contradiction). 
    Hence we can write $n = \Delta_f\Delta_s^\gamma$ for $0\leq \gamma\leq 1$. By the argument above, it is clear that $\gamma$ cannot be negative. On the other hand, if $\gamma>1$ initially, \autoref{lemma:bounding-number-of-sets} makes sure that before running Algorithm \ref{alg:partial-dom-alg}, $\gamma\leq 1$.
    Moreover, we can provide an even better lower bound for $\gamma$ by noticing that if $\gamma\leq 1/2$, the claimed running time of this algorithm is worse than simply running the baseline $n^k$ algorithm, and in fact if $\gamma\leq 1/2$, any algorithm that runs significantly better than $n^k$ would refute the $3$-Uniform Hyperclique Hypothesis (see Section~\ref{section:kh-OV-LB} for details).
    Hence, we may assume $1/2\leq \gamma\leq 1$. 
    We now provide an analysis similar to the one in the proof of Theorem \ref{theorem:partial-dom-algorithm}.
    By Lemma \ref{lemma:listing-bundles}, for any fixed $0\leq c\leq (k-1)/2$, there are at most $\bigO(|X|\Delta_s^{c}\Delta_f^{2c})$ many $c$-bundles in $G'$ and each $c$-bundle has size $1+2c$. 
    Assume that we are given a promise that an optimal solution admits a balanced arity-reducing hypercut after removing two bundles $S_1, S_2$, such that $S_1$ is a $c_1$-bundle and $S_2$ is a $c_2$-bundle. Then we have a total of $\bigO(\min\{n,\Delta_s\Delta_f\}^2\Delta_s^{c_1+c_2}\Delta_f^{2c_1+2c_2})$ iterations to guess $S_1, S_2$, and the remaining number of vertices to guess with the max-weight-triangle algorithm is $k':= k-(2c_1 + 2c_2 + 2)$.
    Up to constant factors, we can bound the running time of our algorithm as follows.
    \begin{align*}
        T_k(n,u, \Delta_s, \Delta_f) \leq &\; (\Delta_f\Delta_s^\gamma)^{\omega\lceil k/3\rceil}\Delta_s & 
        \\ 
        &+ \sum_{0\leq c_1\leq (k-1)/2} \Big( (\Delta_f\Delta_s^\gamma)^{\omega\lceil(k-2c_1-1)/3\rceil+1}\Delta_s^{c_1+1}\Delta_f^{2c_1} & 
        \\&+ \sum_{0\leq c_2\leq (k-2c_1-2)/2} (\Delta_f\Delta_s^\gamma)^{\omega\lceil(k-2c_1-2c_2-2)/3\rceil+2}\Delta_s^{c_1+c_2+1}\Delta_f^{2c_1+2c_2} \Big) & 
    \end{align*}
    We proceed similarly as in the proof of Theorem \ref{theorem:partial-dom-algorithm}, and prove that depending on value of $\omega$, the value of $T_k(n,u, \Delta_s, \Delta_f)$ is bounded by either $(\Delta_f\sqrt{\Delta_s})^k(\Delta_f\Delta_s)^{\bigO(1)}$, or by $(\Delta_f\Delta_s^\gamma)^{k\omega/3}(\Delta_f\Delta_s)^{\bigO(1)}$. In order to achieve this, we bound each of the terms in the expression above by the corresponding value.
    Before doing this, we need to relate the values of $\Delta_f$ and $\Delta_s$. To this end, we write $\Delta_f = \Delta_s^\sigma$ for some fixed $\sigma\geq 0$.
    \begin{claim}
        If $\omega \leq \frac{3(\sigma+1/2)}{(\sigma+\gamma)}$, then $T_k(n,u, \Delta_s, \Delta_f)\leq \bigO\big((\Delta_f\sqrt{\Delta_s})^k(\Delta_f\Delta_s)^{3}\big)$. 
    \end{claim}
    \begin{proof}
        We rewrite the time complexity expression while plugging in $\Delta_f = \Delta_s^\sigma$. 
        \begin{align*}
        T_k(n,u, \Delta_s, \Delta_f) \leq &\; (\Delta_s^\sigma\Delta_s^\gamma)^{\omega\lceil k/3\rceil}\Delta_s & 
        \\ 
        &+ \sum_{0\leq c_1\leq (k-1)/2} \Big( (\Delta_s^\sigma\Delta_s^\gamma)^{\omega\lceil(k-2c_1-1)/3\rceil+1}\Delta_s^{c_1+1}\Delta_s^{\sigma 2c_1} & 
        \\&+ \sum_{0\leq c_2\leq (k-2c_1-2)/2} (\Delta_s^\sigma\Delta_s^\gamma)^{\omega\lceil(k-2c_1-2c_2-2)/3\rceil+2}\Delta_s^{c_1+c_2+1}\Delta_s^{\sigma(2c_1+2c_2)} \Big) & 
    \end{align*}
    We now proceed to show that each of the summands satisfies the desired inequality (up to constant factors).
    Let us start with the simplest summand.
    \begin{align*}
    (\Delta_s^\sigma\Delta_s^\gamma)^{\omega\lceil k/3\rceil}\Delta_s
    &\leq \Delta_s^{(\sigma+\gamma)\omega (k+2)/3}\Delta_s & (\lceil k/3\rceil \leq (k+2)/3)\\
    & \leq \Delta_s^{(\sigma+1/2)(k+2)}\Delta_s & (\omega \leq 3\tfrac{(\sigma + 1/2)}{\sigma+\gamma})\\
    & = (\Delta_f\sqrt{\Delta_s})^{k+2}\Delta_s.\\
    & < (\Delta_f\sqrt{\Delta_s})^{k}\Delta_f^2\Delta_s^3
     \end{align*}
     We now repeat the similar approach for the remaining two terms.
     \begin{align*}
         (\Delta_s^\sigma\Delta_s^\gamma)^{\omega\lceil(k-2c_1-1)/3\rceil+1}\Delta_s^{c_1+1}\Delta_s^{\sigma 2c_1} &\leq \Delta_s^{(\sigma+\gamma)(\omega(k-2c_1+1)/3+1)} \Delta_s^{c_1+1+\sigma 2c_1} \\
         & \leq \Delta_s^{(\sigma+1/2)(k-2c_1+1) + \sigma+ 1} \Delta_s^{c_1+1+\sigma 2c_1} & (\omega \leq 3\tfrac{(\sigma + 1/2)}{\sigma+\gamma})\\
        & = \Delta_s^{(\sigma+1/2)(k+1) + \sigma+ 2} \\
        & = (\Delta_f\sqrt{\Delta_s})^{(k+1)}\Delta_f\Delta_s^2 \\
        & < (\Delta_f\sqrt{\Delta_s})^{k}\Delta_f^2\Delta_s^3.
     \end{align*}
    \begin{align*}
        (\Delta_s^\sigma\Delta_s^\gamma)^{\omega\lceil(k-2c_1-2c_2-2)/3\rceil+2}\Delta_s^{c_1+c_2+1}\Delta_s^{\sigma(2c_1+2c_2)} 
        &\leq \Delta_s^{(\sigma+\gamma)(\omega(k-2c_1-2c_2)/3+2)} \Delta_s^{c_1+c_2+1+\sigma(2c_1+2c_2)}\\
        & \leq \Delta_s^{(\sigma+1/2)(k-2c_1-2c_2)+2\sigma + 2}\Delta_s^{c_1+c_2+1+\sigma(2c_1+2c_2)} \\
        & = \Delta_s^{(\sigma+1/2)k + 2\sigma + 3} \\ 
        & = (\Delta_f\sqrt{\Delta_s})^{k}\Delta_f^2\Delta_s^3\\
    \end{align*}
    Since each of the summands is bounded by $\bigO\big((\Delta_f\sqrt{\Delta_s})^k(\Delta_f\Delta_s)^{5}\big)$, the total time complexity is thus bounded by $f(k)(\Delta_f\sqrt{\Delta_s})^k(\Delta_f\Delta_s)^{5}$ (for some computable function $f$), which for fixed constant $k$ gives us the desired bound.
    \end{proof}
    \begin{claim}
        If $\omega > \frac{3(\sigma+1/2)}{(\sigma+\gamma)}$, then $T_k(n,u, \Delta_s, \Delta_f)\leq \bigO\big((\Delta_f\Delta_s^\gamma)^{k\omega/3}(\Delta_f\Delta_s)^3\big)$.
    \end{claim}
    \begin{proof}
        We proceed similarly as above, by providing a bound for each of the summands.
        This bound clearly holds for the term $(\Delta_f\Delta_s^\gamma)^{\omega\lceil k/3\rceil}\Delta_s$, so we only need to bound the remaining two terms.
        Similarly as above, we will write $\Delta_f = \Delta_s^\sigma$.
        \begin{align*}
        \mathmakebox[4\mathindent][l]{
            (\Delta_s^\sigma\Delta_s^\gamma)^{\omega\lceil(k-2c_1-1)/3\rceil+1}\Delta_s^{c_1+1}\Delta_s^{\sigma 2c_1} }\\
            &\leq \Delta_s^{(\sigma+\gamma)(\omega(k-2c_1+1)/3+1)} \Delta_s^{c_1+1+\sigma 2c_1} \\
            & = \Delta_s^{(\sigma+\gamma)(\omega(k+1)/3+1)}\Delta_s^{-(\sigma+\gamma)\omega 2c_1/3 +c_1+1+\sigma 2c_1} & \text{(rearranging terms)}\\
            & \leq \Delta_s^{(\sigma+\gamma)(\omega(k+1)/3+1)}\Delta_s^{-(\sigma+1/2)2c_1 +c_1+1+\sigma 2c_1} & (\omega>\tfrac{3(\sigma+1/2)}{(\sigma+\gamma)})\\
            & = \Delta_s^{(\sigma+\gamma)(\omega(k+1)/3)}\Delta_s^{\sigma + 2} & (\gamma \leq 1)\\ 
            & = (\Delta_s^\gamma \Delta_f)^{(\omega(k+1)/3)}\Delta_s^2\Delta_f\\
        \end{align*}
        Applying the same approach, we bound the remaining term as follows.
        \begin{align*}
            \mathmakebox[4\mathindent][l]{
            (\Delta_s^\sigma\Delta_s^\gamma)^{\omega\lceil(k-2c_1-2c_2-2)/3\rceil+2}\Delta_s^{c_1+c_2+1}\Delta_s^{\sigma(2c_1+2c_2)} }
            \\
            &\leq \Delta_s^{(\sigma+\gamma)(\omega k/3+2)} \Delta_s^{(\sigma+\gamma)(\omega (-2c_1-2c_2)/3) + c_1+c_2+1+\sigma(2c_1+2c_2)}\\
            & \leq \Delta_s^{(\sigma+\gamma)(\omega k/3+2)} \Delta_s^{(\sigma+1/2)(-2c_1-2c_2) + c_1+c_2+1+\sigma(2c_1+2c_2)}\\
            & \leq \Delta_s^{(\sigma+\gamma)(\omega k/3)} \Delta_s^{\gamma+2(\sigma+1)}\\
            & = (\Delta_f\Delta_s^\gamma)^{\omega k/3}\cdot \Delta_s^3 \Delta_f^2
        \end{align*}
        We have thus bounded each summand by the desired value, and as argued above, since we only have $f(k) = \bigO(1)$ iterations in our two sums, this yields the desired bound on total running time of the algorithm up to constant factors.
    \end{proof}
    By combining the two claims, we obtain that depending on the value of matrix multiplication exponent $\omega$, we can bound the running time of our algorithm by either $\bigO\big((\Delta_f\sqrt{\Delta_s})^k(\Delta_f\Delta_s)^{5}\big)$, 
    or by $\bigO\big((\Delta_f\Delta_s^\gamma)^{k\omega/3}(\Delta_f\Delta_s)^3\big) = \bigO\big(\min\{n,(\Delta_f\Delta_s)\}^{k\omega/3}(\Delta_f\Delta_s)^3\big)$.
    Independently of the value of $\omega$, the running time is thus bounded by the sum of these two terms, and we get
    \[
    T_k(n,u,\Delta_s, \Delta_f) \leq \bigO \big( (\min\{n,(\Delta_f\Delta_s)\}^{k\omega/3} + (\Delta_f\sqrt{\Delta_s})^k )\Delta_f^2\Delta_s^3 \big)
    \]
\end{proof}
Perhaps surprisingly, depending on the size of the universe $u$, in the general case we might be able to get further improvements.
Crucially, in the special case of the Max $k$-Cover when reducing from Partial $k$-Dominating Set (as constructed above), if  the universe is small ($|Y|<\Delta^{3/2}$), then also the number of sets is small ($|X|<\Delta^{3/2}$), and any significant improvement over the baseline $\bigO(n^k)$ algorithm would refute the $3$-uniform $k$-Hyperclique hypothesis (see Section \ref{section:kh-OV-LB}). However, in general this is not necessarily true and we can have instances where the universe is small, while the number of sets is relatively large and it turns out that if we are given such instance, we can exploit the small universe size to get further improvements over our algorithm.
\paragraph{Small Universe Size} We now demonstrate how one can obtain improvements over the algorithm above when the value $|Y|$ is sufficiently small. 
More precisely, we show that we can obtain an algorithm running in time $\bigO\big(((\min\{n,\Delta_f\cdot \min\{u^{1/3}, \sqrt{\Delta_s}\}\})^k + (\min\{n,\Delta_f\min\{\sqrt{u},\Delta_s\}\})^{k\omega/3})(\Delta_s\Delta_f)^{\bigO(1)}\big)$.

For simplicity, let us assume for now that any solution consists of vertices that admit an arity-reducing hypercut, and give a subroutine that solves any such instance efficiently (we will handle the obstructions separately). The main strategy of our approach is to first show that any potential solution contains a "heavy" vertex. Then we consider two cases, either there are many heavy vertices, or there are few heavy vertices in $X$. 
We then provide a win-win argument:
\begin{enumerate}[label=(\roman*)]
    \item If there are many heavy vertices, we prove that any potential solution consists exclusively of "moderately heavy" vertices and we reduce to max-weight-triangle instance similarly as before.
    \item If there are only a few heavy vertices, we can afford to guess one and we then recurse.
\end{enumerate}
Formally the key properties of any solution that we discussed above are provided in the Regularization Lemma below.
\begin{lemma}[Regularization Lemma]\label{lemma:regularization-lemma}
    Let $G=(X\cup Y,E)$ be a bipartite graph with $|X| = n$ and $|Y| = u$ with $\max_{x\in X}\deg(x) = \Delta_s$ (resp. $\max_{y\in Y}\deg(y) = \Delta_f$). Let $H_1\subseteq X$ contain all vertices from $X$ whose degree is at least $\frac{\Delta_s}{k}$ and $H_2\subseteq X$ contain all vertices from $X$ whose degree is at least $\frac{\Delta_s}{2k}$.
    Let $x_1,\dots, x_k\in X$ be the vertices that maximize the value $|N(x_1)\cup \dots \cup N(x_k)|$.
    Then the following conditions are satisfied.
    \begin{enumerate}[label=(\roman*)]
        \item At least one vertex from $x_1,\dots, x_k$ is contained in $H_1$. 
        \item If $|H_1| \geq 2k^2 \cdot \Delta_f$, then all of the vertices $x_1,\dots, x_k$ are contained in $H_2$.
    \end{enumerate}
\end{lemma}
\begin{proof}
    Let us first argue why the first item is true. If we assume that no vertex $x_1,\dots, x_k$ is contained in $H_1$, then $|N(x_1)\cup \dots \cup N(x_k)| < k\cdot \frac{\Delta_s}{k} = \Delta_s$. 
    However, by assumption, $X$ contains a vertex $x$ with degree equal to $\Delta_s$.
    Hence, replacing any of the vertices $x_1,\dots, x_k$ by $x$ would strictly improve the number of covered vertices in $Y$, contradicting maximality of $|N(x_1)\cup \dots \cup N(x_k)|$.
    For the second item, assume that $|H_1|\geq 2k^2\cdot \Delta_f$ and that $x_1\not\in H_2$, that is $\deg(x_1)<\frac{\Delta_s}{2k}$.
    We prove that there is a vertex $x\in H_1$ (and thus in $H_2$), such that replacing $x_1$ by $x$ yields a strictly better solution. For any $x\in X$, let $P(x)$ denote the number of paths of length $2$ in $G$ that contain $x$ as an endpoint. Note that for any vertex $x$ it holds that $P(x)\leq \Delta_s\cdot \Delta_f$. Consider now the following chain of inequalities.
    \begin{align*}
        \sum_{x\in H_1}|N(x)\setminus \big(N(x_2)\cup \dots \cup N(x_k)\big)| & \geq \sum_{x\in H_1} \Big( \deg(x) - \sum_{i=2}^k|N(x_i)\cap N(x)|\Big)&\\
         &\geq |H_1|\cdot \frac{\Delta_s}{k} - \sum_{x\in H_1}\sum_{i=2}^k|N(x_i)\cap N(x)| &\\
        & \geq |H_1|\cdot \frac{\Delta_s}{k} - \sum_{i=2}^kP(x_i) &\\
        & \geq |H_1|\cdot \frac{\Delta_s}{k} - (k-1)\Delta_s\Delta_f &\\
        & \geq |H_1|\cdot \frac{\Delta_s}{k} - k\Delta_s\frac{|H_1|}{2k^2} & (|H_1|\geq 2k^2\cdot \Delta_f)\\
        & = |H_1|\cdot \frac{\Delta_s}{2k}
    \end{align*}
    In particular, this inequality implies that there exists a vertex $x\in H_1$ such that $|N(x)\setminus\big(N(x_2)\cup\dots \cup N(x_k)\big)|\geq \frac{\Delta_s}{2k}>N(x_1)$, hence $|N(x)\cup\dots\cup N(x_k)|>|N(x_1)\cup\dots\cup N(x_k)|$, contradicting the maximality assumption and concluding the proof.
\end{proof}
In light of the regularization Lemma consider the following algorithm.
\begin{algorithm}
\begin{algorithmic}[1]
    \Procedure{regularize-and-solve}{$X,Y,E,k$}
        \State Let $\Delta_s, \Delta_f, H_1, H_2$ be as in \autoref{lemma:regularization-lemma} 
        \If{$|H_1|<2k^2 \Delta_f$}
            \State \label{alg:solve:line4}
            \Return {$\max_{x\in H_1} |N(x)| + $ \Call{regularize-and-solve}{$X-x, Y-N(x), E', k-1$}} \footnotemark 
        \EndIf
        \State $X'\gets H_2$, $Y'\gets N(H_2)$ \Comment{Regularization step}
        \State $E\gets E(G)\cap (X'\times Y')$
        \State Remove all but the heaviest $\min \{k\Delta_s\Delta_f, |X'|\}$ vertices from $X'$. \label{alg:solve:line8}
        \State Let $T$ be as in Lemma \ref{lemma:max-weight-triangle-algorithm} \Comment{Reduction to max-weight-triangle}
        \State \Return max-weight-triangle($T$)
    \EndProcedure
\end{algorithmic}
\caption{}
\label{alg:small-universe-algorithm}
\end{algorithm}
\footnotetext{$E'$ represents the set of edges $E$ restricted to those with no endpoints in $x$.}
\begin{lemma}\label{lemma:alg2-correctness}
    Let $G=(X\cup Y,E)$ be a bipartite graph with $|X| = n$ and $|Y| = u$ with $\max_{x\in X}\deg(x) = \Delta_s$ (resp. $\max_{y\in Y}\deg(y) = \Delta_f$).
    Assume that after running line \ref{alg:solve:line8}, there exists an optimal solution $x_1,\dots, x_k$, 
    that admits an arity-reducing hypercut. Then Algorithm \ref{alg:small-universe-algorithm} returns the correct value.
\end{lemma}
\begin{proof}
    Let $H_1, H_2$ be as defined in the Regularization Lemma and let $x_1,\dots, x_k\in X$ be a solution that maximizes the value $|N(x_1)\cup \dots \cup N(x_k)|$.
    Assume first that the algorithm enters Line \ref{alg:solve:line4} and without loss of generality assume that $x_1\in H_1$ (by Regularization Lemma there exists at least one $x_i\in H_1$).
    Then the algorithm returns the value $|N(x_1)| + \textsc{regularize-and-solve}(X-x_1, Y-N(x_1), E', k-1)$. By applying a simple induction on $k$, we can conclude that this is exactly equal to $|N(x_1)\cup \dots \cup N(x_k)|$.
    
    Assume now that the algorithm does not reach Line \ref{alg:solve:line4}. 
    By applying Regularization Lemma, the Regularization Step yields an equivalent instance.
    Moreover, by Lemma \ref{lemma:bounding-number-of-sets}, Line \ref{alg:solve:line8} also yields an equivalent instance.
    By assumption that there is an optimal solution that contains an arity-reducing hypercut, the same argument as in the proof of Theorem \ref{theorem:partial-dom-algorithm} can be applied to conclude that we get the correct solution.
\end{proof}
\begin{lemma}\label{lemma:alg2-running-time}
    Algorithm \ref{alg:small-universe-algorithm} runs in $\bigO\big((\Delta_f^k + (\min\{n,\Delta_f\min\{\sqrt{u},\Delta_s\}\})^{k\omega/3})\Delta_s^3\Delta_f^2\big)$.
\end{lemma}
\begin{proof}
    After each recursive call, we remove a heavy vertex from $X$ and its neighborhood from~$Y$. By doing so, we may decrease the values of $\Delta_s$ and $\Delta_f$ each time the recursion is called. We denote by $\Delta_s',\Delta_f'$ the values of $\max_{x\in X}\deg(x)$, $\max_{y\in Y}\deg(y)$ respectively, after the last recursive call and the regularization step.
    Let $k'$ be the number of vertices remaining to guess after the last recursive call.
    Recall that by Lemma \ref{lemma:max-weight-triangle-algorithm}, the max-weight-triangle$(V_1,V_2, V_3)$ takes at most $|X'|^{\omega (k'/3+1)}\cdot\Delta_s$, so we need to bound the size of $X$ after the regularization step by $\min\{n,\Delta_f\min\{\Delta_s, \sqrt{u}\}\}$ to get the desired running time.

    Recall that after the regularization step, all vertices in $X$ have degree at least $\frac{\Delta_s'}{2k}$ and at most $\frac{\Delta_s'}{k}$.
    We proceed by double counting the number of edges in $G$.
    Namely, by simple counting argument one can obtain that there are at least $|X'|\frac{\Delta'_s}{2k} = \bigO(|X|\Delta'_s)$ many edges in $G$ after regularization step (each of the $|X'|$ vertices is incident to at least $\frac{\Delta'_s}{2k}$ edges, and since $G$ is bipartite, there are no edges within $X'$).
    In particular, if we denote the number of edges in $G$ by $m$, this implies that after regularization there are at most $\bigO(m/\Delta'_s)$ many vertices in $X$. 
    On the other hand, by the similar argument, there are at most $u\Delta'_f$ many edges, hence, plugging this in for $m$, we obtain the bound for the number of vertices in $X$ after regularization as:
    \[
    |X|\leq \bigO(u\Delta'_f/\Delta'_s) \leq \bigO(u\Delta_f/\Delta'_s).
    \]
    On the other hand, in Line \ref{alg:solve:line8} of the algorithm we remove all but heaviest $\min\{n, k\Delta'_f\Delta'_s\}$ many vertices.
    Hence, we can bound (up to $f(k) = \bigO(1)$ factors) the number of vertices in $X$ as 
    \[
    |X| \leq \min\{\Delta_f'\Delta'_s, \Delta_f\frac{u}{\Delta'_s}, \}\leq \min\{\Delta_f\Delta'_s, u\Delta_f\frac{u}{\Delta'_s}, \}.
    \]
    By observing that the value of $\Delta_s'$ is always between $1$ and $\Delta_s$, we obtain:
    \[
    |X| \leq \max_{1\leq \Delta_s'\leq \Delta_s} \min\{\Delta_f\Delta'_s, \Delta_f\frac{u}{\Delta'_s}\} \leq \Delta_f\sqrt{u}.
    \]
    Combining the upper bounds from above yields
    \[
    |X| \leq \min\{n,\Delta_f\min\{\sqrt{u}, \Delta_s\}\}.
    \]
  %
    We have now proved that once the algorithm reaches the last recursive call with the value $k'$, the last recursive call takes 
    $\bigO\big((\min\{n,\Delta_f\min\{\sqrt{u},\Delta_s\}\})^{k'\omega/3}|X|\Delta_s)\big)$.
    This allows us to bound the total running time as follows (for simplicity we ignore the constant factors in the expression below).
    \begin{align*}
        T_k(n,u,\Delta_s, \Delta_f) \leq \sum_{i = 0}^k \Delta_f^i \cdot (\min\{n,\Delta_f\min\{\sqrt{u},\Delta_s\}\})^{(k-i)\omega/3}|X|\Delta_s
    \end{align*}
    We now consider two cases. First consider the case when $\Delta_f>(\min\{n,\Delta_f\min\{\sqrt{u},\Delta_s\}\})^{\omega/3}$. 
    This yields the following
    \begin{align*}
        T_k(n,u,\Delta_s, \Delta_f) &\leq \sum_{i = 0}^k \Delta_f^i \cdot ((\min\{n,\Delta_f\min\{\sqrt{u},\Delta_s\}\})^{\omega/3})^{(k-i)}|X|\Delta_s\\
        & < \sum_{i = 0}^k \Delta_f^i \cdot \Delta_f^{(k-i)}|X|\Delta_s\\
        & = k\Delta_f^k.
    \end{align*}
    On the other hand, if $\Delta_f\leq (\min\{n,\Delta_f\min\{\sqrt{u},\Delta_s\}\})^{\omega/3}$, we apply a similar argument to show
    \begin{align*}
        T_k(n,u,\Delta_s, \Delta_f) & \leq \sum_{i = 0}^k \Delta_f^i \cdot (\min\{n,\Delta_f\min\{\sqrt{u},\Delta_s\}\})^{(k-i)\omega/3}|X|\Delta_s \\
        & \leq \sum_{i = 0}^k (\min\{n,\Delta_f\min\{\sqrt{u},\Delta_s\}\})^{i\omega/3} \cdot (\min\{n,\Delta_f\min\{\sqrt{u},\Delta_s\}\})^{(k-i)\omega/3}|X|\Delta_s \\
        & = k(\min\{n,\Delta_f\min\{\sqrt{u},\Delta_s\}\})^{k\omega/3}|X|\Delta_s.
    \end{align*}
    Hence, the total running time of the algorithm can be bounded by 
    \begin{align*}
    T_k(n,u,\Delta_s,\Delta_f) &\leq \bigO(\Delta_f^k + (\min\{n,\Delta_f\min\{\sqrt{u},\Delta_s\}\})^{k\omega/3}|X|\Delta_s \cdot |X|\Delta_s) \\
    &\leq \bigO(\Delta_f^k + (\min\{n,\Delta_f\min\{\sqrt{u},\Delta_s\}\})^{k\omega/3}\Delta_f^2\Delta_s^3).
    \end{align*}
\end{proof}
So far we have assumed that the optimal solution admits a balanced arity-reducing hypercut. Clearly, there are instances where this assumption does not hold, and we proceed to show how to construct a self-reduction from any such instance to a smaller instance where we can find an optimal solution that admits an arity-reducing hypercut. 
To this end, we make a distinction between three cases, based on the ratio between the size of the universe and the value of $\Delta_s$:
\begin{enumerate}[label=(\roman*)]
    \item $\Delta_s^2\leq u$\label{item:large-universe}
    \item $\Delta_s^{3/2}\leq u\leq \Delta_s^2$\label{item:intermediate-universe}
    \item $u\leq \Delta_s^{3/2}$.\label{item:small-universe}
\end{enumerate}
Moreover, for the rest of this section, we assume that $n\geq \Delta_f\min\{\sqrt{\Delta_s}, u^{1/3}\}$, since otherwise, by running the baseline $n^{k+o(1)}$ algorithm, we achieve the running time from Theorem \ref{theorem:max-k-cover-alg}.
\subparagraph{Case \ref{item:large-universe}: $\Delta_s^2\leq u$.} In this case, we have $\min\{u^{1/3},\sqrt{\Delta_s}\} = \sqrt{\Delta_s}$, and $\min\{\sqrt{u},\Delta_s\} = \Delta_s$, hence the running time from Theorem \ref{theorem:max-k-cover-alg} becomes 
\[
\bigO\big( (\Delta_f\sqrt{\Delta_s})^k + \min\{n,\Delta_f\Delta_s\}^{k\omega/3}(\Delta_s\Delta_f)^5\big).
\]
We notice that we can achieve this time by a simply running the algorithm from Lemma \ref{lemma:max-k-cover-algorithm-large-universe}. 
\subparagraph{Case \ref{item:intermediate-universe}: $\Delta_s^{3/2}\leq u\leq \Delta_s^2$.} Recall that in this case we want to construct an algorithm that solves Max-$k$-Cover in time ${\bigO\big(((\Delta_f\cdot \sqrt{\Delta_s})^k + (\min\{n,\Delta_f\cdot \sqrt{u}\})^{k\omega/3})\cdot (\Delta_s\Delta_f)^{\bigO(1)}\big)}$.
The strategy is to apply Lemma \ref{lemma:bundles-and-cuts} and guess two disjoint bundles $D_1, D_2$, such that there is an optimal solution $S$ containing these two bundles and such that $S-(D_1\cup D_2)$ admits a balanced arity-reducing hypercut, and then by Lemma \ref{lemma:alg2-correctness} we can find the optimal solution using Algorithm \ref{alg:small-universe-algorithm}.
Intuitively, the first part of the running time will correspond to the contribution of guessing the bundles, while the second part comes from running Algorithm \ref{alg:small-universe-algorithm}. 
\begin{lemma}
    Given a bipartite graph $G = (X\cup Y,E)$ with $|X| = n$, $|Y|=u$, such that every $x\in X$ satisfies $\deg(x)\leq \Delta_s$, and respectively every $y\in Y$ satisfies $\deg(y)\leq \Delta_f$, we can find the $k$ vertices $x_1,\dots, x_k\in X$ maximizing the value $|N(x_1)\cup\dots\cup N(x_k)|$ in time 
    \[{\bigO\big(((\Delta_f\cdot \sqrt{\Delta_s})^k + (\min\{n,\Delta_f\cdot \sqrt{u}\})^{k\omega/3})\cdot (\Delta_s\Delta_f)^{5}\big)}.\]
\end{lemma}
\begin{proof}
    Consider the following algorithm.
    \begin{algorithm}
\begin{algorithmic}[1]
    \State $t'\gets 0$
    \For{bundle $S_1$ with $0 \leq |S_1| \leq k$}
        \For{bundle $S_2$ with $0 \leq |S_1| + |S_2| \leq k$ and $S_1\cap S_2 = \emptyset$}
            \State $k' \gets k - |S1| - |S2|$
            \State $X' \gets X-S_1-S_2$,  $Y'\gets Y-N(S_1)-N(S_2)$
            \State $E' \gets E(G)\cap (X'\times Y')$
            \State $t\gets \max \{t, |N(S1)| + |N(S2)| +\textsc{regularize-and-solve}(X',Y', E', k')$
        \EndFor
    \EndFor 
    \Return $t$
\end{algorithmic}
\caption{}
\label{alg:intermediate-universe-size}
\end{algorithm}
The correctness of the algorithm is a straightforward consequence of Lemmas \ref{lemma:bundles-and-cuts} and \ref{lemma:alg2-correctness}. 
We proceed to show the running time of the algorithm.
To simplify the analysis, let us relate the parameters. We can write $u$ as $\Delta_s^\tau$, for some $3/2\leq \tau\leq 2$.
Similarly, we write $\Delta_f = \Delta_s^\sigma$.
Furthermore, since $n\geq \Delta_f$, we can write $n = \Delta_f\Delta_s^\beta$.
Finally, by setting $\gamma:=\min\{\beta, \tau/2\}$, we can rewrite the running time of Algorithm \ref{alg:small-universe-algorithm} in this notation as:
\[
    \bigO\big((\Delta_f^k + (\min\{n,\Delta_f\sqrt{u})^{\omega\lceil k/3\rceil})\Delta_s^3\Delta_f^2\big) = \bigO\big((\Delta_s^{\sigma k} + \Delta_s^{(\sigma+\gamma)\omega\lceil k/3\rceil})\Delta_s^{3+2\sigma}\big).
\]
It is now easy to see that we can bound the running time of Algorithm \ref{alg:intermediate-universe-size} similarly as in proof of Lemma \ref{lemma:max-k-cover-algorithm-large-universe}, by distinguishing between the cases when 1) $S_1,S_2$ are both empty, 2) $S_1$ is a bundle and $S_2$ is empty, and 3) $S_1,S_2$ are both bundles.
    \begin{equation}\label{eq:max-k-cover-alg-runtime}
    \begin{split}
        T_k(n,u, \Delta_s, \Delta_f) \leq &\;(\Delta_s^{\sigma k} + \Delta_s^{(\sigma+\gamma)\omega\lceil k/3\rceil})\Delta_s^{3+2\sigma} & 
        \\ 
        &+ \sum_{0\leq c_1\leq (k-1)/2} \Big(
        (\Delta_s^{\sigma (k-2c_1-1)} + \Delta_s^{(\sigma+\gamma)(\omega\lceil (k-2c_1-1)/3\rceil+1)})\Delta_s^{3+2\sigma} \cdot \Delta_s^{c_1 + 2\sigma c_1} & 
        \\& + \sum_{0\leq c_2\leq (k-2c_1-2)/2} (\Delta_s^{\sigma (k-2c_1-2c_2-2)} + \Delta_s^{(\sigma+\gamma)(\omega\lceil (k-2c_1-2c_2-2)/3\rceil+1)})\Delta_s^{3+2\sigma} \cdot \Delta_s^{c_1 + c_2 + 2\sigma( c_1+c_2)}\Big)
    \end{split}
    \end{equation}
    The first goal is to get rid of the $(\Delta_s^{\sigma (k-2c_1-1)}$ factors.
    We can do that by noticing that \[\Delta_s^{\sigma (k-2c_1-1)}\Delta_s^{c_1 + 2\sigma c_1} = \Delta_s^{\sigma(k-1) + c_1} \leq \Delta_s^{\sigma(k-1) + (k-1)/2},\] and this term is already achieved by plugging in the extreme value of $c_1 = (k-1)/2$ to the factor $\Delta_s^{c_1 + 2\sigma c_1}$, hence if we remove $(\Delta_s^{\sigma (k-2c_1-1)}$ from the expression above, the value of $T_k$ stays the same (up to $f(k)=\bigO(1)$ factors). 
    By applying the same argument to the factor $\Delta_s^{\sigma (k-2c_1-2c_2-2)}$, we can bound the running time, up to constant factors by the following expression.
    \begin{align*}
        T_k(n,u, \Delta_s, \Delta_f) &\leq \Delta_s^{2+2\sigma} \Big[\Delta_s^{(\sigma+\gamma)\omega\lceil k/3\rceil}\Delta_s & 
        \\ 
        &+ \sum_{c_1 = 0}^{(k-1)/2} \Big(
        \Delta_s^{(\sigma+\gamma)(\omega\lceil (k-2c_1-1)/3+1)\rceil} \cdot \Delta_s^{c_1 +1}\Delta_s^{ 2\sigma c_1} & 
        \\& + \sum_{c_2 = 0}^{(k-2c_1-2)/2} \Delta_s^{(\sigma+\gamma)(\omega\lceil (k-2c_1-2c_2-2)/3\rceil+2)} \cdot \Delta_s^{c_1 + c_2 + 1} \Delta_s^{\sigma(2c_1+2c_2)}\Big)\Big]
    \end{align*}
    We now observe that the expression in the square brackets is exactly the same as the expression in \autoref{eq:max-k-cover-alg-runtime}, hence, we can reuse the already carried out computations to conclude that up to the constant factors the following inequality is true:
    \[
    T_k(n,u,\Delta_s,\Delta_f) \leq \Delta_s^{2+2\sigma} \big(\Delta_s^{(\sigma + \gamma)(\omega k/3)} + \Delta_s^{(\sigma + 1/2)k}\big) \Delta_s^{2\sigma + 3}.
    \]
    Writing this back in terms of parameters $n,u,\Delta_f,\Delta_s$, we obtain:
    \[
        T_k(n,u,\Delta_s,\Delta_f) \leq \big((\Delta_f\min\{n,\sqrt{u}\})^{\omega k/3} + (\Delta_f\sqrt{\Delta_s})^{k}\big) \Delta_s^{5}\Delta_f^4.
    \]
\end{proof}
\subparagraph{Case \ref{item:small-universe}: $u\leq \Delta_s^{3/2}$.} In this case we want to obtain an algorithm solving Max-$k$-Cover in time $\bigO\Big(\big((\Delta_f u^{1/3})^k + (\Delta_f\sqrt{u})^{k\omega/3}\big)\cdot (\Delta_s\Delta_f)^{\bigO(1)}\Big)$. We remark that if we are given a promise that there exists an optimal solution $S$ of size $k$, such that the subhypergraph of the hypergraph representation $\mathcal{H}(G)$ induced on $S$ contains no hyperedges, then vacuously $S$ admits a balanced arity-reducing hypercut.
This motivates the following approach.
First run Algorithm \ref{alg:small-universe-algorithm} and store the returned value.
Then guess a triple of vertices $x_1,x_2,x_3\in X$, such that $\{x_1,x_2,x_3\}$ is a hyperedge in $\mathcal{H}(G)$ and proceed recursively by removing $x_1,x_2,x_3$ from $X$ and their neighborhood from $Y$.
We can bound the number of hyperedges in $\mathcal{H}(G)$ by $u\Delta_f^3$, by noticing that by definition each hyperedge corresponds to a triple of vertices in $X$ sharing a common neighbor in $Y$. There are only $u$ choices for the common neighbor, and for each there are $\binom{\Delta_f}{3}\leq \Delta_f^3$ choices for the three vertices in $X$.
This gives us the amortized time of $\Delta_fu^{1/3}$ per vertex for guessing the hyperedges and finally, after we have guessed all the hyperedges from a solution, by Lemma \ref{lemma:alg2-correctness} running Algorithm \ref{alg:small-universe-algorithm} yields a correct solution on the remaining $k'$ vertices in time $\bigO((\Delta_f\sqrt{u})^{k'\omega/3}(\Delta_f\Delta_s)^3)$.
We prove the details below.
\begin{lemma}\label{lemma:small-universe}
    Given a bipartite graph $G = (X\cup Y,E)$ with $|X| = n$, $|Y|=u$, such that every $x\in X$ satisfies $\deg(x)\leq \Delta_s$, and respectively every $y\in Y$ satisfies $\deg(y)\leq \Delta_f$, we can find the $k$ vertices $x_1,\dots, x_k\in X$ maximizing the value $|N(x_1)\cup\dots\cup N(x_k)|$ in time 
    \[{\bigO\Big(\big((\Delta_f u^{1/3})^k + \min\{n,\Delta_f\sqrt{u}\}^{k\omega/3}\big)\cdot (\Delta_s\Delta_f)^{3}\Big)}\]
\end{lemma}
\begin{proof}
    Consider the following algorithm.
    \begin{algorithm}
    \begin{algorithmic}[1]
        \Procedure{solve}{$X,Y,E,k$}
            \State $t \gets \textsc{regularize-and-solve}(X,Y,E,k)$
            \For{$\{x_1,x_2,x_3\}\in E(\mathcal{H}(G))$}
                \State $X' \gets X-\{x_1,x_2,x_3\}$, $Y'\gets Y-(N(x_1)\cup N(x_2)\cup N(x_3))$
                \State $E' \gets (X'\times Y') \cap E(G)$
                \State $t\gets \max\{t, |N(x_1)\cup N(x_2)\cup N(x_3)| + \textsc{solve}(X', Y', E', k-3))\}$
            \EndFor
            \Return $t$
        \EndProcedure
    \end{algorithmic}
    \caption{}
    \label{alg:small-universe-size}
    \end{algorithm}
    We first argue correctness of this algorithm. Let $x_1,\dots, x_k\in X$ be the vertices that maximize the $|N(x_1)\cup \dots \cup N(x_k)| = \textsc{opt}$. By the proof of Lemma \ref{lemma:alg2-correctness}, the function $\textsc{regularize-and-solve}(X,Y,E,k)$ always returns the value that is $\leq \textsc{opt}$. Moreover, before calling our function recursively, we first delete the neighborhood of the guessed triple from $Y$, thus this property gets preserved in every recursive call.
    It suffices to prove that the returned value is $\geq \textsc{opt}$.
    To this end, we proceed by induction on the number of hyperedges in the corresponding hypergraph $\mathcal{H}(G)$ induced on an optimal solution.
    For a base case, assume that there is an optimal solution such that the corresponding hypergraph $\mathcal{H}(G)$ induced on this solution contains no hyperedges. Then by Lemma \ref{lemma:alg2-correctness}, the function $\textsc{regularize-and-solve}(X,Y,E,k)$ will return $\textsc{opt}$ and since at each recursive call we take the maximum of the value found so far and the value after guessing a hyperedge, clearly the returned value $t$ satisfies $t\geq \textsc{opt}$. 
    On the other hand, if (without loss of generality) $x_1,x_2,x_3\in E(\mathcal{H}(G))$, then the returned value $t$ satisfies $t\geq |N(x_1)\cup N(x_2)\cup N(x_3)| + \textsc{solve}(X-\{x_1,x_2,x_3\}, Y'-(N(x_1)\cup N(x_2)\cup N(x_3)), E', k-3)$. 
    It now suffices to show that $\textsc{solve}(X-\{x_1,x_2,x_3\}, Y-(N(x_1)\cup N(x_2)\cup N(x_3)), E', k-3)$ returns the value $t \geq |(N(x_4)\cup \dots \cup N(x_k))\setminus (N(x_1)\cup N(x_2)\cup N(x_3))|$. 
    By noticing that the subhypergraph of $\mathcal{H}(G)$ induced on $x_4,\dots, x_k$ has strictly fewer hyperedges than the subhypergraph of $\mathcal{H}(G)$ induced on $x_1,\dots, x_k$, we can apply induction hypothesis to conclude the proof.

    We now have to argue the running time of the algorithm. As briefly explained above, we can bound the number of the hyperedges in the graph $\mathcal{H}(G)$ by $u\Delta_f^3$. It is now straightforward to verify that the time complexity of the algorithm satisfies the following inequality (using Lemma \ref{lemma:alg2-running-time}). For simplicity, we drop the constant factors.
    \begin{align*}
        T(n,u,\Delta_s, \Delta_f) &\leq \sum_{c=0}^{k/3} (u\Delta_f^3)^c \big((\Delta_f^{k-3c} + (\min\{n,\Delta_f\sqrt{u}\})^{(k-3c)\omega/3})\Delta_s^3\Delta_f^2\big) \\
        & \leq \sum_{c=0}^{k/3} (u\Delta_f^3)^c \min\{n,\Delta_f\sqrt{u}\}^{(k-3c)\omega/3}\Delta_s^3\Delta_f^2 & \text{(up to constant factors)}
    \end{align*}
    We distinguish between two cases. Either $(u\Delta_f^3)\geq \min\{n,\Delta_f\sqrt{u}\}^\omega$, in which case for any $0\leq c\leq k/3$ it holds
    \[
        (u\Delta_f^3)^c \min\{n,\Delta_f\sqrt{u}\}^{(k-3c)\omega/3}\Delta_s^3\Delta_f^2 \leq (u\Delta_f^3)^{k/3}\Delta_s^3\Delta_f^2.
    \]
    And hence, in this case we can write (up to $f(k)=\bigO(1)$ factors):
    \[
        T(n,u,\Delta_s, \Delta_f) \leq (u^{1/3}\Delta_f)^{k}\Delta_s^3\Delta_f^2.
    \]
    Otherwise, if $(u\Delta_f^3) < \min\{n,\Delta_f\sqrt{u}\}^\omega$, then for any $0\leq c\leq k/3$ it holds 
    \[
        (u\Delta_f^3)^c \min\{n,\Delta_f\sqrt{u}\}^{(k-3c)\omega/3}\Delta_s^3\Delta_f^2 \leq \min\{n,\Delta_f\sqrt{u}\}^{k\omega/3}\Delta_s^3\Delta_f^2.
    \]
    Finally, we can conclude that, up to $f(k)=\bigO(1)$ factors, we can always bound the time complexity of this algorithm as
    \[
        T(n,u,\Delta_s, \Delta_f) \leq \big((u^{1/3}\Delta_f)^{k} + \min\{n,\Delta_f\sqrt{u}\}^{k\omega/3}\big)\Delta_s^3\Delta_f^2. \qedhere
    \]
\end{proof}

\section{Conditional lower bounds via $(k,h)$-maxIP/minIP}\label{section:kh-OV-LB}
In this section we prove that the algorithm we constructed in the last section is conditionally optimal. That is, any significant improvement of our algorithms would refute either $k$-Clique Hypothesis, or $3$-Uniform Hyperclique Hypothesis.
To do this we construct efficient reductions from two intermediate problems, namely $(k,h)$-minIP and $(k,h)$-maxIP.
Notably, for even values of~$h$, we reduce from $(k,h)$-minIP to Partial $k$-Dominating Set and Max-$k$-Cover, and for odd values of $h$, we reduce from $(k,h)$-maxIP to Partial $k$-Dominating Set and Max-$k$-Cover. 
We then show that this by extension gives us efficient reductions from $h$-Uniform Hyperclique Detection (if~$h\geq 3$) and from $k$-Clique Detection (if $h=2$) to Partial $k$-Dominating Set and Max-$k$-Cover. 
More precisely, we prove the following two main theorems for this section.
\begin{theorem}\label{theorem:max-k-cover-LB}
    Given a collection of $n$ sets $X:=\{S_1,\dots, S_n\}$ over the universe $Y:=[u]$ such that the maximum size of a set in $X$ is $\Delta_s$ and the maximum frequency of an element in $Y$ is $\Delta_f$, if there exists $\varepsilon>0$ such that we can solve Max-$k$-Cover in time 
    \begin{itemize}
     \item $\bigO\Big(\big(\min\{n,\Delta_f \cdot \min\{\sqrt u, \Delta_s\}\}\big)^{\omega/3 k(1-\varepsilon)}\Big)$, then $k$-Clique Hypothesis is false.
        \item $\bigO\Big(\big(\min\{n,\Delta_f\cdot \min\{u^{1/h}, \Delta_s^{1/(h-1)}\}\}\big)^{k(1-\varepsilon)}\Big)$ for $h\geq 3$ then $h$-Uniform Hyperclique Hypothesis is false.
       \item $\bigO\Big(\big(\min\{n,\Delta_f\cdot \min\{u^{1/k}, \Delta_s^{1/(k-1)}\}\}\big)^{k(1-\varepsilon)}\Big)$  then the $k$-OV Hypothesis is false.
    \end{itemize}
\end{theorem}
\begin{theorem}\label{theorem-partial-dom-LB}
    Given a graph $G$ with $n$ vertices and maximum degree $\Delta$, if there exists $\varepsilon>0$ such that we can solve Partial $k$-Dominating Set in time 
    \begin{itemize}
        \item $\min\{n,\Delta^2\}^{\omega/3 k(1-\varepsilon)}$, then $k$-Clique Hypothesis is false.
        \item $\min\{n,\Delta^{h/(h-1)}\}^{k(1-\varepsilon)}$ for $h\geq 3$ then the $h$-Uniform Hyperclique Hypothesis is false.
        \item $\min\{n,\Delta^{k/(k-1)}\}^{k(1-\varepsilon)}$ then the $k$-OV Hypothesis is false.
    \end{itemize}
\end{theorem}

Before we prove these theorems, let us highlight some interesting aspects of the underlying reductions: We achieve them by a single core  reduction from $(k,h)$-minIP/maxIP which we instantiate with different values for $h\in \{2,3, \dots, k\}$. This reduction is enabled by ensuring a strong regularity property in the given $(k,h)$-minIP/maxIP instances -- interestingly, such a strong regularity property can be a achieved in a simpler way for $(k,h)$-minIP/maxIP than for $h$-uniform hyperclique~\cite{FischerKRS25}; we circumvent the use of this result, which would have given an alternative, more complicated approach.

Let us formally introduce our notion of regularity of instances: 
We say that the sets $V_1,\dots, V_k$ are \emph{regular}, if for every vector $v_i\in V_i$, the set of coordinates~$y$ such that $i$ is an active index for $y$ and $v_i[y] = 1$ has the same size.
More generally, for any $r\leq h$ we say that the sets $V_1,\dots, V_k$ are \emph{$r$-regular}, if for every $r$-tuple $v_{i_1}\in V_{i_1},\dots, v_{i_r}\in V_{i_r}$ the product $v_{i_1}\odot \dots \odot v_{i_r}$ is the same.
By applying simple combinatorial gadgets, we now prove that we can without loss of generality assume that any given instance is $r$-regular for every $r\leq h-1$. This turns out to be an extremely useful property to have when constructing the reductions that we need to show hardness of Partial $k$-Dominating Set and Max $k$-Cover.
\begin{lemma}\label{lemma:r-regularization-khov}
    Given sets consisting of $n$ $d$ dimensional vectors $A_1,\dots, A_k$ and an integer $r<h$, one can construct the corresponding sets $A_1',\dots, A_k'$ of dimensions $f(k,r)\cdot d$ for some computable function $f$, such that
    \begin{itemize}
        \item $A_1',\dots, A_k'$ are $r$-regular.
        \item For any $r< s\leq h$ and any $a_{i_1}\in A_{i_1}, \dots, a_{i_s}\in A_{i_s}$ (for pairwise distinct $i_j$) it holds that $a_{i_1}\odot\dots \odot a_{i_s} = a'_{i_1}\odot\dots\odot a'_{s}$ (for vectors $a'_{i_j}\in A'_{i_j}$ corresponding to $a_{i_j}$).
    \end{itemize}
\end{lemma}
\begin{proof}
    Fix an arbitrary subset $S = \{s_1,\dots, s_r\}$ of $[k]$. 
    For each vector $a_i$ in $A_i$ with $i\in [k]\setminus S$ let the corresponding vector $a_i'\in A_i'$ be obtained by concatenating an all zero vector of size $(2^r-1)d$ to $a_i$.
    We construct the remaining vectors as follows.
    For each binary string $b$ of length $r$ define
    \[
    e(a_{s_j}, b) := 
    \begin{cases}
     a_{s_j}, & \text{if $b[j]=0$} \\
     \overline{a_{s_j}}, & \text{if $b[j]=1$}
    \end{cases}
    \]
    and let $a_{s_j}'$ be obtained by concatenating the vectors $e(a_{s_j}, 0\dots00), e(a_{s_j}, 0\dots01), \dots, e(a_{s_j}, 1\dots11)$.
    
    Observe that for every $r$-tuple of vectors $a_{s_1}\in A_{s_1}, \dots, a_{s_r}\in A_{s_1}$ and for every $j\in [d]$ there is a unique binary string $b$ such that $e(a_{s_1},b)[j] = \dots = e(a_{s_r},b)[j] = 1$.
    Note that we have added $(2^r-1)d$ many new coordinates to each vector. Associate the active indices to each of the added coordinates to contain $s_1,\dots, s_r$ and assign the remaining $(h-r)$ active indices arbitrarily.
    It is easy to verify now that each vector $a'_{s_1}\in A_{s_1},\dots, a'_{s_r}\in A_{s_r}$ satisfies $a'_{s_1}\odot \dots \odot a'_{s_r} = d$, while the product of any other $r$-tuple of vectors stays the same as before adding the new coordinates.
    Finally, repeating this process for every subset $S \subseteq [k]$ of size $r$ gives the $r$-regularity of $A'_1,\dots A'_k$.

    We further observe that if we take any set of $(r+1)$ vectors $a'_{s_1}\in A'_{s_1},\dots, a'_{s_{r+1}}\in A'_{s_{r+1}}$ (for pairwise distinct $s_1,\dots, s_{r+1}$), in each entry $j>d$ there is at least one $i$ such that $a'_{s_i}[j] = 0$.
    In particular, $a'_{s_1}\odot\dots\odot a'_{s_{r+1}} = a_{s_1}\odot\dots\odot a_{s_{r+1}}$.
    We also note that the dimension of every vector $a'_i$ is at most $2^{r}\cdot\binom{k}{r} d$.
\end{proof}
We can now apply the construction from the lemma above to get the desired regularization.
\begin{lemma}\label{lemma:khov-regularity}
    Let $V_1,\dots, V_k\subseteq \{0,1\}^d$, with each coordinate $y\in [d]$ associated to active indices $i_1,\dots, i_h\in [k]$ as above.
    We can construct the corresponding sets $V_1',\dots, V_k'\subseteq \{0,1\}^{d'}$ such that each vector $v_i\in V_i$ corresponds to a unique vector $v_i'\in V_i'$ and vice versa and the following conditions are satisfied.
    \begin{itemize}
        \item $d' = f(k)d$ for a computable function $f$.
        \item $V'_1,\dots, V'_k$ are $r$-regular for every $1\leq r<h$.
        \item For any pairwise distinct indices $i_1,\dots i_h \in [k]$ and the vectors $v_{i_1}\in V_{i_1},\dots, v_{i_h}\in V_{i_h}$, the product $v_{i_1}\odot\dots \odot v_{i_h} = v'_{i_1}\odot \dots \odot v'_{i_h}$ (for vectors $v'_{i_j}\in V'_{i_j}$ corresponding to $v_{i_j}\in V_{i_j}$).
    \end{itemize}
\end{lemma}
\begin{proof}
        Let $f_t(k) = 2^{t}\binom{k}{t}$ for any $1\leq t\leq k$.
    We apply the construction of the previous lemma on sets $A_1,\dots A_k$ and $r=k-1$ to obtain $k-1$-regular sets $A_1^{(k-1)},\dots, A_k^{(k-1)}$ of vectors of dimensions at most $2^{k-1}kd = f_{k-1}(k)\cdot d$ that satisfy the second constraint.
    Then, apply the same construction recursively on sets $A_1^{(k-1)},\dots, A_k^{(k-1)}$ and $r=k-2$, to obtain $(k-1)$-regular, $(k-2)$-regular sets $A_1^{(k-2)},\dots, A_k^{(k-2)}$ of vectors of dimensions at most $2^{k-2}\binom{k}{k-2}f_{k-1}(k)\cdot d = f_{k-2}(k)f_{k-1}(k) \cdot d$.
    
    Proceed recursively and set $A'_i = A^{(1)}_i$.
    The dimension of vectors in $A'_i$ are $d\cdot \prod_{i=1}^{k-1}f_i(k) = f(k)\cdot d$.
    Clearly the conditions of the lemma are satisfied.
\end{proof}
It is well known that an algorithm solving $k$-minIP/maxIP in time $O(n^{k-\varepsilon})$ would refute $k$-OV hypothesis (this is trivial for $k$-minIP; for $k$-maxIP, see e.g. \cite{KarppaKK18} where a proof for $k=2$ is given). 
For completeness, we adapt this approach to show that an efficient algorithm solving $(k,h)$-maxIP would imply an efficient algorithm for $(k,h)$-OV, and remark that the reduction from $(k,h)$-OV to $(k,h)$-minIP is trivial.
\begin{lemma}\label{lemma:maxip-hardness}
    Let $A_1,\dots,A_k\subseteq \{0,1\}^d$ be given sets each consisting of $n$ $d$ dimensional binary vectors, together with the set of $h$ associated active indices for each coordinate $y\in [d]$.
    Let $d' = 2^h\binom{k}{h}\cdot d$.
    We can construct an instance $A_1',\dots, A_k'\subseteq \{0,1\}^{d'}$ of size $n$ and an integer $t$ such that there are vectors $a_1\in A_1,\dots,a_k\in A_k$ satisfying $a_1\cdot\dots\cdot a_k = 0$ if and only if there are vectors $a'_1\in A'_1,\dots,a'_k\in A'_k$ satisfying $a'_1\cdot\dots\cdot a'_k \geq t$.
\end{lemma}
\begin{proof}
    By applying the exact same construction from Lemma \ref{lemma:r-regularization-khov}, but plugging in $r = h$, we get $h$-regular instance $A_1',\dots, A_k'\subseteq \{0,1\}^{d'}$ (more precisely, the product of each $h$ vectors from pairwise distinct sets is equal to $d$).
    We now iterate through each vector $a'_i\in A'_i$ (for every $i\in [k]$) and set $a'_i[y]=0$ for each $y\in [d]$. 
    By doing so, we notice that in the product of any $h$ vectors $a'_{i_1}\in A'_{i_1},\dots, a'_{i_h}\in A'_{i_h}$ from pairwise distinct sets we lose exactly the contribution of the product from the corresponding vectors $a_{i_1}\in A_{i_i},\dots, a_{i_h}\in A_{i_h}$.
    In particular, this gives for any $a'_{i_1}\in A'_{i_1},\dots, a'_{i_h}\in A'_{i_h}$
    \[
        a'_{i_1}\odot\dots\odot a'_{i_h} = d - a_{i_1}\odot\dots\odot a_{i_h},
    \]
    for the vectors $a_{i_1}\dots a_{i_h}$ corresponding to $a'_{i_1}\dots a'_{i_h}$.
    By setting $t:=\binom{k}{h}d$, the desired follows directly. 
\end{proof}

We now proceed to show that a significant improvement to any of our algorithms from the previous section would yield a significant improvement to the one of the $(k,h)$-minIP, or $(k,h)$-maxIP. 
In order to do that, we distinguish between two cases depending on the parity of $h$.
Namely, if $h$ is odd, we reduce from $(k,h)$-maxIP, whereas if $h$ is even, we reduce from $(k,h)$-minIP.
In particular, we will first provide a general reduction framework that will be sufficient to cover both reductions and then we will verify the details of the two reductions separately.
\begin{lemma} \label{lemma:max-k-cover-lb-reduction}
    Let $2\leq h\leq k$ be fixed integers and $n,u,\Delta_s\leq u, \Delta_f\leq n$ be given positive integers. Let $A_1,\dots, A_k\subseteq \{0,1\}^d$ be sets consisting of $\min\{n,\Delta_f\cdot \min\{u^{1/h}, \Delta_s^{1/(h-1)}\}\}$ many $d$-dimensional binary vectors \footnote{$d = |A_i|^{\delta}$, for any $\delta>0$.}, with each coordinate $y\in [d]$ associated to $h$ active indices $i_1,\dots, i_h\in [k]$.
    We can construct a bipartite graph $G = (X\cup Y,E)$ satisfying the following conditions:
    \begin{itemize}
        \item $X$ consists of at most $\bigO(\min\{n,\Delta_f\cdot \min\{u^{1/h}, \Delta_s^{1/(h-1)}\}\})$ many vertices and for every $x\in X$ it holds that $\deg(x)\leq \bigO(\Delta_s\cdot d)$.
        \item $Y$ consists of at most $\bigO(u\cdot d)$ many vertices and for every $y\in Y$ it holds that $\deg(y)\leq \bigO(\Delta_f)$.
        \item We can compute positive integers $t,\alpha$ such that $X$ contains $k$ vertices $x_1,\dots, x_k$ satisfying $|N(x_1)\cup \dots \cup N(x_k)|\geq t$ if and only if there are vectors $a_1\in A_1,\dots, a_k\in A_k$ satisfying $a_1\cdot\dots\cdot a_k\geq \alpha$ if $h$ is odd (reduction from $(k,h)$-maxIP), and if $h$ is even, $|N(x_1)\cup \dots \cup N(x_k)|\geq t$ if and only if $a_1\cdot\dots\cdot a_k \leq \alpha$ (reduction from $(k,h)$-minIP).
        \item $G$ can be constructed deterministically in time $\bigO((|X|\Delta_s + |Y|\Delta_f)\cdot d)$.
    \end{itemize}
\end{lemma}
\begin{proof}
    By \autoref{lemma:khov-regularity}, we can assume without loss of generality that $A_1,\dots, A_k$ are $r$-regular sets for each $r<h$. 
    Moreover, we can assume without loss of generality that for each $h$-tuple of pairwise distinct indices $i_1,\dots, i_h \in [k]$ the set $|a(i_1,\dots, i_h)| = d'$ (for $d' = \bigO(d)$).
    Let $s$ be a positive integer that we will fix later and 
    \begin{align*}
        & X =  X_1\cup \dots \cup X_k  \\
        & Y = \bigcup_{\substack{1\leq i_1<\dots<i_h\leq k \\ g_1,\dots, g_h \in [s]}} D_{i_1\dots i_h}^{(g_1\dots g_h)} \cup \bigcup_{\substack{i\in [k] \\ 1\leq j<\ell \leq s}} P^{i}_{j,\ell} 
    \end{align*}
    where each $X_i$ corresponds to $A_i$, each $D_{i_1\dots i_h}^{g_1\dots g_h}$ consists of a copy of $d'$ vertices corresponding to coordinates in $[d]$, for which the active indices are $i_1,\dots, i_h$, and each $P^{i}_{j,\ell}$ corresponds to a copy of $[100k^{k}ds^{h-2}]$.

    We divide each $X_i$ uniformly into $s$ many groups labeled $1,\dots, s$, each group containing at most $\Delta_f$ many vertices. 
    Add an edge between a vertex $x_{i_j}\in X_{i_j}$ (for any $1\leq j\leq h$) and a vertex $t\in D_{i_1\dots i_h}^{(g_1\dots g_h)}$ if and only if the following conditions are satisfied:
    \begin{itemize}
        \item $x_{i_j}$ is in the group labelled $g_j$ in $X_{i_j}$. 
        \item The vector $a_{i_j}$ corresponding to $x_{i_j}$ has the entry $1$ in the coordinate corresponding to $t$.
    \end{itemize}
    Intuitively, the $i_j$'s indicate which indices are active for the given coordinate, while the $g_j$'s ensure we only add edges from one group per active set (thus controlling the maximal degree of vertices in $D_{i_1\dots i_h}^{(g_1\dots g_h)}$).
    Finally, we add an edge between $x_i\in X_i$ and $p\in P^i_{j,\ell}$ if and only if $x_i$ is contained in one of the groups $j,\ell$ in $X_i$.
    Intuitively, this gadget penalizes the selection of vertices from the same set in the solution. 
    There are no other edges in $G$.

    We fix $s:=|X|/\Delta_f$ and proceed to show that with this choice for $s$, the bounds for the number of vertices in $G$ and the maximal degrees are satisfied. 
    We first notice that $|X| = k\cdot |A_1| = \bigO(n)$, so we proceed to count the vertices in $Y$. To this end, we count the vertices in $\bigcup_{\substack{1\leq i_1<\dots<i_h\leq k \\ g_1,\dots, g_h \in [s]}} D_{i_1\dots i_h}^{(g_1\dots g_h)}$ and in  $\bigcup_{\substack{i\in [k] \\ j,\ell \in [s]}} P^{i}_{j,\ell}$ separately. 
    Note that we have at most $k^{h} = \bigO(1)$ choices for indices $i_1,\dots, i_h$ and at most $s^h \leq u$ many choices for the indices $g_1,\dots, g_h$. 
    Since each set $D_{i_1\dots i_h}^{(g_1\dots g_h)}$ has $\bigO(d)$ many vertices, thus :
    \[\Big|\bigcup_{\substack{1\leq i_1<\dots<i_h\leq k \\ g_1,\dots, g_h \in [s]}} D_{i_1\dots i_h}^{(g_1\dots g_h)}\Big| \leq \bigO(du) .\]
    On the other hand, we have a total of $k\cdot s^{2} = \bigO(s^2)$ many sets 
    $P^{i}_{j,\ell}$, and each contains $\bigO(ds^{h-2})$ many vertices. Hence we get:
    \[\Big|\bigcup_{\substack{i\in [k] \\ j,\ell \in [s]}} P^{i}_{j,\ell}\Big|\leq \bigO(ds^h) \leq \bigO(du).\]
    It remains to argue that the maximal degree conditions are satisfied. 
    Consider a vertex $x_{i}$ that is contained in the group labeled $j$ in set $X_i$. 
    By construction, $x_i$ is adjacent to all vertices in $P_{j,\ell}^i$, where $\ell\in [s]$ is arbitrary. 
    In total, this gives us $s$ sets, each consisting of $\bigO(ds^{h-2})$ many vertices, bounding a total neighborhood size of $x_i$ in $\bigcup P_{j,\ell}^i$ by at most $\bigO(ds^{h-1}) \leq \bigO(d\Delta_s)$. 
    On the other hand, $x_i$ is adjacent to at most $\bigO(d)$ many vertices in each $D_{i_1,\dots, i_h}^{(g_1,\dots, g_h)}$ with $i = i_r$ for some $1\leq r\leq h$ and the corresponding index $g_r = j$.
    There are at most $k^{h-1} = \bigO(1)$ choices for the indices $i_1,\dots, i_h$ and at most $s^{h-1} \leq \bigO(\Delta_s)$ many choices for the indices $g_1,\dots, g_h$. 
    This allows us to bound the degree of $x_i$ as 
    \begin{align*}
    \deg(x_i) &\leq \Big|N(x_i)\cap \big(\bigcup_{\substack{i\in [k] \\ j_1\dots j_{h-1} \in [s]}} P^{i}_{j,\ell}\big) \Big| + \Big|N(x_i)\cap \big(\bigcup_{\substack{1\leq i_1<\dots<i_h\leq k \\ g_1,\dots, g_h \in [s]}} D_{i_1\dots i_h}^{(g_1\dots g_h)}\big) \Big| \\
    & \leq \bigO(d\Delta_s + d\Delta_s) = \bigO(d\Delta_s).
    \end{align*}
    On the other hand, any vertex $y\in P^{i}_{j,\ell}$ is only adjacent to at most $\bigO(\Delta_f)$ many vertices in groups labelled $j$ and $\ell$ of set $X_i$.
    Similarly, any $y\in D_{i_1\dots i_h}^{(g_1\dots g_h)}$ is adjacent only to the $\bigO(\Delta_f)$ in group $g_j$ of set $i_j$ for all $j\in [h]$. This allows us to bound $\deg(y)\leq \bigO(\Delta_f)$ for any $y\in Y$. 
    
    We now proceed to show that any choice of vertices $x_1,\dots, x_k \in V(G)$ that maximizes the value $|N(x_1)\cup \dots \cup N(x_k)|$ satisfies (without loss of generality) $x_1\in X_1,\dots, x_k\in X_k$. We do so by assuming we are given any subset $S$ of $X$ such that no vertex from $X_i$ is contained in $S$, and showing that we can replace some vertex from $S$ by any vertex from $X_i$ so that the number of dominated vertices in $Y$ increases.
    \begin{claim}
        Let $S\subseteq X$ be a set containing $k$ vertices from $X$, such that for some $i\in [k]$, the set $X_i\cap S$ is empty. 
        Then there exists a vertex $b\in S$, such that for any vertex $x_i\in X_i$ it holds that $|N(S)| \leq |N(S\setminus\{b\}\cup \{x_i\})|$.
    \end{claim}
    \begin{proof}
        We first observe that since $S\cap X_i$ is empty, there exists some $\alpha\neq i$ such that $|S\cap X_\alpha| \geq 2$. 
        Let $x_\alpha,x'_\alpha$ be two vertices contained in $S\cap X_\alpha$, and $x_i$ be any vertex from $X_i$.
        We can observe that there is at least one set $P^\alpha_{j,\ell}$ such that $P^\alpha_{j,\ell}\subseteq N(x_\alpha)\cap N(x_{\alpha}')$.
        Hence, we get the following
        \begin{equation} \label{eq:remove-x_alpha-from-S}
            \begin{split}
                |N(S\setminus\{x_{\alpha}'\})| &\geq |N(S)| - \deg(x_{\alpha}') + |P^\alpha_{j,\ell}| \\
                & \geq |N(S)| - \deg(x_{\alpha}') + 100k^kds^{h-2}
            \end{split}
        \end{equation}
        On the other hand, by adding any vertex $x_i\in X_i$ to the set $S\setminus \{x_{\alpha}\}$, we get
        \begin{align*}
        \mathmakebox[\mathindent][l]{
        |N(S\setminus\{x_{\alpha}'\}\cup\{x_i\})| }\\
        & \geq |N(S\setminus\{x_{\alpha}'\})| + \deg(x_i) - \sum_{y\in S\setminus x_{\alpha}}|N(y) \cap N(x_i)| & \text{(by I.E. principle)}\\
            & \geq \big(|N(S)| - \deg(x_{\alpha}') + 100k^kds^{h-2}) + \deg(x_i) - \sum_{y\in S\setminus x_{\alpha}}|N(y) \cap N(x_i)| & \text{(\autoref{eq:remove-x_alpha-from-S})} \\ 
            & = |N(S)| + 100k^kds^{h-2} - \sum_{y\in S\setminus \{x_{\alpha'}\}}|N(y) \cap N(x_i)| & (\deg(x_i) = \deg(x'_\alpha))\\ 
            & \geq  |N(S)| + 100k^kds^{h-2} - k^{h-1}s^{h-2}d' & \\
            & \geq |N(S)|, &
        \end{align*}
        where second to last inequality follows by observing that for any pair of vertices $x_i\in X_i, x_j\in X_j$, where $i,j$ are distinct and $x_i$, $x_j$ are from groups $\ell_1,\ell_2$ respectively, the set $N(x_i)\cap N(x_j)$ contains only vertices from $D_{i_1\dots i_h}^{(g_1\dots g_h)}$ where $i,j\in \{i_1,\dots, i_h\}$ and the corresponding group labels are $\ell_1, \ell_2$. Observe that there are at most $s^{h-2}$ many choices for the remaining indices $\{i_1,\dots i_h\}-\{i,j\}$ and at most $k^{h-2}$ many choices for the remaining group indices.
        Recall that each set $D_{i_1\dots i_h}^{(g_1\dots g_h)}$ has $d'$ many vertices.
        Exploiting the assumption that there are no vertices from $X_i$ in $S$ and plugging in this bound, we get $\sum_{y\in S\setminus \{x_{\alpha'}\}}|N(y) \cap N(x_i)|\leq \sum_{y\in S\setminus \{x_{\alpha'}\}}k^{h-2}s^{h-2}d' = k^{h-1}s^{h-2}d'$.
    \end{proof}
    Now fix a set $S\subseteq X$ of size $k$, such that for any other subset of size $k$ $S'\subseteq X$ it holds that $|N(S)|\geq |N(S')|$.
    By the claim above, we can assume without loss of generality that $S = \{x_1,\dots, x_k\}$ and $S\cap X_i = \{x_i\}$ for each $i\in [k]$.
    We proceed to count the number of vertices in $N(S)$.
    We can observe that since $S$ contains exactly one vertex from each set, for every pair $x_i, x_j\in S$ for $i\neq j$ it holds that $N(x_i)\cap N(x_j) \subseteq \bigcup_{\substack{1\leq i_1<\dots<i_h\leq k \\ g_1,\dots, g_h \in [s]}} D_{i_1\dots i_h}^{(g_1\dots g_h)}$.
    Moreover, for every $(h+1)$ tuple of vertices from $S$ $x_{i_1},\dots, x_{i_{h+1}}$ it holds that $N(x_{i_1})\cap \dots \cap N(x_{i_{h+1}}) = \emptyset$.
    Applying the principle of inclusion-exclusion, we get 
    \begin{align*}
        |N(S)| &= \sum_{r\in [k]} (-1)^{r+1} \cdot \sum_{1\leq i_1<\dots < i_r\leq k} |N(x_{i_1}) \cap\dots \cap N(x_{i_r})| &\\
        & = \sum_{r\in [h]} (-1)^{r+1} \cdot \sum_{1\leq i_1<\dots < i_r\leq k} |N(x_{i_1}) \cap\dots \cap N(x_{i_r})| 
    \end{align*}
    We now observe that since sets $V_1,\dots, V_k$ are $r$-regular for every $r<h$, it holds that
    \[\sum_{r\in [h-1]} (-1)^{r+1} \cdot \sum_{1\leq i_1<\dots < i_r\leq k} |N(x_{i_1}) \cap\dots \cap N(x_{i_r})| = \sum_{r\in [h-1]} (-1)^{r+1} \cdot \sum_{1\leq i_1<\dots < i_r\leq k} |N(x'_{i_1}) \cap\dots \cap N(x'_{i_r})|,\]
    for any $x'_1\in X_1,\dots, x'_k\in X_k$.
    Hence, the total number of vertices dominated by $S$ depends only on the value $\sum_{1\leq i_1<\dots < i_h\leq k} |N(x_{i_1}) \cap\dots \cap N(x_{i_h})|$.
    In particular, if we fix $t:=\sum_{r\in [h-1]} (-1)^{r+1} \cdot \sum_{1\leq i_1<\dots < i_r\leq k} |N(x_{i_1}) \cap\dots \cap N(x_{i_r})|$, we get 
    \[
    |N(S)| = t + (-1)^{h+1} \cdot \sum_{1\leq i_1<\dots < i_h\leq k} |N(x_{i_1}) \cap\dots \cap N(x_{i_h})|.
    \]
    Finally, by noticing that $|N(x_{i_1}) \cap\dots \cap N(x_{i_h})|>0$ if and only if the corresponding vectors satisfy $v_{i_1}\odot \dots \odot v_{i_h}>0$, we can conclude that if $h$ is even, 
    the set $S = \{x_1,\dots, x_k\}$ that maximizes the value of $|N(S)|$ exactly corresponds to the set of vectors $\{v_1,\dots, v_k\}$ that minimizes the value $v_1\cdot \dots \cdot v_k$. That is, $|N(S)|\geq t-\alpha$ if and only if $v_1\cdot \dots \cdot v_k\leq \alpha$, for any non-negative integer $\alpha$, giving us a reduction from $(k,h)$-minIP.
    Similarly if $h$ is odd, the set $S = \{x_1,\dots, x_k\}$ that maximizes the value of $|N(S)|$ exactly corresponds to the set of vectors $\{v_1,\dots, v_k\}$ that maximizes the value $v_1\cdot \dots \cdot v_k$. That is, $|N(S)|\geq t+\alpha$ if and only if $v_1\cdot \dots \cdot v_k\geq \alpha$ for any non-negative integer $\alpha$, giving us a reduction from $(k,h)$-maxIP.
\end{proof}
By using the last lemma, we can now show that a fast algorithm for Max-$k$-Cover would give us a fast algorithm for $(k,h)$-minIP (resp. $(k,h)$-maxIP).
Formally, we state and prove this property below.
\begin{lemma}\label{lemma:max-k-cover-lb-khov}
    For any fixed $k\geq 2$, $2\leq h \leq k$, there exists a bipartite graph $G = (X\cup Y,E)$ with $|X|=n$, $|Y| = u$, $\max_{x\in X}\deg(x) = \Delta_s$, $\max_{y\in Y}\deg(y) = \Delta_f$ such that the following holds. Let $N_h:=\min\{n,\Delta_f\cdot \min\{u^{1/h}, \Delta_s^{1/(h-1)}\}\}$ then:
    \begin{enumerate}
        \item If $h\geq 3$ and there is an algorithm solving Max-$k$-Cover on $G$ in time $\bigO(N_h^{k(1-\varepsilon)})$ for some $\varepsilon>0$, then there exists a $\delta>0$, such that we can solve any $(k,h)$-minIP instance (if $h$ is even, otherwise any $(k,h)$-maxIP instance) $A_1,\dots, A_k$ with $|A_1| = \dots = |A_k| = N_h$ of dimensions $d=N_h^\delta$ in time $\bigO(N_h^{k(1-\varepsilon')})$ for some $\varepsilon'>0$.
        \item If there exists an algorithm solving Max-$k$-Cover on $G$ in time $\bigO(N_2^{k\omega/3(1-\varepsilon)})$ for some $\varepsilon>0$, then there exists a $\delta>0$ such that we can solve any $(k,2)$-minIP (maxIP) instance $A_1,\dots, A_k$ with $|A_1| = \dots = |A_k| = N_2$ of dimensions $d=N_2^\delta$ in time $\bigO(N_2^{\omega k/3(1-\varepsilon')})$ for some $\varepsilon'>0$.
    \end{enumerate}
\end{lemma}
\begin{proof}
    We prove the first item, and the proof for the second item follows the similar lines.
    Let $\varepsilon>0$ be arbitrary and fix $\delta = \frac{\varepsilon}{2k}$.
    For fixed $k\geq h\geq 3$, let $A_1,\dots, A_k$ be a given instance of $(k,h)$-OV (if $h$ is even, otherwise $(k,h)$-maxIP) with $|A_1| = \dots = |A_k| = N_h$ and $d = N_h^\delta$, and let $G$ be the graph as in the proof of Lemma \ref{lemma:max-k-cover-lb-reduction}.
    Recall that $G$ is a bipartite graph with parts $(X,Y)$ and $\bigO(N_h)$ vertices in $X$, at most $\bigO(ud)$ vertices in $Y$, such that any $x\in X$ satisfies $\deg(x)\leq d\Delta_s$, similarly any $y\in Y$ satisfies $\deg(y)\leq \Delta_f$.

    Assume that there is an algorithm solving Max-$k$-Cover on any graph $G$ in time \[\bigO\big(\min\{n,\Delta_f\cdot \min\{(du)^{1/h}, (d\Delta_s)^{1/(h-1)}\}\}^{k-\varepsilon}\big).\]
    From the instance $A_1,\dots, A_k$, we could in linear time construct $G$, run this algorithm on $G$ and applying Lemma \ref{lemma:max-k-cover-lb-reduction}, this would give us an answer of the original instance. 
    This yields an algorithm solving $(k,h)$-minIP (i.e. $(k,h)$-maxIP) in time
    \begin{align*}
    \mathmakebox[4\mathindent][l]{
        \bigO\big(\min\{n,\Delta_f\cdot \min\{(du)^{1/h}, (d\Delta_s)^{1/(h-1)}\}\}^{k-\varepsilon}\big) }\\
        & = \bigO\big(\min\{n,\Delta_f\cdot \min\{(N_h^\delta u)^{1/h}, (N_{h}^\delta\Delta_s)^{1/(h-1)}\}\}^{k-\varepsilon}\big) \\
        & \leq  \bigO\big(N_{h}^{\delta k}\min\{n,\Delta_f\cdot \min\{u^{1/h}, \Delta_s^{1/(h-1)}\}\}^{k-\varepsilon}\big) \\
        &\leq \bigO\big(N_{h}^{\varepsilon/2}\min\{n,\Delta_f\cdot \min\{u^{1/h}, \Delta_s^{1/(h-1)}\}\}^{k-\varepsilon}\big) \\ 
        & = \bigO\big(N_{h}^{k-\varepsilon/2}\big)
    \end{align*}
\end{proof}
The previous two lemmas show that if our algorithms for Max-$k$-Cover could be significantly improved, then for some $h\geq 2$, we would also obtain a significant improvement over $(k,h)$-minIP ($(k,h)$-maxIP).
In order to prove Theorem \ref{theorem:max-k-cover-LB}, it remains to show that any such improvement for $(k,h)$-minIP ($(k,h)$-maxIP) would refute the corresponding hardness assumption depending on the value of $h$. We remark that it is sufficient to show the hardness for $(k,h)$-OV, since by \ref{lemma:maxip-hardness}, this implies the hardness for both optimization variants, $(k,h)$-minIP and $(k,h)$-maxIP as well.
\begin{lemma}\label{lemma:hyperclique-hardness-of-khov}
    Let $k$, $h$ be fixed positive integers such that $2\leq h\leq k$ and $\delta>0$ be arbitrary.
    For every sufficiently large positive integer $n$, there exists a $(k,h)$-OV instance $A_1,\dots, A_k$ with $|A_1| = \dots = |A_k| = n$, and $d = n^\delta$ such that
    \begin{enumerate}
        \item If $h = k$, and for some $\varepsilon>0$ there exists an algorithm solving $A_1,\dots, A_k$ in time $\bigO(n^{k-\varepsilon})$, then the $k$-OV Hypothesis is false.
        \item \label{item:hyperclique-hardness} If $h\geq 3$, and for some $\varepsilon>0$ there exists an algorithm solving $A_1,\dots, A_k$ in time $\bigO(n^{k-\varepsilon})$, then the $h$-Uniform Hyperclique Hypothesis is false.
        \item \label{item:clique-hardness} If $h=2$, and for some $\varepsilon>0$ there exists an algorithm solving $A_1,\dots, A_k$ in time $n^{\omega k/3 - \varepsilon}$, then the $k$-Clique Hypothesis is false.
    \end{enumerate}
\end{lemma}
\begin{proof}
    We first remark that by setting $h=k$, by definition, the $(k,h)$-OV problem is exactly equivalent to $k$-OV problem and an algorithm running in $\bigO(n^{k-\varepsilon})$ for this problem would refute $k$-OVH.
    Thus, it remains to prove the remaining two statements.
    We prove \autoref{item:hyperclique-hardness} and note that the proof for \autoref{item:clique-hardness} is analogous.

    Let $q$ be a positive integer whose value we will fix later.
    We reduce from $h$-Uniform $(qk)$-Hyperclique Detection problem in $(qk)$-Partite graph.
    Let $G = (X_1,\dots, X_{qk}, E)$ be a $(qk)$-partite $h$-uniform hypergraph with $|X_i| = n$ for each $i\in [qk]$.
    Let $\overline E$ denote the set of non-edges in $G$, i.e.
    \[
    \overline{E} := \big\{e\in \binom{V(G)}{h} \mid \forall i\in [k], e\cap X_i\leq 1 \big\} \setminus E
    \]
    Let $d = \bigO(n^h)$ and let $A_1,\dots, A_k$ be the sets consisting of $\bigO(n^q)$ $d$-dimensional vectors defined as follows.
    Each vector $a_i$ in $A_i$ corresponds to a collection of vertices $x_{1}\in X_{q(i-1)+1},\dots, x_q\in X_{qi}$. 
    By abuse of notation (for the sake of simplicity), we will identify each vector $a_i$ with the set $\{x_1,\dots, x_q\}$, so that we can use the set-theoretic notation directly on the elements $a_i$ (e.g. containment, union, etc.).    
    Let each coordinate $y$ correspond uniquely to an element $\{x_1,\dots, x_h\}$ of $\overline{E}$. For any vector $a_i\in A_i$, set the value of $a_i[y]$ to $0$ if 
    there exists an $a_i' \in A_i$ such that 
    $a_i'\cap \{x_1,\dots, x_h\}>a_i\cap \{x_1,\dots, x_h\}$
    , otherwise set $a_i[y]$ to $1$.
    We now proceed to associate to each coordinate $y\in [d]$ a set of active indices. 
    We associate to $y$ an index $i$ as active if the non-edge corresponding to $y$ contains a vertex from one of the sets $X_{q(i-1)+1},\dots, X_{qi}$.
    More precisely, if $\{x_1,\dots, x_h\}$ is the non-edge corresponding to $y$, and if $x_1\in X_{i_1},\dots, x_h\in X_{i_h}$, we associate to $y$ the set of active indices $\lfloor(i_1-1)/q \rfloor + 1,\dots, \lfloor(i_h-1)/q\rfloor +1$.
    Clearly, this will result in some coordinates being associated to less than $h$ indices.
    To those coordinates, we associate the remaining indices arbitrarily, until each coordinate $y$ has exactly $h$ active indices associated to it.

    We claim that there is an $h$-uniform $(qk)$-hyperclique $x_1\in X_1,\dots, x_{qk}\in X_{qk}$ if and only if there are vectors $a_1\in A_1,\dots, a_k\in A_k$ satisfying $a_1\cdot\dots\cdot a_k = 0$. 
    Assume first that there is an $h$-uniform $(qk)$-hyperclique $x_1\in X_1,\dots, x_{qk}\in X_{qk}$.
    We claim that vectors $a_1,\dots, a_k$ such that $a_i$ corresponds to $x_{q(i-1)+1}\in X_{q(i-1)+1},\dots, x_{qi}\in X_{qi}$ satisfy the orthogonality condition.
    It is sufficient to show that any $h$-tuple $a_{i_1},\dots, a_{i_h}$ satisfies $a_{i_1}\odot \dots \odot a_{i_h} = 0$.
    Assume for contradiction that $a_{i_1}\odot \dots \odot a_{i_h} \geq 1$. 
    This implies that there is a coordinate $y$ for which all of the coordinates $a_{i_1},\dots, a_{i_h}$ are active and $a_{i_1} = \dots = a_{i_h} = 1$. 
    Consider the non-edge $e$ corresponding to $y$. 
    By the construction of the active indices, there exist $a'_{i_1}\in A_{i_1},\dots, a'_{i_h}$ such that $e\subseteq a'_{i_1}\cup \dots \cup a'_{i_h}$.
    However, since $a_{i_1} = \dots = a_{i_h} = 1$, we have that $|a_{i_j}\cap e| \geq |a'_{i_j}\cap e|$ for each $j\in[h]$. Hence, $e\subseteq a_{i_1}\cup \dots \cup a_{i_h}$, implying that the vertices in $a_{i_1}\cup \dots \cup a_{i_h}$ do not form a hyperclique in $G$, which yields a contradiction.
    Conversely, assume that there are vectors $a_1\in A_1,\dots, a_k\in A_k$ such that $a_1\cdot\dots\cdot a_k = 0$. We claim that the vertices in $a_1\cup \dots\cup a_k$ form a hyperclique in $G$.
    To this end it is sufficient to show that no $e\in \overline{E}$ is contained in $a_1\cup \dots\cup a_k$.
    Assume for contradiction that there is a non-edge $e$ contained in $a_1\cup \dots\cup a_k$.
    Consider the coordinate $y$ which $e$ corresponds to and assume that $i_1,\dots, i_h$ are active indices associated to $y$. 
    If, for some $j\in [h]$, there exists $a'_{i_j}\in A_{i_j}$ such that $a'_{i_j}\cap e > a_{i_j}$, then clearly $e\not\subseteq a_1\cup \dots\cup a_k$, hence by our assumption that this containment holds, we can conclude that no such vertex exists, and in particular this implies that $a_{i_j}[y] = 1$ for each $j\in[h]$. However, this further implies $a_{i_1}\odot \dots \odot a_{i_h} \geq 1$ and moreover $a_1\cdot\dots\cdot a_k\geq 1$, contradicting the assumption that $a_1\cdot\dots\cdot a_k = 0$.
    We may thus conclude that no such non-edge exists and in particular, that the vertices in $a_1\cup \dots\cup a_k$ form a hyperclique in $G$.

    Assume that for some $\delta>0$ and $\varepsilon>0$, there exists an algorithm $\mathcal{A}$ solving any instance $A_1,\dots, A_k$ of $(k,h)$-OV, with $|A_i| = N$ and $d = N^\delta$ in time $N^{k-\varepsilon}$.
    Let $q = h/\delta$. Given a $(qk)$-partite $h$-uniform hypergraph $G$, construct the instance of $(k,h)$-OV instance $A_1,\dots, A_k$ as above and run $\mathcal{A}$ on the obtained instance.
    We have: $N:=|A_i|\leq n^{q}$ and $d\leq n^{h} = n^{q\delta} = N^{\delta}$.
    Hence, $\mathcal{A}$ solves this instance and consequently the original $h$-Uniform $(qk)$-Hyperclique Detection instance in time $N^{k-\varepsilon} = n^{qk-q\varepsilon}$, thus refuting the $h$-Uniform $k$-Hyperclique hypothesis.
\end{proof}
Theorem \ref{theorem:max-k-cover-LB} now follows directly by combining Lemma \ref{lemma:hyperclique-hardness-of-khov} with Lemma \ref{lemma:max-k-cover-lb-khov} and \ref{lemma:maxip-hardness}.
It remains to prove Theorem \ref{theorem-partial-dom-LB}.
To this end, we will reuse the reduction for Max-$k$-Cover and verify that even in the monochromatic instance, all of the desired properties get preserved.
\begin{lemma} \label{lemma:partial-domination-lb-reduction}
    Let $2\leq h\leq k$ be fixed integers and $n,\Delta\leq n$ be given positive integers. Let $A_1,\dots, A_k\subseteq \{0,1\}^d$ be sets consisting of $\min\{n, \Delta^{h/(h-1)}\}$ many $d$-dimensional binary vectors with $d = |A_i|^{\delta}$, for any $\delta>0$, and each coordinate $y\in [d]$ associated $h$ active indices $i_1,\dots, i_h\in [k]$.
    We can construct a graph $G = (X\cup Y,E)$ satisfying the following conditions:
    \begin{itemize}
        \item $G$ consists of at most $\bigO(\min\{n, \Delta^{h/(h-1)}\}d)$ many vertices and for every $x\in V(G)$, it holds that $\deg(x)\leq \bigO(\Delta\cdot d)$.
        \item We can compute positive integers $t,\alpha$ such that $G$ contains $k$ vertices $x_1,\dots, x_k$ satisfying $|N(x_1)\cup \dots \cup N(x_k)|\geq t$ if and only if there are vectors $a_1\in A_1,\dots, a_k\in A_k$ satisfying $a_1\cdot\dots\cdot a_k\geq \alpha$ if $h$ is odd (reduction from $(k,h)$-maxIP), and if $h$ is even, $|N(x_1)\cup \dots \cup N(x_k)|\geq t$ if and only if $a_1\cdot\dots\cdot a_k = 0$ (reduction from $(k,h)$-OV).
        \item $G$ can be constructed deterministically in time $\bigO(|\min\{n, \Delta^{h/(h-1)}\}|\Delta\cdot d)$.
    \end{itemize}
\end{lemma}
\begin{proof}
    Given a such instance $A_1,\dots, A_k$ construct the bipartite graph $(X\cup Y,E)$ as in the proof of Lemma \ref{lemma:max-k-cover-lb-reduction}, by setting $u=n$, $\Delta_f = \Delta_s = \Delta$.
    Moreover, we make one slight adjustment to the size of the sets $P^i_{j,\ell}$ and set $|P^i_{j,\ell}| = 100k^kd\max\{s^{h-2},\Delta/s\}$. It is easy to verify that this modification does not change anything in the proof of Lemma \ref{lemma:max-k-cover-lb-reduction} (there we would use the value $|P^i_{j,\ell}| = 100k^kd\max\{s^{h-2},\frac{\Delta_s}{s}\}$).
    We now need to verify that the sizes match our desired values. 
    In particular, we obtain a graph that has $\bigO(\min\{n,\Delta_f\cdot \min\{u^{1/h}, \Delta_s^{1/(h-1)}\}\}d)$ many vertices\footnote{The statement of the lemma gives a crude upper bound for the set $Y$, however by carefully examining the proof, one can see that we actually get this upper bound.}. We can rewrite this value in terms of our parameters as
    \[
        \min\{n,\Delta_f\cdot \min\{u^{1/h}, \Delta_s^{1/(h-1)}\}\}d = \min\{n, n^{1/h}\Delta, \Delta^{h/(h-1)}\}d.
    \]
    We now only have to verify that the term $\Delta n^{1/h}$ vanishes. Indeed, if $n^{1/h}\geq \Delta^{1/(h-1)}$, then clearly $\min\{n^{1/h}\Delta, \Delta^{h/(h-1)}\} = \Delta^{h/(h-1)}$.
    On the other hand, if $n^{1/h}\leq \Delta^{1/(h-1)}$, then equivalently $\Delta\geq n^{(h-1)/h}$, and hence $\Delta n^{1/h} \geq n^{(h-1)/h} n ^{1/h} = n$, i.e. $\min\{n, n^{1/h}\Delta\} = n$.
    This gives us the proof of the first statement. It remains to prove the second one. We can observe that if we have an optimal solution $S$ of size $k$, such that $S\cap Y$ is empty, the result follows by applying Lemma \ref{lemma:max-k-cover-lb-reduction}.
    It is thus sufficient to prove that given a set $S\subseteq V(G)$ of size $k$, such that there is a vertex $y\in S\cap Y$, we can replace $y$ by some vertex $x\in X$ such that $|N(S - \{y\} \cup \{x\})| \geq |N(S)|$.
    Indeed, any vertex $y\in Y$ is adjacent to at most $h\Delta$ many vertices, hence we obtain:
    \[
        |N(S - \{y\})|\geq |N(S)| - |N(y)|\geq |N(S)|- h\Delta.
    \]
    On the other hand, we notice that if $S\cap Y$ is non-empty, then since $S$ consists of $k$ vertices, there exists a set $X_i$ such that $S\cap X_i$ is empty.
    In particular, this further implies that $N(S)$ contains no vertices from any set $P^i_{j,\ell}$.
    Let $x\in X_i$ be an arbitrary vertex from the group labelled $j$ in $X_i$, and observe that 
    \[
    |N(S\cup \{x\})|\geq |N(S)| + \sum_{\ell \in [s]}|P^i_{j,\ell}| \geq |N(S)| + s|P^i_{j,\ell}| \geq k \Delta.
    \]
    Hence, by combining the last two inequalities, we get
    \[
        |N(S - \{y\} \cup \{x\})| \geq |N(S)| - h\Delta +  k \Delta \geq |N(S)|.
    \]
\end{proof}
Completely analogous proof as in Lemma \ref{lemma:max-k-cover-lb-khov} shows the following.
\begin{lemma}\label{lemma:par-dom-lb-khov}
    For any fixed $k\geq 2$, $2\leq h \leq k$, there exists a graph $G$ with $|V(G)|=n$, $\max_{x\in V(G)}\deg(x) = \Delta$, such that the following holds. Let $N_h:=\min\{n, \Delta^{h/(h-1)}\}$.
    \begin{enumerate}
        \item If $h\geq 3$ and there is an algorithm solving Partial $k$-Dominating Set on $G$ in time $\bigO(N_h^{k(1-\varepsilon)})$ for some $\varepsilon>0$, then there exists a $\delta>0$, such that we can solve any $(k,h)$-OV instance $A_1,\dots, A_k$ with $|A_1| = \dots = |A_k| = N_h$ of dimensions $d=N_h^\delta$ in time $\bigO(N_h^{k(1-\varepsilon')})$ for some $\varepsilon'>0$.
        \item If there exists an algorithm solving Partial $k$-Dominating Set on $G$ in time $\bigO(N_2^{k\omega/3(1-\varepsilon)})$ for some $\varepsilon>0$, then there exists a $\delta>0$ such that we can solve any $(k,2)$-OV instance $A_1,\dots, A_k$ with $|A_1| = \dots = |A_k| = N_2$ of dimensions $d=N_2^\delta$ in time $\bigO(N_2^{\omega k/3(1-\varepsilon')})$ for some $\varepsilon'>0$.
    \end{enumerate}
\end{lemma}
Finally, by combining Lemma \ref{lemma:par-dom-lb-khov} with Lemmas
\ref{lemma:hyperclique-hardness-of-khov} and \ref{lemma:maxip-hardness}, Theorem \ref{theorem-partial-dom-LB} follows.

\section{Influence of Sparsity on Max $k$-Cover and Partial $k$-Dominating Set}\label{sec:sparse}
This section is dedicated to resolving the complexity of the Max $k$-Cover and Partial $k$-Dominating Set in sparse graphs. More precisely, we aim to classify the complexity of the two problems when parameterized by the number of edges \footnote{Here we consider the graph-theoretic formulation of Max $k$-Cover problem.}.
Fischer et al. \cite{FischerKR24} proved that the canonical decision version of the Max $k$-Cover, namely the $k$-Set Cover problem requires $m^{k- o(1)}$ time unless the $k$-OV hypothesis fails. 
Together with Proposition \ref{prop:baseline-max-cover}, we get a full understanding of the fine-grained complexity of Max $k$-Cover problem in sparse graphs.
For completeness, we state this result here. 
\begin{proposition}[Complexity of Max $k$-Cover in sparse graphs]
    Let $k\geq 2$. We can solve Max $k$-Cover on a given bipartite graph $G=(X \cup Y,E)$ with $|X| = n$, $|Y| = u$ in time 
    \[
    \MM\Big(n^{\lceil \frac{k}{2}\rceil}, m, n^{\lfloor \frac{k}{2}\rfloor}\Big).
    \]
    If $\omega = 2$, or $k\geq 8$, this running time is bounded by $m^{k+o(1)}$.
    Moreover, any algorithm solving Max-$k$-Cover in time $O(m^{k-\varepsilon})$ for any $\varepsilon>0$ would refute the $k$-OV hypothesis.
\end{proposition}
The complexity landscape of Partial $k$-Dominating Set in sparse graphs is much more interesting. We first focus on the case when $k\geq 3$. 
Similarly as above, from the $k$-Dominating Set lower bound due to Fischer et al. \cite{FischerKR24}, we get the $k$-OV based lower bound of $mn^{k-2-o(1)}$ for each $k\geq 3$. The natural question is if we can match this lower bound. 
By a simple modification of the construction from the proof of Theorem \ref{thm:sparse-graph-pd-hardness}, we show that it is unlikely to match this for every dependence between $m$ and $n$, unless $3$-Uniform $k$-Hyperclique hypothesis fails. 
In particular, we show an incomparable lower bound of $m^{3k/5-o(1)}$, and more generally prove the following theorem.
\begin{theorem}\label{thm:sparse-graph-pd-hardness}
        Given a graph $G$ with $n$ vertices $m$ edges, if there exists $\varepsilon>0$ such that we can solve Partial $k$-Dominating Set in time 
    \begin{itemize}
        \item $\bigO\left( \min\{n,m^{2/3}\}^{\frac{k\omega}{3}-\varepsilon}\right)$, then $k$-Clique Hypothesis is false.
        \item $\bigO\left(\min\{n,m^{\frac{h}{2h-1}}\}^{k-\varepsilon}\right)$, for $k> h\geq 3$, then the $h$-Uniform $k$-Hyperclique Hypothesis is false.
        \item $\bigO\left(\min\{n,m^{\frac{k}{2k-1}}\}^{k-\varepsilon}\right)$, for $k\geq 2$, then the $k$-OV Hypothesis is false.
        \item $\bigO\left(mn^{k-2-\varepsilon}\right)$, then the $k$-OV Hypothesis is false.
    \end{itemize}
\end{theorem}
We note that unlike for Theorem \ref{theorem-partial-dom-LB}, the $k$-Clique and the $3$-Uniform $k$-Hyperclique lower bounds are not incomparable and in particular, lower bounds based on $3$-Uniform $k$-Hyperclique are stronger regardless of the value of $2<\omega<2.372$~\cite{AlmanDWXXZ25}.
\begin{corollary}
    Given a graph $G$ with $n$ vertices $m$ edges, if there exists $\varepsilon>0$ such that we can solve Partial $k$-Dominating Set in time $\bigO\left(m^{3k/5-\varepsilon}\right)$, then the $3$-Uniform $k$-Hyperclique Hypothesis is false.
\end{corollary}
However, we still obtain two incomparable conditional lower bounds, namely the $\bigO\left(m^{3k/5-\varepsilon}\right)$ one from the $3$--Uniform $k$--Hyperclique hypothesis and the $\bigO\left(mn^{k-2-\varepsilon}\right)$ from the $k$-OV hypothesis.
Perhaps surprisingly, we show that we can match both of those incomparable lower bounds.
We defer the proof of Theorem \ref{thm:sparse-graph-pd-hardness} to the next section, and focus on the algorithms in this section.
The rest of this section will be dedicated to proving the following theorem.
\begin{theorem}
    For every $k\geq 9$, there is an algorithm solving Partial $k$-Dominating Set in time \[
    mn^{k-2+o(1)} + m^{3k/5+\bigO(1)}.
    \]
    If $\omega=2$, this running time can be achieved for every $k\geq 3$.
\end{theorem}
The strategy is to first show that we can find all the solutions $S$ such that the induced subgraph $G[S]$ contains an edge efficiently by using the standard matrix multiplication type of argument \cite{EisenbrandG04}. 
Then we construct a recursive algorithm that detects any solution that forms independent set in $G$. If the value of matrix multiplication exponent $\omega$ is small enough, this already suffices to obtain the desired running time. 
However, to match the $3$--Uniform $k$--Hyperclique based lower bound with the current value of $\omega$, we additionally need to apply the Regularization Lemma (i.e. Lemma \ref{lemma:regularization-lemma}).
\begin{lemma}
    Let $k\ge 3$, and let $G$ be a graph with $n$ vertices and $m$ edges and $t\leq n$ be arbitrary.
    We can enumerate all sets of $k$ vertices $S$ that satisfy the following two conditions: 
    \begin{enumerate}
        \item The induced subgraph $G[S]$ contains an edge,
        \item The vertices in $S$ dominate at least $t$ vertices (i.e. $|N[S]|\geq t$),
    \end{enumerate}
     in time 
     \[
     \MM \left(mn^{\lfloor \frac{k-2}{2}\rfloor}, n, n^{\lceil \frac{k-2}{2}\rceil}\right)
     \]
     If $k\geq 9$, or $\omega = 2$, this time becomes $mn^{k-2}$.
\end{lemma}
\begin{proof}
    Let $A$ be a $\{0,1\}$ matrix whose rows are indexed by subsets of $V$ of size $\lfloor \frac{k-2}{2}\rfloor + 2$ whose induced subgraphs contain an edge, and columns are indexed by the vertices $V$. Set $A[P,v] = 1$ if and only if $v\not\in N[P]$.
    Similarly, let $B$ be a $\{0,1\}$ matrix whose columns are indexed by subsets of $V$ of size $\lceil \frac{k-2}{2}\rceil$, and rows are indexed by $V$. 
    Set $B[v, Q] = 1$ if and only if $v\not\in N[Q]$.
    Similarly as in Proposition \ref{prop:baseline-max-cover}, if we define $C:=A\cdot B$, then $C[P,Q]$ counts the number of vertices $v\in V$ that are not dominated by $P\cup Q$ (i.e. $C[P,Q] = n-\left|N[P\cup Q]\right|$). 
    Hence for every subset $S\in \binom{V}{k}$ that contains an edge, we can read off the value of $|N[S]|$ just from $C$. 
    Since the dimensions of $A$ and $B$ are $\bigO\left( mn^{\lfloor \frac{k-2}{2}\rfloor}\right)\times n$ and $n\times \bigO\left( n^{\lceil\frac{k-2}{2}  \rceil} \right)$, respectively, the claimed running time follows.
\end{proof}
If there exists a set $S\in \binom{V}{k}$ such that $G[S]$ contains an edge and $\max_{T\in \binom{V}{k}}|N[T]| = |N[S]|$, then the previous algorithm will find it. 
We now show that if for all $S$ that satisfy $\max_{T\in \binom{V}{k}}|N[T]| = |N[S]|$, $S$ induces an independent set, we can still find a valid solution that maximizes the value $|N[S]|$ efficiently.

\begin{algorithm}
\begin{algorithmic}[1]
    \Procedure{independent-partial-DS}{$G,k$}
        \If{$k=2$}
            \State Run the algorithm from \autoref{prop:baseline-max-cover}.
        \EndIf
        \State Let $n:= |V(G)|; \; m:=|E(G)|;\; \Delta:= \max_{v\in V(G)} \deg(v)$.
        \State Let $H_1, H_2$ be as in \autoref{lemma:regularization-lemma}.
        \If{$\Delta \geq m^{2/5}$ or $|H_1|\leq 2k^2\Delta$}
            \State
            \Return {$\max_{x\in H_1} (\, |N[x]| + {}$\Call{independent-partial-DS}{$G-N[x], k-1$}$\,)$} 
        \Else
            \State $X\gets H_2, Y \gets N(H_2)$ \label{alg:ind-pds:line-8}
            \State Remove all but the heaviest $\min \{k\Delta^2, |X|\}$ vertices from $X$.
            \State 
            \Return \Call{Partial-DS}{$X, Y, k$} \Comment{Function from \autoref{alg:partial-dom-alg}}
        \EndIf
    \EndProcedure
\end{algorithmic}
\caption{}
\label{alg:ind-pds}
\end{algorithm}

\begin{lemma}
     Let $k\ge 3$, and let $G$ be a graph with $n$ vertices and $m$ edges.
     If there exists an independent set $I$ of size $k$ such that for any $S\in \binom{V}{k}$ such that $G[S]$ contains an edge it holds that $N[I] > N[S]$,
     then in time bounded by $m^{3k/5}n^{\bigO(1)}$ we can find $I$. 
\end{lemma}
\begin{proof}
    Consider Algorithm \autoref{alg:ind-pds}.
    We first show that if such an independent set $I$ exists, this algorithm will detect it correctly.
    Let $\Delta$ be the maximum degree in $G$ and let $H_1 := \{v\mid \deg(v)\geq \frac \Delta k\}$, $H_2 := \{v\mid \deg(v)\geq \frac{\Delta}{2k}\}$.
    By Lemma \ref{lemma:regularization-lemma}, any solution contains a vertex from $H_1$.\footnote{Note that in general when finding an independent set $I$ of size $k$ that maximizes $N[I]$, this is not necessary, since it is possible that $G$ contains no independent sets of size $k$ that intersect $H_1$, but if this was the case, then there is a strictly better solution $T$, such that the induced subgraph $G[T]$ contains an edge, contradicting our assumption that $\max_{S\in \binom{V}{k}} N[S] = N[I]$ for some independent set $I$.}
    We observe that, since we are only looking for solutions that induce an independent set, by deleting the closed neighborhood of the guessed vertex, we are not destroying any potential solutions. 
    Formally, combining the two arguments, if for each $\ell$, $I^{(\ell)}_G$ denotes the set of all independent sets in $G$ of size $\ell$, we have the following equality:
    \begin{align*}
        \max_{S\in I_G^{(k)}} |N_G[S]| &= \max_{\substack{v\in V\\ S'\in I_{G-N[v]}^{(k-1)}}} |N[v]| + |N[S']| & \\
        &= \max_{\substack{v\in H_1\\ S'\in I_{G-N[v]}^{(k-1)}}} |N[v]| + |N[S']| & \left(I\cap H_1\neq \emptyset\right).
    \end{align*}
    Now applying a simple induction establishes the correctness of the recursive step, where the correctness of base case is discussed in the Appendix \ref{sec:baseline}.
    The correctness of the remaining part follows easily from Lemmas \ref{lemma:regularization-lemma}, \ref{lemma:bounding-number-of-sets}, and the proof of Theorem \ref{theorem:partial-dom-algorithm}, combined with the assumption that there exists an independent set $I$ of size $k$ such that for any $S\in \binom{V}{k}$ such that $G[S]$ contains an edge it holds that $N[I] > N[S]$.

    It remains to analyze the time complexity of the algorithm $T(m,n,k)$. 
    Let us first consider the time complexity of the last recursive call assuming $k>2$ at the time of the last recursive call (namely we enter the else block on line \ref{alg:ind-pds:line-8}).
    As argued in the proof of Theorem \ref{theorem:partial-dom-algorithm}, the running time of this procedure is bounded by 
    \[
     T_1(m,n,k) \leq \left(\Delta^{3k/2} + \min\{|X|, \Delta^2\}^{k\omega/3}\right)\Delta^{\bigO(1)}.
    \]
    We claim that we can bound this value by $m^{3k/5}\Delta^{\bigO(1)}$.
    We first note that we can only enter line~\ref{alg:ind-pds:line-8} if the value of $\Delta<m^{2/5}$. Hence $\Delta^{3k/2}\leq m^{3k/5}$ as desired. 
    For the second term, if the matrix multiplication exponent $\omega<2.25$, we get the similar result, since clearly $\Delta^{2k\omega/3}\leq m^{3k/5}$ when $\Delta\leq m^{2/5}$. 
    However, with the current value of $\omega$, we need to be slightly more careful and that is where the Regularization Lemma (Lemma \ref{lemma:regularization-lemma}) comes into play. 
    Namely, we can only enter line \ref{alg:ind-pds:line-8} if $|H_1|> 2k^2\Delta$, hence by Regularization Lemma, all the solution vertices are contained in $H_2$ (hence there are at most $\bigO(\frac{m}{\Delta})$ many choices for each solution vertex). 
    In particular, $|X|\leq \bigO\left(\frac{m}{\Delta}\right)$.
    We now make a simple case distinction.
    \begin{itemize}
        \item If $\Delta\leq m^{\frac{9}{10\omega}}$, we have
        \begin{align*}
            \min\{|X|, \Delta^2\}^{k\omega/3}&\leq \Delta^{2k\omega/3} 
            \\& \leq m^{\frac{3k}{5}} & (\Delta\leq m^{\frac{9}{10\omega}}).
        \end{align*}
        \item If $\Delta\geq m^{\frac{9}{10\omega}}$, we have
        \begin{align*}
            \min\{|X|, \Delta^2\}^{k\omega/3}&\leq |X|^{k\omega/3} 
            \\& \leq \bigO\left(\left(\frac m \Delta\right)^{k\omega/3}\right)  & \text{(Regularization Lemma)} \\
            &\leq \bigO\left(m^{\left(1-{\frac{9}{10\omega}}\right) k\omega/3} \right)& (\Delta\geq m^{\frac{9} {10\omega}}) \\
            & = \bigO\left(m^{\frac{10\omega-9}{30}k}\right) \\
            & \leq \bigO\left(m^{3k/5}\right) & (\omega<2.7).
        \end{align*}
    \end{itemize}
    Plugging this back in, we can conclude that $T_1(m,n,k)\leq \left(m^{3k/5}\Delta^{\bigO(1)}\right)$. 
    Finally, before giving the full running time of the algorithm above, we need to state a few more simple observations.
    \begin{enumerate}
        \item $T(m,n,2)\leq n^{\omega+o(1)}$.
        \item At each step we recurse on at most $m^{3/5}$ branches.
        \item For each $m'\leq m$ and $n'\leq n$, it holds that $T(m',n',k) \leq T(m,n,k)$.
        \item For each $x\in V$, we can construct the graph $G-N[x]$ in time $\bigO(m)$
    \end{enumerate}
    We can hence bound the total running time (up to constant factors) for each $k\geq 3$ as follows.
    \begin{align*}
        T(m,n,k) &\leq m^{3/5}\left( m+ T(m,n,k-1)\right) + T_1(m,n,k) & \\
        & \leq m^{3/5}\cdot T(m,n,k-1) + T_1(m,n,k) & \left(\text{for $k\geq 3$, }T(m,n,k-1)\geq m\right)\\
        & \leq m^{3/5}\cdot T(m,n,k-1) + m^{3k/5}\Delta^{\bigO(1)}& \left(T_1(m,n,k) \leq m^{3k/5}\Delta^{\bigO(1)}\right)\\
        &\leq m^{3(k-2)/5}n^{\omega+o(1)} + m^{3k/5}\Delta^{\bigO(1)} \\
        & \leq  m^{3k/5}n^{\bigO(1)}.& \qedhere
    \end{align*}
\end{proof}
\subsection{Algorithm for Partial $2$-Dominating Set}
A very interesting special case which was not considered in depth in the previous section is the case $k=2$ (in the previous section we just gave a baseline algorithm running in $n^\omega$ that does not exploit sparsity).
In particular, the lower bound construction for $k$-Dominating Set from \cite{FischerKR24} gives no meaningful lower bound for this special case, and they also show that $2$-Dominating Set can be solved in near-linear time $\tilde\bigO(m)$ if $\omega = 2$. 
This raises the question of whether we can obtain a similar algorithm for Partial $2$-Dominating Set.
The lower bound construction from the next section answers this question negatively by providing a non-trivial conditional lower bound based on OVH:
\begin{theorem} \label{thm:2pds-lb}
    If there exists an algorithm solving Partial $2$-Dominating Set in time $\bigO\left(m^{4/3-\varepsilon}\right)$, then the OV Hypothesis is false.
\end{theorem}
This section is dedicated to showing that we can in fact construct an algorithm that matches this lower bound (if $\omega=2$). In particular we prove the following theorem.
\begin{theorem}\label{thm:2pds-algo}
    There exists an algorithm that solves Partial $2$-Dominating Set in time $m^{\frac{2\omega}{\omega+1}+o(1)}$.
\end{theorem}
The strategy to proving this theorem is to first argue that if both solution nodes have degree at least $d$ (for some $d$ that we determine later), we can apply an approach similar to the one from the classical sparse triangle counting algorithm \cite{AlonYZ97} to efficiently detect any such solution.
On the other hand, if there is a solution that contains a vertex of degree at most $d$, a slightly more involved counting argument shows that we can correctly detect any such solution in time $\bigO(m\cdot d)$.
\begin{lemma}\label{lemma:k2-heavy-solutions}
    Given a graph $G$ with $n$ vertices and $m$ edges, in time $m\cdot d + \left(\frac{m}{d}\right)^{\omega+o(1)}$, we can detect any solution $x_1,x_2$ that maximizes the value $\max_{x_1,x_2\in V}|N[x_1]\cup N[x_2]|$, with $d\leq \deg(x_2)\leq \deg(x_1)$.
\end{lemma}
\begin{proof} 
    Recall that by the principle of inclusion-exclusion, we have $|N[i]\cup N[j]| = |N[i]| + |N[j]|-|N[i]\cap N[j]|$.
    Hence, it suffices to show that we can for each $i,j$ compute the value $|N[i]\cap N[j]|$ in the claimed time.
    Note that there are at most $s:= \bigO(\frac{m}{d})$ many vertices of degree at least $d$. 
    \begin{claim}
        Let $A$ be a $\{0,1\}$--matrix of dimensions $\frac{m}{d}\times n$, with at most $m$ ones. 
        Then the matrix $AA^T$ can be computed in time $m\cdot d + \left(\frac{m}{d}\right)^{\omega+o(1)}$.
    \end{claim}
    \begin{subproof}
        We follow a simple heavy-light approach as in~\cite{YusterZ05}.
        Let $A_d$ be the $\frac{m}{d}\times \bigO(\frac{m}{d})$ submatrix of $A$ such that each column of $A_{\ge d}$ has at least $d$ many ones. 
        Since $A$ has at most $m$ many ones in total, clearly $A_{\ge d}$ has at most $\bigO(\frac{m}{d})$ many columns.
        Hence, computing $B:=A_{\ge d}A_{\ge d}^T$ takes time~\smash{$(\frac{m}{d})^{\omega+o(1)}$}.
        Consider now the submatrix $A_{<d}$ of $A$ consisting of the columns with less than $d$ ones. 
        Fix a pair of indices $(i,k)$ such that $A_{<d}[i,k] = 1$.
        Now for all $j$, such that $A_{<d}[k,j] = 1$ increment $B[i,j]$ by one.
        After doing this for all pairs $(i,k)$, clearly $B$ will precisely be equal to $AA^T$.
        Note that there are at most $m$ pairs $(i,k)$, satisfying $A_{<d}[i,k] = 1$ and for each such pair, by construction of $A_{<d}$, there are at most $d$ many indices $j$ such that $A_{<d}[k,j] = 1$. 
        This step hence takes a total of $\bigO(m\cdot d)$ time, yielding the desired time to compute $AA^T$.
    \end{subproof}
    We now construct a submatrix $A$ of the adjacency matrix of $G$ consisting only of those rows with degree at least $d$. 
    Applying the claim above, we can compute the matrix $B:=AA^T$ in time $m\cdot d + \left(\frac{m}{d}\right)^{\omega+o(1)}$, and note that for each pair of indices $(i,j)$, it holds that $B[i,j] = |N(i) \cap N(j)|$. 
    Now for each pair of vertices $(i,j)$, we have 
    \[|N[i]\cup N[j]| = |N[i]| + |N[j]| - (A[i,j]+B[i,j]). \qedhere\]
\end{proof}
We now focus on the remaining part, namely finding solutions $x_1,x_2$ such that $\deg(x_2)\leq d$.
For the rest of this section, let $V_d:=\{v\in V\mid d\leq \deg(v)< 2d\}$.
We first prove that if there exists a solution $x_1,x_2$ that intersects $V_d$, we can efficiently find it and by running this process.
We begin by proving a simple observation.
\begin{observation}\label{obs:2pds-Vd}
    Let $G$ be a graph with maximum degree $\Delta$. 
    Let $x_1,x_2$ be a solution to Partial $2$-Dominating Set on $G$, such that $x_2\in V_d$ and $\deg(x_2)\leq \deg(x_1)$.
    Then $\deg(x_1)\geq \max\{\Delta-2d, \frac \Delta 2\}$.
\end{observation}
\begin{proof}
    We note that $\deg(x_1)\geq \frac \Delta 2$ was already argued in Lemma \ref{lemma:regularization-lemma}, so we only need to prove that $\deg(x_1)\geq \Delta-2d$.
    By the inclusion-exclusion principle, we have that 
    \[\max\{|N[x_1], |N[x_2]|\} \leq |N[x_1] \cup N[x_2]|\leq  |N[x_1]| + |N[x_2]|.\]
    Assume that $x_2\in V_d$ and $\deg(x_1)<\Delta-2d$, and let $x$ be a vertex such that $\deg(x) = \Delta$ (highest degree vertex).
    Then by the inequality above 
    \begin{align*}
        |N[x_1] \cup N[x_2]|& \leq  |N[x_1]| + |N[x_2]| \\
        &< \Delta \\&\leq |N[x]|
        \\ &\leq |N[x] \cup N[x_2]|,
    \end{align*}
    hence $x,x_2$ is a strictly better solution than $x_1,x_2$, contradicting the assumption that $x_1,x_2$ was a valid solution to Partial $2$-Dominating Set.
\end{proof} 
\begin{lemma}\label{lemma:2pds-light-solutions}
    Let $d\leq \frac{\Delta}{4}$.
    There exists an algorithm that in time $\bigO(m\cdot d)$ either returns a pair of vertices that maximizes the value $\max_{x_1\in V, x_2\in V_d}|N(x_1)\cup N(x_2)|$, or correctly reports that no solution contains a vertex from $V_d$ (i.e. $\max_{x_1\in V, x_2\in V_d}|N(x_1)\cup N(x_2)|<\max_{\substack{y_1,y_2\in V}}|N(y_1)\cup N(y_2)|$).
\end{lemma} 
\begin{proof}
    Let $H_d$ be the set consisting of all vertices $x$ that satisfy $\deg(x)\geq \Delta-2d$. 
    By the previous observation, if $x_2$ is in $V_d$, then $x_1$ must be contained in $H_d$. 
    Let $x$ be a vertex of degree $\Delta$ and let $S:=N[x]$, $R:=V\setminus S$.
    \begin{claim}
        If there exists a vertex $v$ such that $|N(v)\cap R|\ge 2d +1$, we can report that no solution contains a vertex from $V_d$.
    \end{claim}
    \begin{subproof}
        Any pair of vertices $x_1, x_2$ such that $x_2\in V_d$ dominates at most $(\Delta + 1) + 2d$ many vertices.
        Let $y$ be any vertex that contains $\ge 2d +1$ many neighbors in $R$. 
        Then we have:
        \begin{align*}
            |N[x] \cup N[y]| &\ge (\Delta+1) + (2d+1) \\&= \Delta + 2d + 2 \\& > |N[x_1] \cup N[x_2]|,
        \end{align*}
        therefore $x_1, x_2$ cannot be a valid solution.
    \end{subproof}
    We may thus assume that for any vertex $v\in V$ it holds that $|N(v)\cap R| \leq 2d$.
    \begin{claim}\label{claim:2pds-2}
        If there exists a vertex $v\in H_d$ that is \emph{not} adjacent to at least $4d+1$ many vertices in $S$, we can report that no solution contains a vertex from $V_d$.
    \end{claim}
    \begin{subproof}
        Since $v$ is in $H_d$, it has degree at least $\Delta-2d$ and hence at least $(\Delta-2d) - (\Delta-(4d+1)) = 2d+1$ neighbors in $R$. Hence, by the previous claim, we can conclude that no solution contains a vertex from $V_d$.
    \end{subproof}
    We can thus further assume that for any vertex $v\in H_d$, it holds that $|S\setminus N(v)| \leq 4d$.
    For any $X,Y\subseteq V$, let $E(X,Y)$ (resp. $\overline{E}(X,Y)$) denote the set of edges (resp. non-edges) between $X$ and $Y$ (i.e. the set $(X\times Y)\cap E$ and $(X\times Y)\setminus E$ respectively).
    \begin{claim}\label{claim:2pds-3}
        Assuming that the last two claims do not report that there is no valid solution that intersects $V_d$, then
        \begin{enumerate}
            \item We can enumerate all triples of vertices $(x,s,y)\in H_d\times S \times V_d$ such that $\{x,s\}\in \overline E$ and $\{s,y\}\in E$ in time $\bigO(m\cdot d)$.
            \item We can enumerate all triples of vertices $(x,r,y)\in H_d\times R \times V_d$ such that $\{x,r\}\in E$ and $\{r,y\}\in E$ in time $\bigO(m\cdot d)$.
        \end{enumerate}
    \end{claim}
    \begin{subproof}
        Note that for any $d\leq \frac{\Delta}{4}$, if $x_1,x_2$ is a valid solution with $x_2\in V_d$, then 
        $H_d$ consists of at most $\bigO\left(\frac{m}{\Delta}\right)$ many vertices.
        Thus we can, for each vertex $x\in H_d$ enumerate each of the $\bigO(d)$ many vertices $s\in S$ that are non-adjacent to $x$ (assuming that the procedure from the claim above did not report that there is no valid solution intersecting $V_d$), and finally since the maximum degree in $G$ is $\Delta$, we can in $\bigO(\Delta)$ enumerate all neighbors of $s$ that are in $V_d$. 
        In total this enumerates all desired triples $(x,s,y)\in H_d\times S \times V_d$, in time $\bigO(m\cdot d)$. 
        The triples $(x,r,y)\in H_d\times R \times V_d$ such that $\{x,r\}\in E$ and $\{r,y\}\in E$ can be enumerated similarly.
    \end{subproof}
    We now create two empty dictionaries $C_S$ and $C_R$. 
    We enumerate all triples $(x,s,y)\in H_d\times S \times V_d$ such that $\{x,s\}\in \overline E$ and $\{s,y\}\in E$ as in the previous claim and for each such $(x,s,y)$, if $(x,y)$ is a key in $C_S$, we increment the value stored in the table: $C_S[x,y]\gets C_S[x,y] + 1$.
    Otherwise, initialize the entry $C_S[x,y] = 1$.
    We fill in the dictionary $C_R$ similarly, using the enumerated triples $(x,r,y)\in H_d\times R \times V_d$ from the previous claim.
    It is easy to see that for any pair $x,y\in H_d \times V_d$, the following two equalities hold (we assume $C_S[x,y]$ (resp. $C_R[x,y]$) to be $0$ for all non-initialized pairs):
    \begin{align}
        &|\left(N(x) \cup N(y)\right)\cap S| = |N(x)\cap S| + C_S[x,y] \label{eq:5}\\
        &|\left(N(x) \cup N(y)\right)\cap R| = |N(x)\cap R| + |N(y)\cap R| - C_R[x,y]\label{eq:6}.
    \end{align}
    From this point on, by combining the values from the two equalities above, we can for each pair of vertices $(x,y) \in H_d \times V_d$ compute the quantity $|N(x) \cup N(y)|$ in constant time.
    \begin{claim}
        For any constant $c$, there exists an algorithm that returns, in time $\bigO(m\cdot d)$, the set $\mathcal H\subseteq H_d\times V_d$, such that:
        \begin{enumerate}
            \item $|\mathcal H| = \max\{|H_d \times V_d|, c\cdot m\cdot d\}$.
            \item For every pair $(x,y)\in \mathcal H$ and for any pair $(x',y') \in (H_d\times V_d) \setminus \mathcal H$, the inequality \[|N(x')| + |N(y')\cap R| \leq |N(x)| + |N(y)\cap R|\] is satisfied.
        \end{enumerate}
    \end{claim}
    \begin{subproof}
        Consider the following algorithm.
        \begin{algorithm}
            \begin{algorithmic}[1]
                \State $\text{count}\gets 0$
                \State $\mathcal{H} \gets \emptyset$
                \For{$ \text{sum}= \Delta + 2d, \dots, 0$}
                    \For{$q = 0,\dots, 2d$}
                        \For {$y\in V_d$ satisfying $|N(y)\cap R| = q$}
                            \For {$x\in H_d$ satisfying $|N(x)| = \text{sum}- q$}
                                \State Add $(x,y)$ to $\mathcal H$
                                \State $\text{count}{+}+$
                                \If{$\text{count} \ge c\cdot m\cdot d$}
                                    \State \Return $\mathcal H$
                                \EndIf 
                            \EndFor
                        \EndFor
                    \EndFor
                \EndFor 
            \end{algorithmic}
            \caption{}
            \label{alg:vertex-pair-sorting}
        \end{algorithm}
        
        It is easy to verify that this algorithm halts after at most $\bigO(\Delta\cdot d + c\cdot m\cdot d) = \bigO(m\cdot d)$ many steps and returns the set $\mathcal H$ that satisfies both of the desired conditions.
    \end{subproof}
    From Equations \ref{eq:5} and \ref{eq:6}, we can see that for any pair $(x,y)\in H_d\times V_d$, if $C_S[x,y]$ is not initialized, then $|N(x) \cup N(y)| = |N(x)|+|N(y)\cap R|-C_R[x,y]$.
    Hence, by utilizing the previous algorithm, we can prove the following claim.
    \begin{claim}
        There exists an algorithm that in time $\bigO(md)$ finds a pair of vertices $(x,y)\in H_d\times V_d$, such that the following holds:
        \begin{enumerate}
            \item $C_S[x,y]$ is not initialized.
            \item For each pair $x',y'$ such that $C_S[x',y']$ is not initialized, $|N(x)\cup N(y)|\geq |N(x')\cup N(y')|$.
        \end{enumerate}
    \end{claim}
        \begin{subproof}
            Let $c$ be any constant such that $c\cdot md > |C_S| + |C_R|$ (where $|C_S|, |C_R|$ denote the number of initialized pairs) and let $\mathcal H$ be the set of size $c\cdot md$ returned by Algorithm \ref{alg:vertex-pair-sorting}. We now run the following algorithm on the pairs in $\mathcal H$
            \begin{algorithm}
                \begin{algorithmic}[1]
                    \State $\text{current}\gets null$
                    \State $M\gets 0$
                    \For{$(x,y)\in \mathcal H$}
                        \If{$C_S[x,y]$ is initialized}
                            \State \textbf{continue} \Comment{skip the pairs with $C_S[x,y]$ initialized}
                        \EndIf 
                        \If{$|N(x)|+|N(y)\cap R|-C_R[x,y] >M$}
                            \State $\text{current}\gets (x,y)$
                            \State $M\gets |N(x)|+|N(y)\cap R|-C_R[x,y]$
                        \EndIf 
                    \EndFor 
                    \State \Return current
                \end{algorithmic}
                \caption{}
                \label{alg:iterating-H}
            \end{algorithm}       

            Clearly, there are at most $\bigO(md)$ many iterations and each runs in $\bigO(1)$ time, hence the total time complexity is $\bigO(md)$. 
            We now verify that this algorithm is correct.
            It is also easy to see that the algorithm returns the pair $(x,y)$ that maximizes the value $|N(x)\cup N(y)|$ out of all pairs $(x,y)\in \mathcal H$ for which $C_S[x,y]$ is not initialized.
            We now show that it is enough to look at pairs $(x,y)\in \mathcal H$.
            First of all, if $|\mathcal H| = |H_d\times V_d|$, it is trivial. 
            Hence, we can assume that $|\mathcal H| = c\cdot md$. 
            Let $(x,y)$ be the pair that the algorithm returned,
            and assume for contradiction that there is a pair $(x',y')\in (H_d\times V_d)\setminus \mathcal H$ such that $C_S[x',y']$ is not initialized and $|N(x')\cup N(y')|>|N(x)\cup N(y)|$.
            As already argued above, since $C_S[x',y']$ is not initialized, we have 
            \begin{align*}
                |N(x')\cup N(y')| &=|N(x')|+|N(y')\cap R|-C_R[x,y] &\text{(Eq \ref{eq:5} and \ref{eq:6})}\\
                &\le |N(x')|+|N(y')\cap R| \\
                & \le \min_{(x^*, y^*)\in \mathcal H} |N(x^*)|+|N(y^*)\cap R| & \text{(by construction of $\mathcal H$)}
            \end{align*}
            However, by the choice of the constant $c$, we have that $|\mathcal H|>|C_S| + |C_R|$, and since we skip at most only $|C_S|$ pairs, at some iteration, we will come across a pair $(x^\dagger, y^{\dagger})$, for which it holds that $C_R[x^\dagger, y^{\dagger}] = 0$ (i.e. $C_R[x^\dagger, y^\dagger]$ is not initialized), and we have:
            \begin{align*}
                |N(x^\dagger)\cup N(y^\dagger)| &= |N(x^\dagger)|+|N(y^\dagger)\cap R| -C_R[x^\dagger,y^\dagger]&\\
                & = |N(x^\dagger)|+|N(y^\dagger)\cap R| & (C_R[x^\dagger,y^\dagger] = 0)\\
                & \ge \min_{(x^*, y^*)\in \mathcal H} |N(x^*)|+|N(y^*)\cap R| & \left((x^\dagger,y^\dagger)\in \mathcal H\right)\\
                & \geq |N(x')\cup N(y')|
            \end{align*}
        \end{subproof}
        Finally, it only remains to check the pairs $(x,y)$ for which $C_S$ is initialized. However, there are only $\bigO(md)$ such pairs and by utilizing Equations \ref{eq:5} and \ref{eq:6}, we can compute the value of $|N(x)\cup N(y)|$ for each such pair in constant time. \footnote{Note that unlike earlier in the paper, this lemma considers the open neighborhoods of the solution vertices. This is only due to clarity of the presentation, to avoid dealing with constant additive factors at each step, as well as distinguishing between whether the solution vertices are adjacent or not, but we remark that it is very simple to extend each argument to also hold for the closed neighborhood.}
\end{proof}
\begin{proof}[Proof (of Theorem \ref{thm:2pds-algo})]
    Let $\gamma = m^{\frac{\omega-1}{\omega+1}}$ and run the algorithm from Lemma
    \ref{lemma:k2-heavy-solutions} to find all potential solutions $x_1,x_2$ with $\gamma \leq \deg(x_2)\leq \deg(x_1)$ in time $m^{\frac{2\omega}{\omega+1}+o(1)}$.
    Let $M$ be the maximum number of dominated vertices by any potential solution considered so far.
    Now, for each $0\leq \ell\leq \lceil\log \gamma \rceil$ set $d:=2^\ell$ and run the algorithm from Lemma \ref{lemma:2pds-light-solutions} in time $\bigO(m\cdot d)$ and update $M$ accordingly to keep track of the best solution seen at each point.
    Clearly, this covers the whole search space and the correctness follows from Lemmas \ref{lemma:2pds-light-solutions} and \ref{lemma:k2-heavy-solutions}.
     Furthermore, note that we are running the algorithm from Lemma \ref{lemma:2pds-light-solutions} only $\bigO(\log m)$ times, and the running time at each iteration is bounded by $\bigO(m\cdot \gamma) = \bigO\left(m^{\frac{2\omega}{\omega+1}}\right)$.
    Therefore, the total running time is bounded by $m^{\frac{2\omega}{\omega+1}+o(1)}$.
\end{proof}

\subsection{Hardness of Partial $k$-Dominating Set in Sparse Graphs}
In this section we prove Theorem \ref{thm:sparse-graph-pd-hardness}. 
The proof essentially uses the same construction as the proof of Theorem \ref{theorem-partial-dom-LB}, with carefully chosen parameters that assure the number of vertices and edges in the reduction to remain as desired.
\begin{lemma} \label{lemma:sparse-partial-domination-lb-reduction}
    Let $2\leq h\leq k$ be fixed integers and $n, m$ be given positive integers with $n\leq m\leq n^2$. Let $A_1,\dots, A_k\subseteq \{0,1\}^d$ be sets consisting of $\min\{n, m^{h/(2h-1)}\}$ many $d$-dimensional binary vectors, and each coordinate $y\in [d]$ associated with $h$ active indices $i_1,\dots, i_h\in [k]$.
    We can construct a graph $G = (X\cup Y,E)$ satisfying the following conditions:
    \begin{itemize}
        \item $G$ consists of at most $\bigO(n)$ many vertices and at most $\bigO(m\cdot d)$ many edges.
        \item We can compute positive integers $t,\alpha$ such that $G$ contains $k$ vertices $x_1,\dots, x_k$ satisfying $|N(x_1)\cup \dots \cup N(x_k)|\geq t$ if and only if there are vectors $a_1\in A_1,\dots, a_k\in A_k$ satisfying $a_1\cdot\dots\cdot a_k\geq \alpha$ if $h$ is odd (reduction from $(k,h)$-maxIP), and if $h$ is even, $|N(x_1)\cup \dots \cup N(x_k)|\geq t$ if and only if $a_1\cdot\dots\cdot a_k = 0$ (reduction from $(k,h)$-OV).
        \item $G$ can be constructed deterministically in time $\bigO(m\cdot d)$.
    \end{itemize}
\end{lemma}
\begin{proof}
    Given a such instance $A_1,\dots, A_k$ construct the bipartite graph $(X\cup Y,E)$ as in the proof of Lemma \ref{lemma:partial-domination-lb-reduction}, by setting $\Delta:=m^{\frac{h-1}{2h-1}}$.
    It is easy to check that by plugging in the value of $\Delta$ as above, all of the properties of Lemma \ref{lemma:partial-domination-lb-reduction} translate nicely to our desired properties, and the proof follows directly.
\end{proof}
The proof of the following lemma is analogous to the proof of Lemma \ref{lemma:max-k-cover-lb-khov}.
\begin{lemma}\label{lemma:sparse-par-dom-lb-khov}
    For any fixed $k\geq 2$, $2\leq h \leq k$, there exists a graph $G$ with $|V(G)|=n$, $|E| = m$, such that the following holds. Let $N_h:=\min\{n, m^{h/(2h-1)}\}$.
    \begin{enumerate}
        \item If $h\geq 3$ and there is an algorithm solving Partial $k$-Dominating Set on $G$ in time $\bigO(N_h^{k(1-\varepsilon)})$ for some $\varepsilon>0$, then there exists a $\delta>0$, such that we can solve any $(k,h)$-OV instance $A_1,\dots, A_k$ with $|A_1| = \dots = |A_k| = N_h$ of dimensions $d=N_h^\delta$ in time $\bigO(N_h^{k(1-\varepsilon')})$ for some $\varepsilon'>0$.
        \item If there exists an algorithm solving Partial $k$-Dominating Set on $G$ in time $\bigO(N_2^{k\omega/3(1-\varepsilon)})$ for some $\varepsilon>0$, then there exists a $\delta>0$ such that we can solve any $(k,2)$-OV instance $A_1,\dots, A_k$ with $|A_1| = \dots = |A_k| = N_2$ of dimensions $d=N_2^\delta$ in time $\bigO(N_2^{\omega k/3(1-\varepsilon')})$ for some $\varepsilon'>0$.
    \end{enumerate}
\end{lemma}
One final ingredient missing in the proof of Theorem \ref{thm:sparse-graph-pd-hardness} is a straightforward consequence of the lower bound construction for the $k$-Dominating Set problem in sparse graphs given by Fischer et al. \cite{FischerKR24}.
\begin{lemma}\label{lemma:soda24-hardness}
    For any fixed $k\geq 3$, if there exists an algorithm solving Partial $k$-Dominating Set on graphs with $n$ vertices and $m$ edges in time $\bigO(mn^{k-2-\varepsilon})$, then the $k$-OV Hypothesis is false.
\end{lemma}
The proof of Theorem \ref{thm:sparse-graph-pd-hardness} now follows directly from Lemmas \ref{lemma:soda24-hardness}, \ref{lemma:sparse-par-dom-lb-khov}, \ref{lemma:sparse-partial-domination-lb-reduction}, and \ref{lemma:hyperclique-hardness-of-khov}.

\bibliographystyle{plainurl}
\bibliography{refs}

\appendix

\section{Baseline algorithm for Max $k$-Cover}
\label{sec:baseline}

For completeness, we adapt the $k$-Dominating Set algorithm by Eisenbrand and Grandoni \cite{EisenbrandG04} to Max $k$-Cover, which establishes a baseline algorithm.
Here, we use the bipartite graph formulation of Max $k$-Cover established in Section~\ref{sec:algorithms} .

\begin{proposition}[Baseline Algorithm]\label{prop:baseline-max-cover}
    Let $k\geq 2$. We can solve Max-$k$-Cover on a given bipartite graph $G=(X \cup Y,E)$ with $|X| = n$, $|Y| = u$ in time 
    \[
    \MM\Big(n^{\lceil \frac{k}{2}\rceil}, u, n^{\lfloor \frac{k}{2}\rfloor}\Big).
    \]
\end{proposition}
\begin{proof}
    Let $A$ be a matrix over $\{0,1\}$ whose rows are indexed by subsets of $X$ of size $\lceil\frac{k}{2}\rceil$ and columns are indexed by the elements of $Y$, such that $A[S,y] = 1$ if and only if there is \emph{no} vertex $x\in S$ such that $\{x,y\}\in E$. 
    Similarly, let $B$ be a matrix over $\{0,1\}$ whose columns are indexed by subsets of $X$ of size $\lfloor\frac{k}{2}\rfloor$ and rows are indexed by the elements of $Y$, such that $B[y,T] = 1$ if and only if there is \emph{no} vertex $x\in T$ such that $\{x,y\}\in E$.
    Define $C:=A\cdot B$.
    It is straightforward to verify that each entry $C[S,T]$ counts the number of vertices $y\in Y$ such that no vertex in $S\cup T$ is adjacent to $y$.
    Hence, from $C$ we can read off the desired optimal value $\max_{S,T} u-C[S,T]$.
   Since  the dimensions of $A$ and $B$ are $\bigO(n^{\lceil \frac{k}{2}\rceil})\times u$ and $u \times \bigO(n^{\lfloor \frac{k}{2}\rfloor})$, respectively, the claimed running time follows.
\end{proof}
As a consequence of the proposition above, we directly get the following baseline algorithm for Partial $k$-Dominating Set.
\begin{corollary}\label{cor:baseline-pd}
    Let $k\geq 2$. Partial $k$-Dominating Set on a given graph $G=(V,E)$ with $|V| = n$ can be solved in time 
    \[
    \MM\Big(n^{\lceil \frac{k}{2}\rceil}, n, n^{\lfloor \frac{k}{2}\rfloor}\Big).
    \]
    If $k\geq 8$, or $\omega=2$, this running time becomes $n^{k+o(1)}$.
\end{corollary}
Furthermore, by noticing that $u\leq n\Delta_s$, we obtain another consequence of this algorithm.
\begin{corollary}
    Let $k\geq 2$. We can solve Max-$k$-Cover on a given bipartite graph $G=(X \cup Y,E)$  with $|X| = n$, $|Y| = u$, and $\Delta_s = \max_{x\in X}\deg(x)$ in time 
    \[
    \MM\Big(n^{\lceil \frac{k}{2}\rceil}, n\Delta_s, n^{\lfloor \frac{k}{2}\rfloor}\Big).
    \]
    If $\omega = 2$, this running time becomes $n^{k+o(1)} + \Delta_s n^{\lceil k/2\rceil + 1+o(1)}$.
\end{corollary}
\end{document}

%% file: imports.tex
\usepackage[mathscr]{eucal}
\usepackage[fleqn]{amsmath}
\usepackage{amssymb,bbm,amsthm}
\usepackage{thm-restate,thmtools}
\usepackage{subcaption}
\usepackage[utf8]{inputenc}
\usepackage{graphicx}
\usepackage{hyperref}

\newcommand{\bigO}{\mathcal{O}}

\usepackage{comment,color}
\usepackage{mathtools}

\usepackage{amsthm}
\usepackage{float}
\usepackage{cancel}
\usepackage{algorithm}
\usepackage[noend]{algpseudocode}
\usepackage[margin=1in, papersize={8.5in,11in}]{geometry}

\usepackage{enumitem}

\declaretheorem[numberwithin=section]{theorem}
\declaretheorem[unnumbered, name=Theorem]{theorem*}
\declaretheorem[numberlike=theorem]{lemma}

\declaretheorem[numberlike=theorem]{proposition}

\declaretheorem[numberlike=theorem]{corollary}

\declaretheorem[unnumbered]{claim}

\declaretheorem[numberlike=theorem]{observation}

\declaretheorem[unnumbered, name=Definition]{definition*}
\declaretheorem[
  numberlike=theorem,
  shaded={rulecolor=black, rulewidth=1pt},
  name=Definition,
]{boxed-definition}

\declaretheorem[unnumbered, name=Construction]{construction*}

\declaretheorem[unnumbered, name=Conjecture]{conjecture*}
\declaretheorem[numberlike=theorem, name=Hypothesis]{hypothesis}
\declaretheorem[unnumbered, name=Hypothesis]{hypothesis*}

\newcommand{\rem}[3]{\textcolor{blue}{\textsc{#1 #2:}}
  \textcolor{red}{\textsf{#3}}}
\newcommand{\marvin}[2][says]{\rem{Marvin}{#1}{#2}}
\newcommand{\mirza}[2][says]{\rem{Mirza}{#1}{#2}}
\newcommand{\nick}[2][says]{\rem{Nick}{#1}{#2}}

\renewcommand{\rem}[3]{}

\newcommand{\FOP}[1]{$\mathrm{FOP}_k$}
\DeclareMathOperator\MM{MM}
\makeatletter
\renewcommand\paragraph{%
  \@startsection{paragraph}
    {4}
    {\z@}
    {3.25ex \@plus1ex \@minus.2ex}
    {-1em}
    {\normalfont\normalsize\bfseries\addperiod}}
\newcommand{\addperiod}[1]{#1\@addpunct{.}}

\newcommand{\norm}[1]{\left\lVert#1\right\rVert}
\newcommand{\hyperH}{\mathcal{H}}

%% file: chapters/introduction.tex
\section{Introduction}
\label{sec:intro}
Consider the following scenario: In a social network modeled as a graph $G=(V,E)$ with user base~$V$ and friendship relation $E$, we are given a budget to hire $k$ users (i.e., \emph{influencers}) to spread some information (e.g., to advertise a new product). A natural and simple measure of effectiveness is the number of users that can be directly reached by these influencers. Formally, we seek to maximize the union of their neighborhoods: $|N_G(x_1) \cup \dots \cup N_G(x_k)|$ over all choices $x_1, \dots, x_k \in V$.\footnote{Here, $N_G(v) = \{ u\in V \mid \{u,v\} \in E\}$.} This objective serves as a clean proxy for more complex \emph{network diffusion} models that are widely studied in practice (see, e.g.,~\cite{KempeKT03, ChenWY09} and references therein). Equivalently, this problem is also known as \emph{Partial $k$-Dominating Set}: Given a graph $G$, find the largest number $t$ of nodes that can be \emph{dominated} by some nodes $x_1, \dots, x_k$, where a node is dominated if it is adjacent to or identical with one of the selected nodes~$x_1, \dots, x_k$.

As a natural optimization problem that generalizes the classical $k$-Dominating Set problem, the complexity of Partial $k$-Dominating Set is well understood: It is $W[2]$-hard~\cite{DowneyF95} (parameterized by~$k$) and thus does not admit $f(k)n^{O(1)}$-time algorithms unless~\makebox{$\mathrm{W}[2] = \mathrm{FPT}$}. From a fine-grained viewpoint, it cannot even be solved in time $O(n^{k-\epsilon})$ (for any $\epsilon > 0$ and $k\ge 3$), assuming the Strong Exponential Time Hypothesis (SETH), due to a reduction by Pătraşcu and Williams~\cite{PatrascuW10}. On the other hand, Eisenbrand and Grandoni's algorithm for $k$-Dominating Set~\cite{EisenbrandG04} extends to solve Partial $k$-Dominating Set in time $n^{k+o(1)}$ for all $k\ge 8$.\footnote{If the matrix multiplication exponent $\omega$ equals $2$ then the algorithm has running time $n^{k + o(1)}$ even for all $k\ge 2$; see Appendix~\ref{sec:baseline}.} Thus, for sufficiently large $k$ the Partial $k$-Dominating Set problem has complexity precisely $n^{k\pm o(1)}$.

At first glance, this state of affairs offers little hope for improvement. However, this impression may be misleading. A core principle in parameterized algorithm design is to look beyond input size and instead develop algorithms whose running time depends on more refined structural parameters. In our setting, the optimal value~$t$ -- the number of nodes that can be dominated by $k$ choices -- stands out as a natural parameter: It is always bounded by~$n$, yet likely smaller in practice. This gap potentially opens the door to faster algorithms, possibly in time $t^{k\pm O(1)}$ which are not ruled out by conditional lower bounds. Many real-world graphs have sublinear maximum degree~$\Delta$ -- e.g., many models for social networks yield $\Delta \leq O(\sqrt{n})$. In such cases we have $t \le k\Delta = O(k\sqrt{n})$, hence such an algorithm would reduce the time complexity to roughly the square root of exhaustive-search time. This leads to our first guiding question:

\begin{center}
	\textbf{Question 1:} Can we obtain a $t^{k\pm O(1)}$-time algorithm for Partial $k$-Dominating Set? \\ If not, what is the best running time that we can achieve?  
\end{center}

Note that research on exponential-time algorithms for Partial $k$-Dominating Set~\cite{Blaser03, KneisMR07, NederlofR10, KoutisW16} culminates in a $2^t n^{O(1)}$-time algorithm~\cite{KoutisW16}. Unfortunately, already for $t=\omega(\log n)$ such algorithms are superpolynomial, so they cannot be used to answer Question 1 positively.

\medskip
More generally, Partial $k$-Dominating Set is a special case of the classic combinatorial optimization problem \emph{Maximum $k$-Coverage} (or \emph{Max $k$-Cover}): Given a family of sets $\mathcal{F} = \{S_1,\dots, S_n\}$ over the universe $[u]\coloneqq \{1,\dots, u\}$, compute the maximum number of items that can be covered using $k$ sets, i.e., $\max_{i_1, \dots, i_k} |S_{i_1} \cup \cdots \cup S_{i_k}|$.\footnote{Specifically, we obtain Partial $k$-Dominating Set by setting $\mathcal{F} = \{ N_G(v)\cup \{v\} \mid v\in V\}$.} The computational complexity of Max $k$-Cover has received even more interest than Partial $k$-Dominating Set.
Beyond the hardness results for exact algorithms that can be derived from Partial $k$-Dominating Set, strong inapproximability results for Max $k$-Cover are known: Even approximating the problem better than the factor of $(1-1/e)$ achieved by its classic greedy algorithm is NP-hard~\cite{Feige98}. This result could recently be strengthened~\cite{Manurangsi20} to rule out even  $n^{o(k)}$-time algorithms for better-than-greedy approximations, assuming gap-ETH. 

Analogously to Question 1, our goal is to determine improvements over exhaustive-search running time $n^{k\pm O(1)}$ for Max $k$-Cover. Here, two parameters particularly lend themselves to an investigation: the maximum size $s$ of any input set, as well as the maximum frequency $f$ of any element in the universe. These parameters are well-studied in how they determine the approximability of the related \emph{Set Cover} problem: (1) The greedy algorithm computes a $(1+ \ln s)$-approximation on the size of the smallest set cover~\cite{Johnson74a,Lovasz75,Chvatal79}; obtaining an approximation factor of $\ln s - O(\ln \ln s)$ is NP-hard~\cite{Feige98, Trevisan01}. (2) It is possible to approximate the minimum set cover size up to a factor of~$f$ (see, e.g.,~\cite{Halperin02}), but any $(f-1-\epsilon)$-approximation is NP-hard~\cite{DinurGKR05}, and in fact, even an $(f-\epsilon)$-approximation assuming the Unique Games Conjecture~\cite{KhotR08}. For the setting of computing Max $k$-Cover, our second and technically even more ambitious question is as follows:
\begin{center}
	\textbf{Question 2:} What is the optimal running time for Max $k$-Cover in terms of $n$, $u$, $s$ and $f$?
\end{center}



\marvin{give further references to Partial Dominating Set works or Maximum $k$ Coverage works?} \nick{Always a good idea in my opinion.}

\marvin{where to cite~\cite{BonnetPS16}?}

\marvin{cite Jain et al. SODA'23 paper somewhere?}

%
%
%
%
%
%

\subsection{Our Results}

\paragraph{\boldmath Question 1: Complexity of Partial $k$-Dominating Set}
Our first main result is to settle the fine-grained complexity of Partial $k$-Dominating Set in terms of the number of vertices $n$ and the optimal value $t$, thereby successfully answering Question 1. In fact, we show that while a running time of $t^{k \pm O(1)}$ \emph{cannot} be achieved (assuming that at least one of two established hardness assumptions hold), we can nevertheless obtain a \emph{conditionally tight} algorithm improving over exhaustive-search time in many cases. 

\begin{theorem}[Fine-grained Complexity of Partial $k$-Dominating Set, informal version] \label{thm:main1}
    Assuming the clique and 3-uniform hyperclique hypotheses, the optimal running time for Partial $k$-Dominating Set is \[\min\{ t^{\frac{3}{2}k} + \min\{t^2,n\}^{\frac{\omega}{3}k}, n^k\} \qquad \text{up to FPT factors of the form $f(k)n^{O(1)}$.}\]  
\end{theorem}

Here, $2 \leq \omega \leq 2.372$~\cite{AlmanDWXXZ25} is the exponent of matrix multiplication. The conditional optimality in the above theorem is based on two plausible and well-established hypotheses on the complexity of detecting cliques of size $k$ in graphs and hypergraphs, respectively. Both have been used to give a number of tight conditional lower bounds, see, e.g.,~\cite{AbboudBVW18, BringmannW17, Chang19, KunnemannM20, DalirrooyfardVW22} for applications of the clique hypothesis and, e.g., \cite{AbboudBDN18,LincolnVWW18,BringmannFK19,KunnemannM20,DalirrooyfardVW22, Kunnemann22, DalirrooyfardJW22} for applications of the 3-uniform hyperclique hypothesis. See Section~\ref{sec:prelims} for details.

\marvin{the running time curve depending on $\Delta$ would currently consist of 4 parts? $nt$, $t^{2\omega/3 * k}$, $n^{2\omega/3 k} + t^{3/2 k}$, $n^k$? add the corresponding picture for current value of $\omega$?}

As $t$ and the maximum degree $\Delta$ are tightly related via $t/k \le \Delta \le t$, we obtain the same running time bounds when replacing $t$ by $\Delta$.
In particular, in the aforementioned realistic instances with maximum degree~\makebox{$\Delta=\Theta(\sqrt{n})$}, the resulting running time curiously depends on whether $\omega \le 2.25$: If $\omega\ge 2.25$, we obtain a running time of \smash{$f(k)n^{\frac{\omega}{3}k+O(1)}$} which is optimal assuming the $k$-clique hypothesis. If $\omega \le 2.25$ we obtain a running time of \smash{$f(k)n^{\frac{3}{4}k + O(1)}$} which is optimal under the 3-uniform hyperclique hypothesis. If $\omega=2.25$, the resulting running time of \smash{$f(k)n^{\frac{3}{4}k + O(1)}$} would be optimal under both hypotheses.

Furthermore, our results determine that one can solve the problem in linear time up to a threshold $\Delta \leq O(n^{\min\{\frac{2}{3},\frac{3}{2\omega}\}\frac{1}{k} - o(\frac{1}{k})})$, and one can beat exhaustive-search time $n^{k\pm o(1)}$ whenever $\Delta \le O(n^{2/3-\Omega(1)})$; conversely, if $\Delta \ge n^{2/3-o(1)}$ then exhaustive-search running time is necessary (assuming the 3-uniform hyperclique hypothesis).

Interestingly, as detailed in our technical overview below, conditional lower bounds guided our search towards our algorithm. Our main technical ingredient is a new algorithmic approach via so-called \emph{arity-reducing hypercuts} -- a win-win argument that either allows us to reduce to a Maximum-weight Triangle instance, or to identify vertices of the optimal solution at small cost. But also the conditional lower bounds are new and interesting. A conceptual challenge is to construct certain ``regular'' instances (as in~\cite{FischerKRS25}) which we manage to overcome in a unified way for both the clique-based and hyperclique-based lower bounds (and also the OV-based ones for the upcoming Theorem~\ref{thm:main3}) by starting from a carefully chosen intermediate problem.

\paragraph{\boldmath Question 2: Complexity of Max $k$-Cover}
Next, with additional technical effort, we extend our algorithms and conditional lower bounds to also resolve Question 2 -- settling the fine-grained complexity of Max $k$-Cover in terms of the number of sets $n$, the maximum set size $s$, the universe size $u$, and the maximum frequency $f$ of any element. Specifically, we obtain the following result.

\begin{theorem}[Fine-grained Complexity of Max $k$-Cover, informal version]\label{thm:main2}
	Assuming the clique and 3-uniform hyperclique hypotheses, the optimal running time for solving Max $k$-Cover for a set family $\mathcal{F} = \{S_1,\dots, S_n\}$ over the universe $[u]$ with maximum size $s\coloneqq \max_i |S_i|$ and maximum frequency $f \coloneqq \max_{y\in [u]} |\{ i \mid y\in S_i\}|$ is linear in the input size plus 
	\[ \min \left\{ (f\cdot \min\{u^{1/3}, \sqrt{s}\})^k + \min\{n,f\cdot \min\{\sqrt{u}, s\}\}^{k\omega/3}, n^k\right\} \quad \text{up to $g(k)\cdot (sf)^{O(1)}$ factors.} \]
\end{theorem}
From Theorem~\ref{thm:main2}, we can easily read off conditionally tight running time bounds for essentially any parameter setting. We illustrate some interesting parameter regimes in the following (omitting terms that are linear in the input size):
\marvin{Should we drop the middle items?} \nick{I think the second one is quite cool to have, we can drop the third one if someone prefers.}
\begin{itemize}
	\item If all set sizes are subpolynomial, i.e., $s=n^{o(1)}$, we obtain a tight running time of $g(k) f^{k\pm c}$ for some constant $c$ independent of $k$. Thus, in this case the running time is essentially determined by the maximum frequency $f$ alone.
	\item If $f \approx n^{1/3}$, $u \approx \sqrt{n}$ and $s$ may be arbitrary, we obtain a tight running time of $g(k) n^{k/2\pm c}$ for some $c$ independent of $k$; this bound is essentially the square root of the exhaustive-search baseline.
	\item We can precisely characterize when exhaustive-search running time $g(k)n^{k\pm c}$ is conditionally optimal: specifically, whenever $u\ge (n/f)^{3-o(1)}$ and $s\ge (n/f)^{2-o(1)}$. Conversely, whenever $u\le O((n/f)^{3-\epsilon})$ or $s\le O((n/f)^{2-\epsilon})$, we can beat exhaustive-search running time by a factor $g(k) n^{\delta k}$ with $\delta>0$.
\end{itemize}
\marvin{add partial vertex cover in $O(1)$-uniform hypergraphs? would get a result in $|E|$ and $\Delta$.}
 With this classification, we can immediately read off interesting results for special cases that have been studied for their own sake, such as the partial $k$-vertex cover problem in hypergraphs: Solving this problem on $n$-vertex $m$-edge $r$-uniform hypergraphs corresponds to Max $k$-Cover with $u=m$ and $f=r$. 

Extending our results for Partial $k$-Dominating Set to obtain Theorem~\ref{thm:main2} is far from straightforward: In particular, for small universe sizes, we design an algorithm that is based on a regularization step that allows us to find the optimal solution even more efficiently than using the previous arguments. 

\paragraph{\boldmath Bonus Question: Partial $k$-Dominating Set in Sparse Graphs}
As a further use of the techniques developed for our main questions, we essentially settle the complexity of Partial $k$-Dominating set in \emph{sparse} graphs (where the number of edges $m$ can be seen as yet another natural parameter), by proving the following results:

\begin{theorem}[Partial $k$-Dominating Set in sparse graphs]\label{thm:main3}
    For all $k\ge 3$, the optimal time complexity for Partial $k$-Dominating Set is $g(k)\left(mn^{k-2\pm o(1)} + m^{3k/5\pm O(1)}\right)$, assuming the 3-uniform hyperclique hypothesis.

    For the remaining case $k=2$, we obtain:
    \begin{itemize}[itemsep=0pt,topsep=\smallskipamount]
        \item an $O(m^{\frac{2\omega}{\omega+1}})$-time algorithm for Partial 2-Dominating Set, and
        \item a lower bound of $m^{\frac{4}{3}-o(1)}$ assuming the OV Hypothesis, matching the upper bound if $\omega=2$.
    \end{itemize}
\end{theorem}

We remark that this reveals Partial 2-Dominating Set as a curious counterpart to All-Edges Sparse Triangle, see~\cite{VWilliamsX20}: Both problems can be solved in time $O(m^{2\omega/(\omega+1)})$, while matching hardness results exists if $\omega=2$. For All-Edge Triangle Counting, such lower bounds are known based on the 3SUM~\cite{Patrascu10} and APSP~\cite{VWilliamsX20} hypotheses, while for Partial 2-Dominating Set, we establish an OV-based lower bound. The results of Theorem~\ref{thm:main3}, particular the algorithms, crucially rely on technical contributions of the arity-reducing hypercuts and regularization method.
\marvin{added discussion. does that make sense?} \nick{I like it a lot! (Moved it before the discussion on clique versus hyperclique; feel free to change.)}

\paragraph{Beyond our Main Questions}
A major research question in fine-grained complexity is to understand the relationship between the clique hypothesis and the 3-uniform hyperclique hypothesis. Our results reveal that Max $k$-Cover  suffers from conditional lower bounds from either hypothesis; notably, these turn out to be \emph{incomparable} under current bounds on $\omega$. The fact that we can nevertheless obtain a matching upper bound via our arity-reducing hypercuts (even without knowing whether $\omega < 2.25$) is encouraging: If a conceptually similar approach could be made to work for Klee's measure problem --  which for general $d\ge 4$ also suffers from incomparable lower bounds from the $k$-clique hypothesis~\cite{Chan10} and the 3-uniform hyperclique hypothesis~\cite{Kunnemann22, GorbachevK23},  one would break a long-standing time barrier in computational geometry~\cite{OvermarsY91,Chan10, Chan13}.

\subsection{Detailed Results and Technical Overview}\label{sec:technical-overview}

To obtain our results, we apply the paradigm of fine-grained complexity and algorithm design: we find increasingly higher conditional lower bounds and use the resulting insights to obtain faster algorithms, culminating in a conditionally optimal algorithm. We first start with the smallest non-trivial case of $k=2$.

In the following overview, we use that $\Delta \le t \le k\Delta$, where $\Delta$ denotes the maximum degree in the given graph. That is, in asymptotic bounds, we can use $\Delta$ and $t$ interchangeably. Thus, from now on, we will usually only consider the maximum degree $\Delta$ rather than $t$. Note that the worst-case input size is $\Theta(n\Delta)$.

We start with a useful proposition which follows from a simple exchange argument. The proof is deferred to Section~\ref{sec:algorithms}, in which the Proposition is proven in the more general formulation of Lemma~\ref{lemma:bounding-number-of-sets}.
\begin{proposition}\label{prop:exchange-pds}
    Let $H$ denote the set of the $\min\{k\Delta^2,n\}$ highest-degree vertices in $G$ (breaking ties arbitrarily). There exists an optimal solution consisting only of vertices in $H$, i.e.,
    \[\max_{v_1,\dots, v_k \in V} | N(v_1) \cup \dots \cup N(v_k)| = \max_{v_1,\dots, v_k \in H} | N(v_1) \cup \dots \cup N(v_k)|.\]
\end{proposition}

Note that the above proposition \emph{does not} imply that we may simply drop all vertices $V\setminus H$ from $G$. In general, these nodes are crucially involved in the objective value of any $x_1,\dots, x_k$.

\subsubsection{OV-optimal algorithm for $k=2$}
We start with the following simple algorithm for $k=2$ that beats the $O(n^\omega)$ time baseline by Eisenbrand and Grandoni~\cite{EisenbrandG04} whenever $\Delta = O(n^{\omega/4-\epsilon})$.

\begin{theorem}\label{thm:k2-algo}
	We can solve Partial $2$-Dominating Set in time $O(n\Delta + \Delta^4)$.
\end{theorem}
\begin{proof}
	By reading in the input graph $G$, we can compute $\deg_G(v)$ for all $v\in V$ and thus the set $H$ of the $\min\{2\Delta^2,n\}$ highest-degree vertices in time $O(n\Delta)$. We initialize a table $T[x,x'] = \deg_G(x)+\deg_G(x')$ for each $x,x'\in H$. For each $x\in H$, we traverse its neighbors $y\in N_G(x)$, and decrement $T[x,x']$ by 1 for every $x'\in H\cap N_G(y)$. Note that this terminates with the values $T[x,x']=|N_G(x) \cup N_G(x')|$ after $O(\Delta^4)$ steps, as for every $x\in H$, there are at most $\Delta$ choices for $y$, and for every $y$, there are at most $\Delta$ choices for $x'$. It remains to return $\max_{x,x'\in X} T[x,x']$, which concludes the $O(n\Delta + \Delta^4)$-time algorithm.
\end{proof}

It turns out that this simple algorithm is conditionally optimal. Specifically, we prove a matching lower bound based on the Orthogonal Vectors  (OV) Hypothesis\footnote{In the Orthogonal Vectors (OV) problem, we are given sets $A_1,A_2\subseteq \{0,1\}^d$ and the task is to determine whether there is an orthogonal pair $a_1\in A_1,a_2\in A_2$. A version of the OV Hypothesis states that this problem requires time $n^{2-o(1)}$ even when $d=n^{o(1)}$. See Section~\ref{sec:prelims} for details.}: The aim is to prove a lower bound of $\Delta^{4-o(1)}$ for instances of Partial Domination with $n$ nodes and maximum degree $\Delta = O(n^{\gamma})$ for $0 < \gamma \le 1/2$. Here, we sketch a simplification of our more general reduction given in Section~\ref{section:kh-OV-LB} when applied to the case $k=2$.  

To this end, we consider an OV instance $A_1,A_2\subseteq \{0,1\}^d$ with $|A_1|=|A_2|=s^2$ and $d= s^{o(1)}$. Without loss of generality, we may assume that all vectors $x\in A_1\cup A_2$ have the same number of ones, i.e., $\norm{a}_1 = \norm{a'}_1$ for all $a\in A_1, a'\in A_2$.\footnote{For the proof of a more general statement, see Lemma~\ref{lemma:khov-regularity}.} We divide each $A_i$ into the groups $A_i^{(1)},\dots, A_i^{(s)}$ of size $s$ each. 

The core of the construction is a graph $G$ with vertex set $A_1 \cup A_2 \cup Z$ where $Z$ is a set of $d s^2$ auxiliary vertices denoted as $(y,g_1, g_2)$ with $y\in [d], g_1,g_2\in [s]$. Any vertex $a\in A_j^{(g)}$ will be adjacent to $(y,g_1,g_2)\in Z$ if and only if $g_j = g$ and $a[y] = 1$.  It is straightforward to see that for any $a_1\in A_1, a_2 \in A_2$, we have 
\begin{equation} |N(a_1) \cup N(a_2)| = \norm{a_1}_1s + \norm{a_2}_1s - a_1 \cdot a_2, \label{eq:k2-main}
\end{equation}
where $a_1 \cdot a_2 = \sum_{y=1}^{d} a_1[y]\cdot a_2[y]$  denotes the inner product of $a_1$ and $a_2$. Exploiting that $C\coloneqq \norm{a_1}_1 = \norm{a_2}_1$ is independent of the choice of $a_1,a_2$, we conclude that there exists $a_1 \in A_1, a_2\in A_2$ dominating at least $2Cs$ vertices if and only if $A_1,A_2$ contain an orthogonal pair of vectors. This construction immediately yields a  Max $k$-Cover instance $\mathcal{F} = \bigcup_{x\in A_1,A_2} \{ N(x) \}$ over universe~$Z$ that is equivalent to the original OV instance. Note that any node in $G$ has maximum degree $O(ds)$: any node $a_1\in A_1$ is adjacent only to nodes $(y,g,g')\in Z$ where $g$ is such that $a_1 \in A_1^{(g)}$ and $g'\in [s]$ is arbitrary. Symmetrically, any node $a_2\in A_2$ has degree at most $O(ds)$. Finally each node $(y,g_1,g_2)$ is adjacent only to nodes $a_1\in A^{(g_1)}_1$ and $a_2\in A^{(g_2)}_2$.

This can be turned into a desired Partial $k$-Dominating Set instance by adding a gadget of at most $O(ds^2)$ additional nodes (and adjacent edges) that enforce that any optimal solution $u,v\in V$ must be of the form $u\in A_1,v\in A_2$ (or vice versa), while keeping the maximum degree $O(ds)$; for details see Section~\ref{section:kh-OV-LB}.\marvin{add idea in footnote?} Let $n'$ be the total number of nodes in this instance, then for any value $n = \Omega(ds^2)$, we may add $n-n'$ isolated nodes to produce an equivalent Partial $k$-Dominating set instance $G'$ with $n$ nodes and maximum degree $O(ds)=s^{1+o(1)}$ . Any algorithm solving Partial 2-Dominating Set in time $O(\Delta^{4-\epsilon})$ would thus solve OV in time $O(s^{4-\epsilon +o(1)})$, refuting the OV Hypothesis.

Formally, we obtain the following lower bound, proven in the more general Theorem~\ref{theorem-partial-dom-LB}.
\begin{theorem}\label{thm:OVLB}
	Let $\epsilon > 0$ and $0 < \gamma \le 1/2$. If Partial $2$-Dominating Set with $\Delta = O(n^{\gamma})$ can be solved in time $O(\Delta^{4-\epsilon})$, then the OV Hypothesis is false.
\end{theorem}

This shows optimality of Theorem~\ref{thm:k2-algo} for all values $\Delta = O(n^\gamma)$ whenever $\gamma \le 1/2$.\footnote{Note that the $O(n\Delta)$ is the input size, so the additional $O(n\Delta)$ term is always necessary.} Thus, if $\omega= 2$ and the OV hypothesis holds, then the time complexity of Partial 2-Domination is $\min\{n\Delta+\Delta^4,n^2\}^{1\pm o(1)}$.

Interestingly, this conditionally rules out an extension of the $m^{\omega/2+o(1)}$ algorithm for 2-Dominating Set given in~\cite{FischerKR24} to Partial $2$-Dominating Set. To see this, note that for $\Delta = O(\sqrt{n})$, an $m^{\omega/2+o(1)}$ algorithm would solve the problem in strongly subquadratic time $O(n^{3\omega/4}) = O(n^{1.78})$, which would refute the OV Hypothesis. This separates the fine-grained complexities of $k$-Dominating Set and Partial $k$-Dominating Set in sparse graphs. 

\subsubsection{Conditional lower bounds for $k\ge 3$}

Theorem~\ref{thm:k2-algo} generalizes in a straightforward way to achieve a $O(n\Delta + f(k)\Delta^{2k+O(1)})$-time algorithm. The generalization of the OV-based lower bound is less obvious. In Section~\ref{section:kh-OV-LB}, we will obtain a $\Delta^{k+ 1 + \frac{1}{k-1}-o(1)}$-time lower bound based on $k$-OV using a more general reduction sketched below. 
 The target of $\Delta^{k+O(1)}$-time is not ruled out by this reduction. Is it possible to reduce the baseline exponent of $2k+O(1)$ to obtain an exponent of $k+1+\frac{1}{k-1}$?


\paragraph{(Hyper)Clique barrier}
Interestingly, it turns out that this is not possible without breaking the $k$-clique hypothesis:

\begin{theorem}\label{thm:clique-lb}
	Let $\epsilon > 0$, $0 < \gamma \le \frac{1}{2}$ and $f$ be a computable function. If we can solve Partial $k$-Dominating Set with $\Delta = \Theta(n^\gamma)$ in time $f(k)\Delta^{(\frac{2\omega}{3}-\epsilon)k}$ for all sufficiently large $k$, then the $k$-clique hypothesis is false. 
\end{theorem}

This result gives a negative answer to Question 1 assuming the $k$-clique hypothesis.

Furthermore, we obtain an incomparable conditional lower bound based on a different hypothesis, the 3-uniform hyperclique hypothesis. 




\begin{theorem}\label{thm:hyperclique-lb}
	Let $k\ge 3, \epsilon > 0, 0 < \gamma \le \frac{2}{3}$ and $f$ be a computable function. If we can solve Partial $k$-Dominating Set with $\Delta = \Theta(n^\gamma)$ in time $f(k)\Delta^{(\frac{3}{2}-\epsilon)k}$, then the 3-uniform $k$-clique hypothesis is false. 
\end{theorem}

Note that this lower bound gives another negative answer to Question 1 assuming a plausible hardness hypothesis.

\paragraph{Reduction via $(k,h)$-maxIP/minIP}

We prove all of the above conditional lower bounds conveniently using the problems $(k,h)$-maxIP/minIP as an intermediate step. These problems are the natural optimization versions of the $(k,h)$-OV problem, which was originally introduced in~\cite{BringmannCFK22} to obtain constant-factor inapproximability results\footnote{There, it was used only for the case $h=3$.}. In our setting, $(k,h)$-maxIP/minIP can be used to give incomparable conditional lower bounds under different hypotheses ($k$-clique, 3-uniform hyperclique, and $k$-OV) via a single main reduction. Intuitively, they are a restriction of the $k$-Maximum Inner Product ($k$-maxIP) and $k$-Minimum Inner Product ($k$-minIP) problem (see, e.g.,~\cite{SLM19,BringmannCFK22}) such that in each dimension, only the vectors in at most $h$ sets are \emph{active}, i.e., may be different from $1$.

Formally, in the $(k,h)$-maxIP/minIP problem, we are given a $k$-OV instance, i.e., $k$ sets $A_1,\dots, A_k\subseteq\{0,1\}^d $ of $n$ vectors each, with the following additional promise: In each dimension $y\in [d]$, there are $h$ associated indices $i_1,\dots,i_h$ (called \emph{active} indices); we may assume that $a_i[y] = 1$ for all $i\in [k]\setminus\{i_1,\dots, i_h\}$. For $k$-maxIP ($k$-minIP), the task is to maximize (minimize), over all $a_1\in A_1,\dots, a_k\in A_k$, the number of dimensions $y\in [d]$ such that $a_{i_1}[y] \cdots a_{i_h}[y] = 1$, where $i_1,\dots, i_h$ are the active indices for $y$.

Using standard fine-grained reductions~\cite{AbboudBDN18,KarppaKK18}, one can establish that: (1)  $(k,2)$-maxIP/minIP require time $n^{\omega/3 k-o(1)}$-time assuming the $k$-clique hypothesis, (2) $(k,3)$-maxIP/minIP require time $n^{k-o(1)}$ under the $3$-uniform hyperclique hypothesis, and (3) $(k,k)$-maxIP/minIP require time $n^{k-o(1)}$ assuming the $k$-OV hypothesis, see Lemma~\ref{lemma:hyperclique-hardness-of-khov}. 

We use a reduction approach that on a high-level looks similar to our OV-lower bound for the case $k=2$. Here, we give a sketch of the core construction (with several details deferred to Section~\ref{section:kh-OV-LB}). Specifically, let $A_1,\dots, A_k\subseteq \{0,1\}^d$ be a given $(k,h)$-maxIP/minIP instance with $|A_i| = N$; for ease of presentation, assume that $d=n^{o(1)}$. Similar to before, we group each $A_i$ into groups $A_i^{(1)}, \dots, A_i^{(s)}$; here we choose $s= N^{1/h}$. We construct a graph $G'$ that includes the vertex sets $A_1,\dots, A_k$. For any choice of active indices $1\le i_1 < i_2 < \cdots < i_h \le k$, we let $D_{i_1,\dots i_h}$ denote the dimensions $y$ such that $i_1,\dots, i_h$ are the active indices. For each such choice, we introduce a set of additional vertices labelled $(y,g_1,\dots, g_h)$ with $y\in D_{i_1,\dots, i_h}$ and $g_1,\dots, g_h\in [s]$. Crucially, we connect any edge $a_i \in A_i^{(g)}$ and $(y,g_1,\dots, g_h)\in Y_{i_1,\dots, i_h}$ by an edge if and only if there is some $i_a$ with $i_a = i$, $g_a = g$ and $a_i[y] = 1$.\marvin{notation deviates from proof. Harmonize} 

Consider, for any choice $a_1\in A_1,\dots, a_k\in A_k$, its objective value. Specifically, by the inclusion-exclusion principle, we obtain
\begin{equation}
 |N(a_1) \cup \cdots \cup N(a_k)| = \sum_{r=1}^k (-1)^{r+1} \sum_{1\le i_1< \cdots < i_r \le k} |N(a_{i_1}) \cap \cdots \cap N(a_{i_r})|. \label{eq:inc-exc-main}
\end{equation}
Since each node $(y,g_1,\dots, g_h)$ has only neighbors in $h$ sets (specifically, $A_{i_1},\dots, A_{i_h}$, where $i_1,\dots, i_h$ are the active indices of $y$), any term in~\eqref{eq:inc-exc-main} with $r> h$ vanishes. Furthermore, the term for $r=h$ counts precisely the number of coordinates $y$ in which the active indices of $a_1,\dots, a_k$ are all equal to 1, with a multiplicative factor of $(-1)$ if $h$ is even. Thus, if we can make the contribution of all terms with $r<h$ equal to a constant \emph{independent of $a_1,\dots, a_k$}, the optimal value is attained by $a_1\in A_1,\dots, a_k\in A_k$ that maximize (if $h$ is odd) or minimize (if $h$ is even) the number of coordinates $y$ in which all active indices are equal to 1. If we can do this, we can read off the $(k,h)$-maxIP or $(k,h)$-minIP value, respectively, from~\eqref{eq:inc-exc-main}.

Fortunately, this is possible: We show how to add new dimensions to the vectors $A_1,\dots, A_k$ such that (1) all $a_{i_1} \in A_{i_1},\dots, a_{i_r} \in A_{i_r}$ with $1\le i_1 < \cdots < i_r \le k$, $r < h$ have the same number of $y$ for which $i_1,\dots, i_r$ belong to the active indices and $a_{i_1}[y] = \cdots = a_{i_r}[y] = 1$, and (2) for all $a_1\in A_1,\dots,a_k\in A_k$ the number of $y$ for which the active indices $i_1,\dots, i_h$ satisfy $a_{i_1}[y] = \cdots = a_{i_h}[y] = 1$ remains unchanged. 

Note that each node $a_i\in A_i$ is connected to at most $hds^{h-1} = O(dN^{\frac{h-1}{h}}) \le N^{\frac{h-1}{h}+o(1)}$ dimension nodes $(y,g_1,\dots, g_h)$ and each dimension is connected to at most $hN/s = O(N^{\frac{h-1}{h}})$ nodes in $A_i$, $i\in [k]$. Thus, to obtain an instance with maximum degree $\Delta$, we may choose $N = \Delta^{\frac{h}{h-1}-o(1)}$, since then $dN^{\frac{h-1}{h}} \le \Delta$. Setting $h=k$, we obtain the claimed lower bound of $(\Delta^{\frac{k}{k-1}})^{k-o(1)} = \Delta^{k + 1 +\frac{1}{k-1} - o(1)}$ under the $k$-OV Hypothesis. Setting $h=3$, we obtain the claimed lower bound of $\Delta^{\frac{3}{2}(k-o(1))}=\Delta^{\frac{3}{2}k-o(1)}$ under the $3$-uniform hyperclique hypothesis. Finally, setting $h=2$, we obtain the lower bound of $\Delta^{2(\frac{\omega}{3}k-o(1))} = \Delta^{\frac{2\omega}{3}k-o(1)}$ under the $k$-clique hypothesis.   






\subsubsection{A Matching Algorithm via Arity-Reducing Hypercuts}

Our perhaps most interesting technical contribution is an algorithm given in Section~\ref{sec:algorithms} that matches the conditional lower bounds given by Theorems~\ref{thm:clique-lb} and~\ref{thm:hyperclique-lb}. In fact, we exploit combinatorial insights gained by inspecting why we could not strengthen the above conditional lower bounds.

To formalize our main approach, let $G$ be an instance of Partial $k$-Dominating Set. Recall that $H$ denotes the set of $\min\{k\Delta^2, n\}$ highest-degree vertices, which must contain an optimal solution by Proposition~\ref{prop:exchange-pds}. We define a corresponding hypergraph $\hyperH$ on the vertex set $V(\hyperH) = H$ by $\{v_1, ..., v_h\} \in E(\hyperH)$ if and only if $v_1, \dots, v_h$ share at least one common neighbor in $G$, i.e., $N_G(v_1) \cap \cdots \cap N_G(v_h) \ne \emptyset$. 

A first basic ingredient are well-known subcubic algorithms for Maximum Weight-Triangle with small edge weights~\cite{Yuval76,Zwick02}, see also~\cite{VassilevskaW06}: In the formulation that we need, we are given a graph $G'=(V', E')$ together with vertex weights $w(v)\in \{-M, \dots, M\}$ for $v\in V$ and edge weights $w(u,v)\in \{-M, \dots, M\}$ for $u,v\in V$ (such that $w(u,v) = w(v,u)$), the task is to determine the maximum weight of a triangle in $G'$, i.e., 
\[ \max_{x,y,z\in V'} w(x)+w(y)+w(z) + w(x,y)+w(x,z)+w(y,z). \]
This problem can be solved in time $O(Mn^\omega)$, and extends to computing the analogously defined Maximum-weight $k$-clique problem in time $O(Mn^{\omega \lceil \frac{k}{3} \rceil})$.

The above tool immediately gives a fast algorithm for the restricted case that $\hyperH$ contains no edges of arity at least 3, i.e., when $\hyperH$ is a graph: Specifically, we can construct a complete graph $G'$ with $V(G') = X$ and node weights $w(v) = \deg_G(v)$ and edge weights $w(v,v') = -|N_G(v)\cap N_G(v')|$. Note that if $\hyperH$ contains no edges of arity at least 3, then the weight of any clique $v_1,\dots, v_k$ in $G$ is precisely $|N(v_1)\cup \cdots\cup N(v_k)|$. Using the maximum-weight $k$-clique algorithms, we can thus solve Partial $k$-Domination in this restricted case in time $O(|X|^{\omega \lceil \frac{k}{3}\rceil +1})$.

\paragraph{Handling higher arities.}
The above argument fails quite fundamentally when $\hyperH$ contains hyperedges of arity at least 3. Unfortunately, higher-arity dependencies generally cannot be avoided -- optimal solutions in difficult instances may have dependencies of arbitrary arity, as evidenced by the graphs produced in our reduction from $(k,h)$-maxIP/minIP.

We overcome this challenge via a technique that we call \emph{arity-reducing hypercuts}: The idea is investigate, for a fixed optimal solution $S$, the existence of a cut in the subhypergraph $\hyperH_S\coloneqq \hyperH[S]$ induced by $S$ such that no edge connects vertices that are all taken from different sets. The formal definition is as follows. 
\begin{boxed-definition}
	We say that a partition $S_1,S_2,S_3$ of $S$ is an \emph{arity-reducing hypercut} in $\hyperH_S$ if and only if there is no hyperedge $\{v_1,v_2,v_3\}$ in $\hyperH_S$ with $v_i\in S_i$ for $i\in \{1,2,3\}$.\footnotemark  We say that $S_1,S_2,S_3$ is \emph{balanced} if $|S_1|$, $|S_2|$, $|S_3|$ all differ by at most 1. \nick{Is this the only time we use this environment in the paper? If so: Shouldn't we just have a normal definition? (Very weak opinion on my part)}
\end{boxed-definition}
\footnotetext{It would be reasonable to define arity-reducing hypercuts for any arity $r$: Here, we would say that $S_1,\dots, S_r$ is an arity-reducing hypercut if there is no hyperedge $\{v_1,\dots, v_r\}$ with $v_i \in S_i$. However, for our algorithm, we will only exploit $r=3$.}

We exploit this notion using the following win-win argument: We show that either (1) there exists a balanced arity-reducing hypercut $S_1,S_2,S_3$ of $\hyperH_S$, which enables us to find an optimal solution using an appropriate Maximum-Weight Triangle instance or (2) $S$ contains a structure that can guide us towards finding $S$ more efficiently than brute force.

\paragraph{(1) Exploiting an arity-reducing hypercut}
Interestingly, we will be able to exploit the mere \emph{existence} of a balanced arity-reducing hypercut, without the need to explicitly construct such a hypercut. 

For ease of presentation, we will assume that $k$ is divisible by 3 (we will give the full arguments in Section~\ref{sec:algorithms}). We construct a graph $G'$ with vertex set $V(G') = V' \coloneqq \binom{H}{k/3}$ and call the vertices in $V'$ \emph{super nodes}. For any super node $\overline{u} = (u_1,\dots, u_{k/3})$, we define its weight as  $w(\overline{u}) = |N_G(u_1) \cup \cdots N_G(u_{k/3})|$, and for any pair of super nodes $\overline{u} = (u_1,\dots, u_{k/3})$ and $\overline{v} = (v_1, \dots, v_{k/3})$, we define the weight of the edge between them as $w(\overline{u},\overline{v}) = - \left|\bigcup_{i=1}^{k/3}\bigcup_{j = 1}^{k/3} (N_G(u_i) \cap N_G(v_j))\right|$. Crucially, the weight of any triangle $(x_1,\dots, x_{k/3}), (y_1, \dots, y_{k/3}), (z_1,\dots, z_{k/3})$ in $G'$ is a \emph{lower bound} on the objective value for $S\coloneqq \{x_1,\dots, x_{k/3},y_1,\dots,y_{k/3},z_1,\dots,z_{k/3}\}$, i.e., the weight of this triangle is at most
\[|N_G(x_1) \cup \cdots \cup N_G(x_{k/3}) \cup N_G(y_1) \cup \cdots \cup N_G(y_{k/3}) \cup N_G(z_1) \cup \cdots \cup N_G(z_{k/3})|.\] 
Observe that the weight of the triangle can be strictly smaller than the objective value if $\hyperH_S$ contains hyperedges.
Conversely, consider a solution $S$ for which there exists a balanced arity-reducing hypercut $S_1,S_2,S_3$. Since $|S_i| = k/3$, there are super nodes $x,y,z$ in $G'$ corresponding to $S_1,S_2,S_3$. By definition of $S_1,S_2,S_3$, for all $s_1\in S_1,s_2\in S_2,s_3\in S_3$ we have $N_G(s_1)\cap N_G(s_2)\cap N_G(s_3) = \emptyset$, and thus the weight of the triangle $(x,y,z)$ in $G'$ is \emph{equal} to the objective value of $S$. For a detailed proof, we refer to Lemma~\ref{lemma:triangle-weight-bound-inclusion-exclusion}¸

Consequently, by solving a single Max-Weight Triangle instance on $O(|H|^{k/3})$ nodes with weights in $\{-k\Delta/3,\dots, k\Delta/3\}$, we can detect any optimal solution $S$ admitting a balanced arity-reducing hypercut\footnote{Formally speaking, we obtain a lower bound on the optimum that is at least as large as the objective value of any solution $S$ admitting a balanced arity-reducing hypercut.} in time $O(k|H|^{k\omega/3+1})$.

\paragraph{(2) Obstructions to arity-reducing hypercuts.}
The other side of our win-win argument is to detect an optimal solution $S$ admitting no balanced arity-reducing hypercut. To this end, we show that the only potential obstruction to getting such a hypercut is the existence of certain structures which we will call \emph{bundles}. Intuitively, if $\hyperH$ contains no large bundles, there exists a balanced arity-reducing hypercut, and the previous considerations apply. Otherwise, if $\hyperH$ contains large bundles, we can essentially exhaustively search over these large bundles to identify parts of the solution $S$ quickly.

Formally, our notion of \emph{bundle} is defined recursively: A \emph{0-bundle} is a set consisting of a single vertex $v\in V$. For any $c$-bundle $B$ and hyperedge $\{b, x, y\} \in E(\hyperH)$ with $b\in B$ and $x,y\notin B$, we say that the set $B'= B\cup \{x,y\}$ forms a \emph{$(c+1)$-bundle}. 

The following two observations provide the gist of our algorithmic approach: Let $S$ denote an optimal solution, then
\begin{enumerate}
	\item Any partition $S_1,S_2,S_3$ of $S$ such that any bundle $B\subseteq S$ is completely contained in a single set $S_i$ is an arity-reducing hypercut. Consequently, it is not difficult to see that if there are no $c$-bundles $B\subseteq  S$ with $c\ge c_0$ for some $c_0$, then there exists such a partition such that $|S_i| \le k/3 + O(c_0)$: starting with $S_1=S_2=S_3 = \emptyset$, repeatedly take a maximal bundle and place it into the currently smallest set $S_i$, until all vertices are distributed among $S_1,S_2,S_3$. This yields an \emph{almost-}balanced arity-reducing hypercut. 
	\item If there exists a $c$-bundle $B\subseteq S$ with $c\ge c_0$, then we can guess this bundle by enumerating $\Delta^{3c+2}$ bundles\footnote{This follows from a combination of Lemma~\ref{lemma:listing-bundles} and Proposition~\ref{lemma:bounding-number-of-sets}, detailed in Section~\ref{sec:algorithms}.} of size $c$. Thus intuitively, by spending an effort of $\Delta^{3c+2}$, we obtain $|B|=1+2c$ nodes of $S$. (Note that as $c$ increases, this approaches the ratio of $\Delta^{1.5}$ effort per vertex.) This gives rise to a recursive algorithm that for each of the $\Delta^{3c+2}$ bundles $B$, computes the optimal value in the subproblem obtained by restricting $S$ to contain $B$.
\end{enumerate}

\paragraph{Combining both arguments}
A direct implementation of the above argument would result in an algorithm with running time
\begin{align*} 
	& O\left(T_{\mathrm{MaxWeightTriangle}}(|H|^{\frac{k}{3} + O(c_0)}) + \Delta^{\frac{2+3c_0}{1+2c_0}k}\right)  = O\left(\Delta^{\frac{2\omega k}{3} + O(c_0)} + \Delta^{(\frac{3}{2} + \frac{1}{2+4c_0})k}\right). 
\end{align*}
where we used Proposition~\ref{prop:exchange-pds} to bound $|H| \le k\Delta^2$.
By choosing $c_0 = \gamma \sqrt{k}$, this would give an algorithm running in time
\[O\left(n\Delta + (\Delta^{\frac{2\omega}{3}k} + \Delta^{\frac{3}{2}k})\Delta^{O(\sqrt{k})}\right).\]

We refine the above arguments further to achieve a running time of $O(n\Delta + f(k)(\Delta^{\frac{3}{2}k} + \min\{n,\Delta^2\}^{\frac{\omega}{3}k})\Delta^{O(1)})$, reducing the $\Delta^{O(\sqrt{k})}$ overhead over the conditional lower bounds to $\Delta^{O(1)}$. This is achieved by a surprisingly succinct, non-recursive algorithm (see Algorithm~\ref{alg:partial-dom-alg}): It turns out that it suffices to guess two disjoint bundles $S_1, S_2$ of size $0\le |S_1|+|S_2| \le k$ -- after including these bundles into our solution and simplifying the graph accordingly, we show that there exists a balanced arity-reducing hypercut.
 
 \marvin{I guess the final algorithm has running time $O(n\Delta + k^2(\Delta^{\frac{3}{2}k} + \min\{n,\Delta^2\}^{\frac{\omega}{3}k})\Delta^{2})$. (Not quite: $k^{O(k)}$ terms are involved). write as $f(k)$ or include these factors?}

\subsubsection{Extension to Max $k$-Cover}

Let us return to the Max $k$-Cover problem: We can view it as a bichromatic version of Partial $k$-Dominating Set, in which we are given a graph $G=(X\cup Y, E)$ and the task is to maximize, over all $x_1,\dots, x_k\in X$, the number of $y\in Y$ that are adjacent to at least one $x_i$. In this view, $X=\mathcal{F}$, $Y=[u]$, and $E= \{\{S_i, y\} \mid y\in S_i\}$. Consequently, the parameters $n$, $u$, $s$ and $f$ correspond to $|X|$, $|Y|$, $\Delta_s\coloneqq \max_{x\in X} \deg(x)$ and $\Delta_f \coloneqq \max_{y\in Y} \deg(y)$, respectively.

Generally speaking, both our algorithm and conditional lower bounds can be adapted to this more general setting and analyzed in term of these four parameters. Curiously, however, a straightforward generalization does not yield matching upper and lower bounds, specifically for small universe sizes $u$. We are nevertheless able to determine the tight fine-grained complexity of Max $k$-Cover by devising two additional algorithmic improvements:

\paragraph{Small number of 1-bundles for small universes}
The first improvement is surprisingly simple: We may bound the number of 1-bundles by $\Delta_{f}^3 u$ (which follows by guessing a common neighbor of the three vertices in a 1-bundle, as well as a triple of its neighbors). Intuitively, integrating this observation into our algorithm enables us to improve over a term of $(\Delta_f \sqrt{\Delta_s})^k$ in our running time bound whenever $\sqrt[3]{u} \le \sqrt{\Delta_s}$.

\paragraph{Regularization step}
Consider a setting in which $\Delta_s \ge \sqrt{u}$. The natural generalization of our $(k,2)$-maxIP lower bound yields a conditional lower bound of only $(\Delta_f \min\{\Delta_s, \sqrt{u}\})^{\frac{\omega}{3}k-o(1)}$. However, even assuming that there exists a balanced arity-reducing hypercut, our algorithm produces a maximum-weight triangle instance on $(k\Delta_f \Delta_s)^{\frac{k}{3}}$ nodes, using that the exchange argument of Proposition~\ref{prop:exchange-pds} and Lemma~\ref{lemma:bounding-number-of-sets} reduces $X$ to a size of at most $k\Delta_f\Delta_s$. A natural attempt would be to reduce this size even further to $O(\Delta_f \sqrt{u})$ -- however, it appears impossible to improve the exchange argument sufficiently to achieve this.

Instead, we rely on the following more involved argument, proven in Lemma~\ref{lemma:regularization-lemma}: We observe that any optimal solution must contain a node of \emph{high} degree, specifically, degree at least $\Delta_s/k$. We now distinguish two cases: (1) If there are at most $k^2\Delta_f$ many high-degree nodes, we can afford to guess such a high-degree node (even applying this step repeatedly incurs a cost of at most $g(k)\Delta_f^k$ in total, which is dominated the term of $g'(k)(\Delta_f \min\{\sqrt[3]{u}, \sqrt{s}\})^k$ incurred by handling bundles). (2) Otherwise, we prove that \emph{all} nodes have \emph{moderately high} degree, specifically, degree at least $\Delta_s/(2k)$. Crucially, there can be at most $O(\Delta_f \sqrt{u})$ many such nodes: there are at most $u\Delta_f$ edges in $G$, so there can be at most $2k u \Delta_f / \Delta_s \le  f(k) \Delta_f \sqrt{u}$ many nodes of \emph{moderately high} degree.\footnote{A very observant reader might notice a potential issue in this argument: After guessing some solutions nodes according to case (1) and simplifying the graph, we might be left with a smaller value $1 \le \Delta \le \Delta_s$ for $\Delta_s$. However, together with the exchange argument of Lemma~\ref{lemma:bounding-number-of-sets}, we still obtain a bound of $2k \Delta_f \min\{\Delta,u/\Delta\} \le g(k) \Delta_f \sqrt{u}$.} This argument reduces the number of candidates for solution nodes sufficiently, and yields the final improvement to obtain a conditionally tight algorithm even for the general setting of Max $k$-Cover. We present all details in Section~\ref{sec:max-k-cover}.

\subsubsection{Further application: Influence of sparsity for Partial Dominating Set}

Finally, to obtain a conditionally optimal bound for Partial Dominating Set in terms of the number of vertices $n$ and number of edges $m$, our previous techniques turn out to be essential:
Our algorithm for Partial 2-Dominating Set exploits a careful combination of the baseline $O(n^\omega)$-time algorithm due to Eisenbrand and Grandoni, and the $O(m^{2\omega/(\omega+1)})$-time algorithm for Sparse Triangle Detection/Counting~\cite{AlonYZ97}. To always reduce to one of these cases, we employ a rather complex case distinction using the degrees of the solution nodes, which makes a subtle implicit use of the exchange argument of Proposition~\ref{prop:exchange-pds}. The corresponding conditional lower bound follows by a natural adaptation of Theorem~\ref{thm:OVLB}

For higher values of $k\ge3$, we again employ our arity-reducing hypercuts. Notably, in this setting, we obtain matching upper and conditional lower bounds already under current values of $\omega$ (not only if either $\omega \le 2.25$ or $\omega > 2.25$ can be proven) by further employing the Regularization Lemma sketched above.
For a detailed description and proofs, we refer to Section~\ref{sec:sparse}.